\documentclass[12pt,a4paper]{article}
\usepackage{graphicx} 
\usepackage{tikz-cd}
\usepackage{amsmath,amssymb,amsthm,geometry,bbm,mathtools,braket}

\usepackage[hidelinks]{hyperref}
\usepackage[
backend=biber,
style=alphabetic,
]{biblatex}
\usepackage{mdframed}
\usepackage[toc,page]{appendix}

\usepackage{caption}

\newtheorem{theo}{Theorem}[section]

\newtheorem{prop}[theo]{Proposition}

\newtheorem{lemma}[theo]{Lemma}

\theoremstyle{definition}
\newtheorem{definition}[theo]{Definition}

\newtheorem*{remarks}{Remarks}

\numberwithin{equation}{section}

\newcommand{\lune}{L_k}
\newcommand{\la}{\lambda_{k,p}^{-1}}
\newcommand{\lbulk}{L^{\textrm{Bulk}}_k}

\newcommand{\Gcutoff}{\mathcal{G}_\delta}

\newcommand{\DDD}{\mathbb{D}}
\newcommand{\HHH}{\mathbb{H}}
\newcommand{\NNN}{\mathbb{N}}
\newcommand{\RRR}{\mathbb{R}}
\newcommand{\SSS}{\mathbb{S}}

\newcommand{\TTT}{\mathbb{T}}
\newcommand{\XXX}{\mathbb{X}}
\newcommand{\ZZZ}{\mathbb{Z}}

\newcommand{\cE}{\mathcal{E}}
\newcommand{\cF}{\mathcal{F}}
\newcommand{\cH}{\mathcal{H}}
\newcommand{\cI}{\mathcal{I}}
\newcommand{\cN}{\mathcal{N}}
\newcommand{\cO}{\mathcal{O}}
\newcommand{\cS}{\mathcal{S}}

\newcommand{\fE}{\mathfrak{E}}
\newcommand{\fF}{\mathfrak{F}}
\newcommand{\fK}{\mathfrak{K}}

\newcommand{\B}{\textnormal{B}}

\newcommand{\di}{\textnormal{d}}
\newcommand{\diag}{\textnormal{diag}}
\newcommand{\eff}{\textnormal{eff}}
\newcommand{\F}{\textnormal{F}}
\newcommand{\FS}{\textnormal{FS}}
\newcommand{\GS}{\textnormal{GS}}
\newcommand{\HS}{\textnormal{HS}}
\newcommand{\lin}{\textnormal{lin}}
\newcommand{\nor}{\textnormal{nor}}
\newcommand{\op}{\textnormal{op}}
\newcommand{\RPA}{\textnormal{RPA}}
\newcommand{\tr}{\textnormal{Tr}}

\newcommand{\abs}[1]{\left\lvert #1 \right\rvert}
\newcommand{\norm}[1]{\left\lVert #1 \right\rVert}
\newcommand{\eva}[1]{\left\langle #1 \right\rangle}

\usepackage{authblk}

\usepackage[textsize=footnotesize,textwidth=2.5cm]{todonotes}

\title{The Ground State Energy of a Mean-Field Fermi Gas in Two Dimensions}

\author[1,*]{Gregorio Casadei}
\author[2,**]{Sascha Lill}
\affil[1]{e--mail: \href{mailto:gregorio.casadei@studenti.unimi.it}{gregorio.casadei@studenti.unimi.it}}
\affil[2]{ORCID: \href{https://orcid.org/0000-0002-9474-9914}{0000-0002-9474-9914}, e--mail: \href{mailto:sali@math.ku.dk}{sali@math.ku.dk}}
\affil[*]{Università degli Studi di Milano, Via Cesare Saldini 50, 20133 Milano, Italy}
\affil[**]{University of Copenhagen, Universitetsparken 5, DK-2100 Copenhagen, Denmark}

\begin{document}

\maketitle

\begin{abstract}
We rigorously establish a formula for the correlation energy of a two-dimensional Fermi gas in the mean-field regime for potentials whose Fourier transform $\hat{V}$ satisfies $\hat{V}(\cdot) | \cdot | \in \ell^1$. Further, we establish the analogous upper bound for $\hat{V}(\cdot)^2 | \cdot |^{1 + \varepsilon} \in \ell^1$, which includes the Coulomb potential $\hat{V}(k) \sim |k|^{-2}$.
The proof is based on an approximate bosonization using slowly growing patches around the Fermi surface. In contrast to recent proofs in the three-dimensional case, we need a refined analysis of low-energy excitations, as they are less numerous, but carry larger contributions.


\medskip

\noindent Key words: fermionic many-body systems, Fermi gas, bosonization, ground state energy

\medskip

\noindent {\textit{2020 Mathematics Subject Classification}: 81V74, 81Q10, 82D20.}

\end{abstract}

\tableofcontents

\section{Introduction and Main Result}

In the past years, rigorous bosonization techniques allowed for huge progress in the mathematical study of fermionic gases. These techniques are based on the construction of fermionic pair excitation operators, which approximately behave like bosonic operators. The fermionic Hamiltonian is then approximated by a quadratic quasi-bosonic effective Hamiltonian, which can be diagonalized by a Bogoliubov-type transformation. Using a patch-based bosonization, Benedikter, Nam, Porta, Schlein and Seiringer~\cite{benedikter2020optimal,benedikter2021correlation,benedikter2023correlation} first proved a formula for the correlation energy of a 3d mean-field Fermi gas, which was not accessible to earlier mathematical works~\cite{graf1994correlation,hainzl2020correlation}.
The result was shortly afterward extended to Coulomb potentials by Christiansen, Hainzl and Nam~\cite{christiansen2023gell,christiansen2023random,christiansen2024correlation}, using a patch-free bosonization method. Rigorous approximate bosonization also allowed for studying the dynamics~\cite{benedikter2022bosonization}, excitation spectrum~\cite{christiansen2022effective} and momentum distribution~\cite{lill2025bosonized,benedikter2025momentum,benedikter2026momentum} of a 3d mean-field Fermi gas.\\
For the 3d Fermi gas in the dilute thermodynamic limit, Falconi, Giacomelli, Hainzl and Porta~\cite{falconi2021dilute} and Giacomelli~\cite{giacomelli2023optimal,giacomelli2024optimal} applied a similar bosonization technique to improve existing results by~\cite{lieb2005ground} on the ground state energy. By a further improvement of this technique, Giacomelli, Hainzl, Nam and Seiringer~\cite{giacomelli2024huang,giacomelli2025huang} very recently proved a formula conjectured by Huang and Yang~\cite{huang1957quantum} for the ground state energy of the dilute Fermi gas.\\
Let us also point out that Lauritsen and Seiringer~\cite{lauritsen2024ground,lauritsen2024groundLOW} and Lauritsen~\cite{lauritsen2025almost} obtained similar results on the ground state energy of the dilute Fermi gas in 1, 2 and 3 dimensions without bosonization.\\

In this article, we prove a formula for the correlation energy of the 2d mean-field Fermi gas, which is the analog of the 3d formula proven in~\cite{benedikter2023correlation,christiansen2023random}.
Our formula is of the form conjectured by Rajagopal and Kimball~\cite[(21)]{rajagopal1977correlations}, which is the 2d analog of the 3d correlation energy formula conjectured by Gell-Mann and Brueckner~\cite[(19)]{GellMann1957CorrelationEO}.
To prove our result, we use an adaptation to 2 dimensions of the approximate bosonization based on slowly growing patches in~\cite{benedikter2022bosonization}.
More precisely, we consider $N$ fermions on a torus $ \TTT^2 \coloneq [0, 2 \pi]^2 $, described by the Hamiltonian
\begin{equation} \label{eq:HN}
	H_N \coloneq \hbar^2 \sum_{j=1}^N (-\Delta_{x_j})
		+ \lambda \sum_{1 \le i < j \le N} V(x_i - x_j) \;,
\end{equation}
acting on the antisymmetric tensor product space $ L^2(\TTT^2)^{\otimes_{\mathrm{a}} N} $. The mean-field scaling corresponds to choosing $\hbar \coloneq N^{-\frac 12}$ and $\lambda \coloneq N^{-1}$, see also the discussion below.
Further, $ -\Delta_{x_j} $ is the Laplacian acting on the $ j $-th particle, and $ V(x_i-x_j) $ is a position space multiplication operator for some $2\pi$-periodic pair potential function $ V: \RRR^2 \to \RRR $. We assume that the Fourier transform of the latter exists and satisfies
\begin{equation} \label{eq:Vhat}
	\hat{V}(k) = \hat{V}(-k) \ge 0 \quad \forall k \in \ZZZ^2 \;, \qquad
    \hat{V} \in \ell^\infty(\ZZZ^2) \;, \qquad
	\hat{V}(k) \coloneq \int_{\RRR^2} V(x) e^{-ik \cdot x} \di x \;. 
\end{equation}
Our result addresses the \emph{ground state energy}
\begin{equation} \label{eq:GS}
	E_{\GS} \coloneq \inf(\sigma(H_N))
	= \inf_{\substack{\psi \in L^2(\TTT^2)^{\otimes_{\mathrm{a}} N} \\ \norm{\psi} = 1}} \langle \psi, H_N \psi \rangle \;,
\end{equation}
where any $ \psi \in L^2(\TTT^2)^{\otimes_{\mathrm{a}} N} $, $ \norm{\psi} = 1 $ that attains $ E_{\GS} = \eva{\psi, H_N \psi} $ is called a \emph{ground state}.\\
The choice of the mean-filed scaling $ \hbar = N^{-\frac 12} $, $ \lambda = N^{-1} $ in two dimensions is heuristically motivated as follows: We aim at both the kinetic and interaction energy to be extensive, that is, they shall scale\footnote{By $A \sim B$ we mean that there exist constants $c,C > 0$, such that $cB \le A \le CB$.} like $ \sim N $ as $ N \to \infty $. Since each of the $ N $ particles interacts with $ \sim N $ many other particles, the interaction energy is expected to scale like $ \sim \lambda N^2 $, which requires choosing $ \lambda \sim N^{-1} $. To motivate the choice of $ \hbar $, consider the interaction-free case $ V = 0 $. Here, a ground state is given by the Slater determinant (called Fermi ball state or Fermi sea state)
\begin{equation} \label{eq:psiFS}
	\psi_{\FS}(x_1, \ldots, x_N)
	\coloneq (N!)^{-\frac 12} \det \left( (2 \pi)^{-1} e^{i k_j x_\ell} \right)_{j,\ell=1}^N \;,
\end{equation}
where $ (k_j)_{j=1}^N \subset \ZZZ^2 $ is a family of momenta minimizing the kinetic energy
\begin{equation} \label{eq:EFS_kin}
	E_{\FS, \mathrm{kin}}
	\coloneq \eva{\psi_{\FS}, \sum_{j=1}^N (- \Delta_{x_j}) \psi_{\FS}}
	= \sum_{j=1}^N |k_j|^2 \;.
\end{equation}
Without loss of generality, we assume that $ N $ is chosen such that the $ k_j $ fill up a ball, called \emph{Fermi ball}:
\begin{equation} \label{eq:BF}
	\{ k_j \}_{j=1}^N = B_{\F} \;, \qquad
    B_{\F} \coloneq \{ k \in \ZZZ^2 ~|~ |k| < k_{\F}\} \;,
\end{equation}
for some suitable \emph{Fermi momentum} $k_{\F} > 0$ satisfying
\begin{equation}
    k_{\F}^2
    = \frac{1}{2}\left(\inf_{p\in B_F^c}|p|^2 + \sup_{q\in B_F}|q|^2\right) \;.
\end{equation}
Note that $ k_{\F} \sim N^{\frac 12} $ as $ N \to \infty $. We therefore expect $ E_{\FS, \mathrm{kin}} \sim N k_{\F}^2 \hbar^2 $, which motivates the choice $ \hbar \sim k_{\F}^{-1} \sim N^{-\frac 12} $.\\
For a generic interaction $ V \neq 0 $, no closed expression for a ground state or $ E_{\GS} $ is known, but one may derive a simple upper bound on $ E_{\GS} $ by the variational principle, using $ \psi_{\FS} $ as a trial state:
\begin{equation}
	E_{\GS}
	\le E_{\FS}
	\coloneq \eva{\psi_{\FS}, H_N \psi_{\FS}} \;. 
\end{equation}
While $E_{\FS} \sim N$, we rigorously establish the next-order correction to $E_{\GS}$, which is of order $\hbar = N^{-\frac 12}$.

\begin{theo}[Upper and lower bound on the ground state energy] \label{thm:main}
Let the Fourier transform of the interaction potential satisfy $ \hat{V}(k) = \hat{V}(-k) \ge 0 $ and $ \sum_{k \in \ZZZ^2} |k|^{2-b} \hat{V}(k)^2 < \infty $ for some $b \in (0,1)$. Then,
\begin{equation} \label{eq:main_upperbound}
\begin{aligned}
    E_{\GS} 
    &\leq E_{\FS} + E^{\RPA} + o(N^{-\frac 12}) \;,
\end{aligned}
\end{equation}
where, defining $\kappa \coloneq \pi^{-\frac 12}$  such that $k_{\F} = \kappa N^{\frac 12} + o(N^{\frac 12})$, the RPA energy
\begin{equation} \label{eq:ERPA}
    E^{\RPA}
    \coloneqq \hbar \kappa \sum_{k \in \ZZZ^2} \frac{|k|}{\pi} \int_0^\infty F \left( \frac{\hat{V}(k)}{4 \pi} \left(1 - \frac{\lambda}{\sqrt{\lambda^2 + 1}} \right) \right) \di \lambda \;, \qquad
    F(x) \coloneqq \log(1+x) - x \;,
\end{equation}
is bounded by $ 0 \ge E^{\RPA} \ge - C N^{-\frac 12} $.\\
Further, if $ \sum_{k \in \ZZZ^2} |k| \hat{V}(k) < \infty $ holds, then we even have
\begin{equation} \label{eq:main}
\begin{aligned}
    E_{\GS}
    = E_{\FS} + E^{\RPA} + o(N^{-\frac 12}) \;.
\end{aligned}
\end{equation}

\end{theo}

\begin{proof}
The lower bound is proven in Proposition~\ref{prop:lowerbound} and the upper bound in Proposition~\ref{prop:upperbound}.
The bound on $E^{\RPA}$ follows from Lemma~\ref{lem:trace_evaluation}, where $E^{\RPA} \leq 0$ is evident from $F(x) \leq 0$.
\end{proof}

\begin{remarks}
\begin{enumerate}

\item \emph{Main novelties in two dimensions.} Compared to the 3d case, the main complication in 2d is that the relative coupling is now $\lambda / \hbar^2 = 1$ instead of $N^{-\frac 13}$. That means, the \textbf{2d mean-field regime is no longer a regime of small coupling}. This is in part compensated by the fact that certain subsets of $\mathbb{R}^2$ contain much less lattice points than their 3d counterparts. However, some estimates lose their validity.\\
While our proof follows the general strategy of~\cite{benedikter2023correlation}, including ideas from~\cite{christiansen2023random}, we now need a gap argument (Lemma~\ref{lem:cN_delta_bound}) to achieve optimal a priori bounds, similar to the argument in~\cite[Lemma~3.5]{falconi2021dilute}. We further achieve bounds on non-bosonizable terms (Lemma~\ref{lem:cE_1estimate}) through a three-scale decomposition, which replaces the two-scale decomposition of~\cite[Prop.~2.3]{christiansen2023random}, followed by a careful analysis of the number of lattice points at different energy scales. Also, the bound on low-energy excitations in Lemma~\ref{lem:Q_QR_bound} requires an additional effort with respect to the 3d case, and we need to derive the 2d counterparts of some 3d estimates on sums over lattice points (see Appendix~\ref{app:numbertheory}).

\item \label{rem:V} \emph{On the conditions for the potentials.} If $\sum_{k \in \ZZZ^2} |k| \hat{V}(k) < \infty$, then there must exist some $C > 0$ such that $\hat{V}(k) \leq C |k|^{-1} \; \forall k \in \ZZZ^2_*$, since otherwise, the sum would have infinitely many contributions $\ge C$ and thus be divergent. Hence, $\sum_{k \in \ZZZ^2} |k|^2 \hat{V}(k)^2 < \infty$, so in particular the assertion $\sum_{k \in \ZZZ^2} |k|^{2-b} \hat{V}(k)^2 < \infty$ holds with any $b \ge 0$.

\item \emph{Coulomb potential.} It would be highly desirable to establish also a lower bound for the 2d Coulomb potential, $\hat{V}(k) \sim |k|^{-2}$. As mentioned above, the correlation energy for a 3d Fermi gas with Coulomb potentials was very recently established in~\cite{christiansen2023gell,christiansen2024correlation}. However, the method for obtaining a priori estimates on the kinetic energy for the lower bound on $E_{\GS}$ in~\cite{christiansen2024correlation} relies on the relative coupling being $\lambda / \hbar^2 \ll 1$, which is no longer true in 2 dimensions. Here, we instead use an Onsager-type argument as in~\cite{benedikter2021correlation,christiansen2023gell} to derive a priori bounds, which is restricted to the case $\sum_k \hat{V}(k) |k| < \infty$. 
It is an interesting question for future research how to derive a priori bounds for the 2d Coulomb case in spite of the relative coupling being of order 1.\\
Let us also mention that the original prediction by Rajagopal and Kimball~\cite{rajagopal1977correlations} is for $V(x) \sim k_{\F}^{-1} |x|^{-1}$ so $\hat{V}(k) \sim k_{\F}^{-1} |k|^{-1}$. In other words, a 3d Coulomb potential is plugged into the 2d Hamiltonian. Using this potential would massively simplify the analysis, as the factor of $k_{\F}^{-1}$ entails a weak relative coupling.

\end{enumerate}
\end{remarks}

\noindent The rest of this paper is organized as follows: In Section~\ref{sec:mathdef}, we introduce some notation and particle--hole transform the Hamiltonian. We then derive a priori estimates needed for the lower bound on $E_{\GS}$ in Section~\ref{sec:a_priori_bounds}, as well as estimates for non-bosonizable terms in Section~\ref{sec:non_bosonizable}. In Section~\ref{sec:patch-decomposition}, we introduce the patch-based approximate bosonization and compile bosonization error bounds. Based on this, we define the approximately bosonic effective Hamiltonian and Bogoliubov transformation in Section~\ref{sec:pseudobosonic}. After providing bosonization error estimates for the kinetic energy in Section~\ref{sec:linearizing_kinetic}, we finally conclude the bounds on $E_{\GS}$ in Section~\ref{sec:proof_main}.\\
Appendix~\ref{app:numbertheory} contains some number theoretical estimates specific to the 2d case.

\section{Mathematical Definitions}
\label{sec:mathdef}

We largely adopt the notation of~\cite{benedikter2021correlation,benedikter2023correlation}, working in second quantization. The fermionic Fock space over the 2D torus $ \TTT^2 = [0, 2 \pi]^2 $ is defined as
\begin{equation} \label{eq:cF}
	\cF \coloneq \bigoplus_{N=0}^\infty L^2(\TTT^2)^{\otimes_{\mathrm{a}} N} \;,
\end{equation}
with vacuum vector $ \Omega = (1,0,0,\ldots) \in \cF $. The standard fermionic creation and annihilation operators for $ f \in L^2(\TTT^2) $ are $ a^*(f), a(f): \cF \to \cF $ with operator norm bounds $ \norm{a^*(f)}, \norm{a(f)} \le \norm{f}_2 $. To each momentum $ p \in \ZZZ^2 $, we associate a creation and annihilation operator
\begin{equation} \label{eq:aastar}
	a_p^* \coloneq a^*(e_p) \;, \qquad
	a_p \coloneq a(e_p) \;, \qquad
	e_p \coloneq (2 \pi)^{-1} e^{ip \cdot x} \;,
\end{equation}
where $ (e_p)_{p \in \ZZZ^2} \subset L^2(\TTT^2) $ is the plane-wave orthonormal basis and where $ a_p^*, a_p $ satisfy the canonical anticommutation relations (CAR)
\begin{equation} \label{eq:CAR}
	\{ a_p, a_{p'}^* \} = \delta_{p,p'} \;, \qquad
	\{ a_p, a_{p'} \} = \{ a_p^*, a_{p'}^* \} = 0 \qquad
    \forall p, p' \in \ZZZ^2\;.
\end{equation}
This allows to conveniently re-write $ H_N $ (defined in~\eqref{eq:HN}) in momentum space: If we lift $ H_N $ on $ L^2(\TTT^2)^{\otimes_{\mathrm{a}} N} $ to an operator $ \cH_N $ on $ \cF $, then a quick calculation reveals that
\begin{equation} \label{eq:cH_N}
	\cH_N = \sum_{p \in \ZZZ^2} \hbar^2 |p|^2 a_p^* a_p
		+ \frac{1}{2 (2 \pi)^2 N} \sum_{k,p,q \in \ZZZ^2} \hat{V}(k) a^*_{p+k} a^*_{q-k} a_q a_p \;.
\end{equation}
To analyze this Hamiltonian, we introduce the unitary particle--hole transformation $ R: \cF \to \cF $, which flips the operators inside the Fermi ball (defined in $ B_{\F} $~\eqref{eq:BF})
\begin{equation} \label{eq:R}
	R^* a_p R
	\coloneq \chi(p \in B_{\F}^c) a_p + \chi(p \in B_{\F}) a_p^* \;, \qquad
	B_{\F}^c \coloneq \ZZZ^2 \setminus B_{\F} \;.	
\end{equation}
This transformation generates the Fermi sea state~\eqref{eq:psiFS} as $ \psi_{\FS} = R \Omega $. Note that $R^2 = 1$. As in~\cite{benedikter2021correlation},~\cite{christiansen2023random}, for $ k \in \ZZZ^2 $, we introduce the lune and the pair creation and shift operators
\begin{equation} \label{eq:Lkbkdk}
    L_k \coloneq B_{\F}^c \cap (B_{\F}+k) \;, \quad
    b^*(k) \coloneq \sum_{p \in L_k} a_p^* a_{p-k}^* \;, \quad
    d(k) \coloneq \!\! \sum_{p \in B_{\F}^c \cap (B_{\F}^c+k)} \!\!  a_{p-k}^* a_p
        - \!\!  \sum_{h \in B_{\F} \cap (B_{\F}-k)} \!\!  a_{h+k}^* a_h \;,
\end{equation}
where $b(0) = 0$ and $L_0 = \emptyset$. This allows for conveniently rewriting
\begin{equation*}
    R^* \sum_{p \in \ZZZ^2} a_{p+k}^* a_p R
    = b^*(k) + b(-k) + d(k)^* \;.
\end{equation*}
Using the CAR~\eqref{eq:CAR}, we then obtain 
\begin{equation} \label{eq:HNconjugation}
\begin{aligned}
	R^* \cH_N R 
	&= E_{\FS} + \HHH_0 + Q_{\B} + \cE_1 + \cE_2 + \XXX \;, \\
	\HHH_0 &\coloneq \sum_{p \in \ZZZ^2} e(p) a_p^* a_p \;, \quad \textnormal{with} \quad e(p) \coloneq \hbar^2 \abs{|p|^2 - k_{\F}^2} \;,\\
	Q_{\B} &\coloneq \frac{1}{(2 \pi)^2 N} \sum_{k \in \ZZZ^2_*} \hat{V}(k) \left( b^*(k) b(k) + \frac 12 \big( b^*(k) b^*(-k) + b(-k) b(k) \big) \right)\;, \\
	\cE_1 &\coloneq \frac{1}{2 (2 \pi)^2 N} \sum_{k \in \ZZZ^2_*} \hat{V}(k) d^*(k) d(k) \;, \\
	\cE_2 &\coloneq \frac{1}{2 (2 \pi)^2 N} \sum_{k \in \ZZZ^2_*} \hat{V}(k) (d^*(k) b(-k) + \textnormal{h.c.} ) \;, \\
	\XXX &\coloneq -\frac{1}{2 (2 \pi)^2 N} \sum_{k \in \ZZZ^2_*} \hat{V}(k) \sum_{p \in L_k} ( a_p^* a_p + a_{p-k}^* a_{p-k} ) \;, \\
\end{aligned}
\end{equation}
where $ \ZZZ^2_* \coloneq \ZZZ^2 \setminus \{(0,0)\} $. Note that there is an additional $ (2 \pi)^2 $ in the denominator with respect to~\cite[(2.5),(2.6)]{benedikter2023correlation} due to our different Fourier transform convention.

\section{A Priori Estimates}
\label{sec:a_priori_bounds}

To control error terms, we need to establish estimates on expectations of powers of kinetic energy and excitation number operators. In this section, we derive such estimates for approximate ground states in a similar sense to~\cite[(4.18)]{hainzl2020correlation}, which will be useful to prove the lower bound on $E_{\GS}$.

\begin{definition} \label{def:approxGS}
We say that $ \xi \in \cF $ \textbf{belongs to an approximate ground state} $ \psi = R \xi $ if $ R \xi \in L^2(\TTT^2)^{\otimes_{\mathrm{a}} N} $, $ \norm{\xi} = 1 $, and
\begin{equation} \label{eq:approxGS}
    \eva{R \xi, H_N R \xi} - E_{\FS}
    \le C \hbar \;.
\end{equation}
\end{definition}

Since $E_{\GS} \le E_{\FS}$, for any ground state $ \psi_{\GS} $, the vector $ \xi = R \psi_{\GS} $ belongs to an approximate ground state. We start with extracting a bound for $ \HHH_0 $, using an Onsager-type argument as in~\cite[Lemma~4.1]{benedikter2023correlation} and~\cite[Sect.~10.2]{christiansen2023random}.

\begin{lemma}[Onsager bound]
Assume $\hat{V}\geq0$ and $\sum_{k\in\ZZZ^2}|k| \hat{V}(k) <\infty$ and let $\xi\in \cF$ such that $ R \xi \in L^2(\TTT^2)^{\otimes_{\mathrm{a}} N} $. Then, there exists a $C > 0$ such that
\begin{equation} \label{eq:HHH0_operatorbound}
    \langle \xi, \HHH_0 \xi \rangle
    \le \langle R \xi, H_N R \xi \rangle - E_{\FS} + C N^{-\frac{1}{2}} \sum_{q\in\ZZZ^2_*}|q|\hat{V}(q) \;.
\end{equation}
In particular, if $\xi$ belongs to an approximate ground state in the sense of Definition~\ref{def:approxGS}, then
\begin{equation} \label{eq:Onsager-Bound}
	\langle \xi, \HHH_0 \xi \rangle
	\le C N^{-\frac{1}{2}} \;.
\end{equation}
\label{lemma:Onsager-Bound}
\end{lemma}

\begin{proof}
As in~\cite[Lemma~4.1]{benedikter2023correlation}, we complete the square as
\begin{align*}
    0 &\leq \frac 12 \int_{\mathbb{T}^2\times\mathbb{T}^2} \Bigg(\sum_{i=1}^N \delta(x_i - x) - \frac{N}{(2 \pi)^2} \Bigg) V(x-y) \Bigg(\sum_{j=1}^N \delta(x_j - y) - \frac{N}{(2 \pi)^2} \Bigg) \, \di x \, \di y \\
    &= \sum_{1 \leq i < j \leq N} V(x_i - x_j)
        - N^2 \frac{\hat{V}(0)}{2(2 \pi)^2}
        + N \frac{V(0)}{2} \;,
\end{align*}
where we recognize the first term as $ N = \lambda^{-1} $ times the interaction energy in $ H_N $~\eqref{eq:HN}. Thus, adding the kinetic energy, we get
\begin{equation*}
    \sum_{j=1}^N (-\hbar^2 \Delta_{x_j})
    \le H_N + \frac{V(0)}{2} - N \frac{\hat{V}(0)}{2(2 \pi)^2} \;.
\end{equation*}
We now take the expectation in $ R \xi $. A quick calculation for $ R \xi \in L^2(\TTT^2)^{\otimes_{\mathrm{a}} N} $ reveals
\begin{equation} \label{eq:kinetic-bound}
\begin{aligned}
    \eva{\xi, \HHH_0 \xi}
    &= \eva{ R \xi, \left( \sum_{j=1}^N -\hbar^2 \Delta_{x_j} \right) R \xi}
        - \sum_{p \in B_\F} \hbar^2 p^2 \\
    &\leq \big( \langle R \xi, H_N R \xi \rangle - E_{\FS} \big)
        + E_{\FS}
        + \frac{V(0)}{2} 
        - N \frac{\hat{V}(0)}{2(2 \pi)^2} - \sum_{p \in B_\F} \hbar^2 p^2 \;.
\end{aligned}
\end{equation}
The Fermi sea energy can be written as
\begin{equation*}
    E_\FS = N \frac{\hat{V}(0)}{2(2 \pi)^2} - \frac{1}{2(2 \pi)^2N} \sum_{k,k' \in B_\F} \hat{V}(k - k') + \sum_{p \in B_\F} \hbar^2 p^2 \;.
\end{equation*}
Next, observe that
\begin{align*}
    \sum_{k, k' \in B_\F} \hat{V}(k - k') 
    &= \sum_{k \in B_\F} \Bigg( \sum_{k' \in \mathbb{Z}^2} \hat{V}(k - k') - \sum_{k' \in B_\F^c} \hat{V}(k - k') \Bigg)
    = (2 \pi)^2 N V(0) - \sum_{\substack{k \in B_\F \\ k' \in B_\F^c}} \hat{V}(k - k') \;.
\end{align*}
For the second term, recalling $\ZZZ^2_* = \ZZZ^2\setminus\{(0,0)\}$, we have
\begin{equation}
    \sum_{\substack{k\in B_\F\\k'\in B_\F^c}}\hat{V}(k-k')
    = \sum_{k\in B_\F}\sum_{q\in B_\F^c+k}\hat{V}(q)
    =\sum_{q\in\ZZZ^2_*}\left|L_q\right|\hat{V}(q)
    \leq CN^{\frac{1}{2}}\sum_{q\in\ZZZ^2_*}|q|\hat{V}(q) \;.
    \label{eq:3.2}
\end{equation}
Putting together~\eqref{eq:kinetic-bound}--\eqref{eq:3.2} proves the claimed result~\eqref{eq:HHH0_operatorbound}. Then,~\eqref{eq:Onsager-Bound} follows immediately from the definition of an approximate ground state.
\end{proof}

Based on this bound, we derive further a priori estimates, which involve the following gapped number operator.

\begin{definition}
Recall the excitation energy $ e(p) = \hbar^2 \abs{|p|^2 - k_{\F}^2} $. Given $ \delta \in [0,\tfrac 12] $, we define the \textbf{gap} $ \Gcutoff $ and the \textbf{gapped number operator} $ \cN_\delta $ as
\begin{equation} \label{eq:cN_delta}
    \Gcutoff \coloneq \{ p \in \ZZZ^2 ~|~ e(p) \leq \hbar N^{-\delta} \} \;, \qquad
    \cN_\delta \coloneqq \sum_{p \in \ZZZ^2 \setminus \Gcutoff} a_p^* a_p \;.
\end{equation}
\end{definition}

Note that by lattice discretization, we have $e(p) \geq c \hbar ^2$, so there is already a natural gap corresponding to $\delta = \tfrac 12$ and of thickness $N^{-\frac 12}$.\\
A similar $ \cN_\delta $ was introduced in~\cite{benedikter2021correlation} to address the fact that $\HHH_0$ is not stable under propagation by the 3d analog of our quasi-Bogoliubov transformation $ T $ defined in~\eqref{eq:bogoliubov}. We introduce $ \cN_\delta $ for the very same reason. In contrast to the 3d case, we will additionally need the following ``gapped conversion'' to estimate $\cN$ against $\HHH_0$.\\

\begin{lemma} [Bound on $ \Gcutoff $ and gapped conversion] \label{lem:cN_delta_bound}
Given $ \delta \in [0,\tfrac 12] $ and any $ \varepsilon > 0 $, there exist some $ C, C_\varepsilon > 0 $ such that for all $ \xi \in \cF $, $\norm{\xi} = 1$,
\begin{equation} \label{eq:cN_delta_bound}
    |\Gcutoff| \le C_\varepsilon N^{\frac 12 - \delta + \epsilon} \;, \quad
    \eva{\xi, \cN \xi}
    \le |\Gcutoff| + C N^{\frac 12 + \delta} \eva{\xi, \HHH_0 \xi} \;, \quad
    \eva{\xi, \cN_\delta \xi}
    \le C N^{\frac 12 + \delta} \eva{\xi, \HHH_0 \xi} \;.
\end{equation}
\end{lemma}

\begin{proof}
To bound $ |\Gcutoff| $, note that, by definition of $ e(p) $~\eqref{eq:HNconjugation}, $ \Gcutoff $ contains $ p \in \ZZZ^2 $ with $ k_{\F}^2 - N^{\frac 12 - \delta} \le |p|^2 \le k_{\F}^2 + N^{\frac 12 - \delta} $. As $ |p|^2 $ can only take integer values $ |p|^2 = n \in \NNN $, we can decompose $ \Gcutoff $ into $ \le N^{\frac 12 - \delta} $ spheres of the kind $ S_n = \{ p \in \ZZZ^2 ~|~ |p|^2 = n\} $. By Lemma~\ref{lem:sphere_point_estimate}, each sphere has $ |S_n| \le C_\varepsilon N^\varepsilon $ points. This concludes the first bound of~\eqref{eq:cN_delta_bound}. The second bound follows from $ p \notin \cN_\delta \Rightarrow e(p) > N^{-\frac 12 - \delta} $:
\begin{equation*}
    \eva{\xi, \cN \xi}
    = \sum_{p \in \Gcutoff} \eva{\xi, a_p^*a_p \xi}
        +\sum_{p \in \ZZZ^2 \setminus \Gcutoff}\frac{1}{e(p)}e(p) \eva{\xi, a_p^*a_p \xi}
    \leq |\Gcutoff|
        + C N^{\frac 12 + \delta} \eva{\xi, \HHH_0 \xi} \;.
\end{equation*}
The third bound readily follows by dropping the contribution with $ p \in \Gcutoff $.
\end{proof}

In the proof of our final a priori bounds, we will need the following simple estimates.
\begin{lemma}[Naive bounds on $b$ and $d$] \label{lem:naivebounds_bdXXX}
For $k \in \ZZZ^2_*$, let $L_k, b(k)$, and $ d(k)$ be defined as in~\eqref{eq:Lkbkdk}. Then, for all $\xi \in \cF$,
\begin{equation} \label{eq:naivebounds_bdXXX}
    \Vert b(k) \xi \Vert^2 \leq |L_k| \eva{\xi, \cN \xi} \;, \qquad
    \Vert b^*(k) \xi \Vert^2 \leq |L_k| \eva{\xi, (\cN + 1) \xi} \;, \qquad
    \Vert d(k) \xi \Vert^2 \leq 8 \eva{\xi, \cN^2 \xi} \;.
\end{equation}
\end{lemma}

\begin{proof}
By the Cauchy--Schwarz inequality and $\Vert a_{p-k} \Vert \leq 1$,
\begin{equation*}
    \Vert b(k) \xi \Vert^2
    \leq \left( \sum_{p \in L_k} \Vert a_{p-k} a_p \xi \Vert \right)^2
    \leq |L_k| \sum_{p \in L_k} \Vert a_p \xi \Vert^2
    \leq |L_k| \eva{\xi, \cN \xi} \;.
\end{equation*}
Further, using the CAR, we estimate
\begin{align*}
    &\Vert b^*(k) \xi \Vert^2
    = \sum_{p,q \in L_k} \eva{\xi, a_{p-k} a_p a_q^* a_{q-k}^* \xi} \\
    &= \sum_{p,q \in L_k} \eva{\xi, a_q^* a_{q-k}^* a_{p-k} a_p \xi}
        -\sum_{p \in L_k} \eva{\xi, (a_p^* a_p + a_{p-k}^* a_{p-k}) \xi}
        + |L_k|
    \leq \Vert b(k) \xi \Vert^2 + |L_k| \;.
\end{align*}
To bound $d(k)$, we split
\begin{align*}
    \|d(k)\xi\|^2
    &\leq 2 \|d_1(k)\xi\|^2 + 2 \|d_2(k)\xi\|^2 \;, \\
    d_1(k) &\coloneq  \sum_{p\in B_\F^c\cap(B_\F^c+k)}  a_{p-k}^*a_p \;, \qquad
    d_2(k) \coloneq  \sum_{h\in B_\F\cap(B_\F-k)} a_{h+k}^*a_h \;.
\end{align*}
Then, using the CAR and then the Cauchy--Schwarz inequality and $\cN \leq \cN^2$, we get
\begin{align} 
    \|d_1(k)\xi\|^2
    &= \sum_{p,q\in B_\F^c\cap(B_\F^c+k)}\langle\xi,a^*_qa_{q-k}a_{p-k}^*a_p\xi\rangle
    \leq \Bigg\vert \sum_{p,q\in \ZZZ^2}\langle\xi,a_q^*a_{p-k}^*a_{q-k}a_p\xi\rangle \Bigg\vert 
        + \sum_{p\in \ZZZ^2}\langle\xi,a_p^*a_p\xi\rangle \nonumber \\
    &\leq \sum_{p,q\in \ZZZ^2} \Vert a_{q-k} a_p \xi \Vert^2
        + \eva{\xi, \cN \xi} 
    \leq 2 \eva{\xi, \cN^2 \xi}  \;. \label{eq:d_estimate_naive}
\end{align}
The estimate for $d_2(k)$ is analogous.
\end{proof}

Our final a priori bounds now read as follows.

\begin{lemma}[A priori bounds]
Assume $\hat{V}\geq0$ and $\sum_{k\in\ZZZ^2}|k| \hat{V}(k) <\infty$ and let $\xi \in \cF$ belong to an approximate ground state in the sense of Definition~\ref{def:approxGS}. Then, for every $\varepsilon > 0$, there exist $C_\varepsilon, C > 0$ such that
\begin{equation}
	\langle \xi, \cN_\delta \xi \rangle
	\le C N^{\delta} \;, 
    \qquad \langle \xi, \cN \xi \rangle
	\le C_\varepsilon N^{\frac{1}{4} + \varepsilon} \;.
\end{equation}
Further, if $R \xi$ is additionally an eigenvector of $H_N$, then
\begin{equation}
\begin{aligned}   
    \langle \xi, \cN^2 \xi \rangle
	&\le C_\varepsilon N^{\frac{1}{2} + \varepsilon} \;, &
    \qquad \langle \xi, \cN \HHH_0 \xi \rangle
	&\le C_\varepsilon N^{-\frac{1}{4} + \varepsilon} \;,
    \qquad \langle \xi, \cN^2 \HHH_0 \xi \rangle
	\le C_\varepsilon N^{\varepsilon} \;, \\
    \langle \xi, \cN \cN_\delta \xi \rangle
	&\le C_\varepsilon N^{\frac{1}{4} + \delta + \varepsilon} \;, &
    \quad \langle \xi, \cN^2 \cN_\delta \xi \rangle
	&\le C_\varepsilon N^{\frac{1}{2} + \delta + \varepsilon} \;. \\
\end{aligned}
\end{equation}
\label{lemma:A-Priori-Bounds}
\end{lemma}

We remark that in 3d, also $ \HHH_0 \sim \hbar $ and $ \cN_\delta \sim N^\delta $, but $\cN \sim k_{\F}$ instead of our $ \cN \sim k_{\F}^{\frac 12 + \varepsilon} $. This is due to the gap split in Lemma~\ref{lem:cN_delta_bound}, which improves our bound in 2d.

\begin{proof}
The first two bounds readily follow by plugging~\eqref{eq:Onsager-Bound} into~\eqref{eq:cN_delta_bound} and optimizing $ \delta = \frac 14 $ for $ \eva{\xi, \cN \xi} $. To obtain bounds involving higher powers of $ \cN $, we follow the strategy of \cite[Sec.~10.2]{christiansen2023random}: we introduce $\tilde{H}_N \coloneq R^* H_N R - E_\FS$ and note that $ \tfrac 12 \cN = \sum_{p \in B_\F^c} a_p^* a_p = \sum_{h \in B_\F} a_h^* a_h$ on physical excitation states $ \xi \in R[L^2(\TTT^2)^{\otimes_{\mathrm{a}} N}] $. First, we prove that for such states
\begin{equation}
    \langle \xi, \cN^2 \HHH_0 \xi \rangle \le C N^{-\frac 12} \langle \xi, (\cN^2 + 1) \xi \rangle \;.
    \label{eq:3.1}
\end{equation}
From Lemma~\ref{lemma:Onsager-Bound}, we recover $ \eva{\xi, \HHH_0 \xi} \le \eva{R \xi, H_N R \xi} - E_{\FS} + C N^{-\frac 12} $, which implies
\begin{align*}
    \eva{\xi,\cN^2 \HHH_0 \xi}
    &= \eva{\xi, \cN \HHH_0 \cN \xi} 
    \le \eva{\xi, \cN \tilde{H}_N \cN \xi} + C N^{- \frac 12} \eva{\xi, \cN^2 \xi} \\
    &= \tfrac{1}{2} \left( \eva{\xi, \cN^2 \tilde{H}_N \xi}
        + \eva{\xi, \tilde{H}_N \cN^2 \xi}
        - \eva{\xi, [\cN, [\cN, \tilde{H}_N]] \xi} \right)
        + C N^{-\frac 12} \eva{\xi, \cN^2 \xi} \\
    &\leq \left| \eva{\xi, [\cN, [\cN, \tilde{H}_N]] \xi} \right|
        + C N^{-\frac 12} \eva{\xi, \cN^2 \xi}\;,
\end{align*}
where in the last line, we used that $\xi$ is an eigenvector of $\tilde{H}_N$, whose eigenvalue is $\le C \hbar$ due to~\eqref{eq:approxGS}.
We now explicitly compute the double commutator with~\eqref{eq:HNconjugation}, using that $ [\cN, a_p^* a_p] = 0 $, $[\cN, d(k)] = 0$, and $[\cN, b(k)] = - 2 b(k)$:
\begin{equation} \label{eq:doublecommutator_fixed}
\begin{aligned}
    &\left|\left\langle \xi, [\cN,[\cN,\tilde{H}_N] ] \xi \right\rangle\right|
    = |\eva{\xi, \left[\cN,[\cN,(\HHH_0 + Q_{\B} + \cE_1 + \cE_2 + \XXX)]\right] \xi}| \\
    &= |\eva{\xi, \left[\cN,[\cN,(Q_{\B} + \cE_2)]\right] \xi}|
    \leq \frac{C}{N} \sum_{k \in \ZZZ^2_*} \hat{V}(k) 
        \big( \Vert d(k) \xi \Vert \Vert b(-k) \xi \Vert 
        + \Vert b^*(k) \xi \Vert \Vert b(-k) \xi \Vert \big) \\
    &\leq C N^{-1} \sum_{k \in \ZZZ^2_*} \hat{V}(k) |L_k| \eva{\xi, (\cN^2 + 1) \xi} \;,
\end{aligned}
\end{equation}
where we used Lemma~\ref{lem:naivebounds_bdXXX} in the last line\footnote{We remark that the proof in~\cite[Sect.~10.2]{christiansen2023random} contains a gap, as~\cite[(10.33)]{christiansen2023random} uses the incorrect identity $\sum_{p \in \ZZZ^3} \Vert a_p \xi \Vert^2 = \eva{\xi, R^* \cN R \xi}$. We close this gap here by re-ordering operators within~\eqref{eq:doublecommutator_fixed} and then applying the naive bounds from Lemma~\ref{lem:naivebounds_bdXXX}.}
. With $|L_k| \leq C |k| N^{\frac 12}$ and $\sum_{k \in \ZZZ^2_*} \hat{V}(k) |k| < \infty$, we conclude
\begin{align*}
    &\left|\left\langle \xi, [\cN,[\cN,\tilde{H}_N] ] \xi \right\rangle\right|
    \leq C N^{-\frac 12} \eva{\xi, (\cN^2 + 1) \xi} \quad \Rightarrow \quad
    \eva{\xi,\cN^2 \HHH_0 \xi}
    \leq C N^{-\frac 12} \eva{\xi, (\cN^2 + 1) \xi} \;.
\end{align*}
This establishes \eqref{eq:3.1}. To estimate $\langle \xi, \cN^2 \xi \rangle$, we write with Hölder's inequality:
\begin{align*}
    \langle\xi,\cN^2\xi\rangle
    &\leq \eva{\cN\xi,\Bigg( \sum_{p\in \Gcutoff}a_p^*a_p+\sum_{p\notin \Gcutoff}a_p^*a_p \Bigg)\cN\xi}^\frac{2}{3}
    \leq\Bigg(\langle\xi,\cN^2\xi\rangle|\Gcutoff| + \langle\xi,\cN\Bigg(\sum_{p\notin \Gcutoff}a^*_pa_p\Bigg)\cN\xi\rangle\Bigg)^\frac{2}{3} \;.
\end{align*}
If $p\notin \Gcutoff$, then $e(p) \geq N^{-\delta-\frac 12}$ by definition, and we have
\begin{equation*}
    \langle\xi,\cN\Bigg(\sum_{p\notin \Gcutoff}a^*_pa_p\Bigg)\cN\xi\rangle 
    = \sum_{p\notin \Gcutoff}\frac{1}{e(p)}\langle\xi,\cN e(p)a^*_pa_p\cN\xi\rangle
    \leq N^{\delta+\frac{1}{2}}\langle\xi,\cN^2\HHH_0\xi\rangle \;.
\end{equation*}
Thus, with~\eqref{eq:cN_delta_bound} $|\Gcutoff| \leq C_\varepsilon N^{\frac 12 - \delta + \varepsilon}$ and with~\eqref{eq:3.1}, we finally get
\begin{equation*}
    \langle\xi,(\cN^2 + 1)\xi\rangle
    \leq C_\varepsilon \left(N^{\frac{1}{2}-\delta+\varepsilon}+N^\delta\right)^\frac{2}{3}\langle\xi,(\cN^2 + 1)\xi\rangle^\frac{2}{3} \;.
\end{equation*}
We find that $\delta=1/4$ is optimal, where
\begin{equation*}
    \langle\xi,(\cN^2 + 1)\xi\rangle \leq C_\varepsilon N^{\frac{1}{2}+\varepsilon}
    \quad \Rightarrow \quad
    \langle\xi,\cN^2\xi\rangle \leq C_\varepsilon N^{\frac{1}{2}+\varepsilon}
    \quad \overset{\eqref{eq:3.1}}{\Rightarrow} \quad
    \langle\xi,\cN^2\HHH_0\xi\rangle \leq C_\varepsilon N^\varepsilon \;.
\end{equation*}
By the Cauchy--Schwarz inequality, we obtain
\begin{equation*}
    \langle\xi, \cN\HHH_0\xi\rangle\leq\langle\xi,\HHH_0\xi\rangle^\frac{1}{2}\langle\xi,\cN^2\HHH_0\xi\rangle^\frac{1}{2}\leq C_\varepsilon N^{-\frac{1}{4}+\varepsilon} \;.
\end{equation*}
It remains to show the estimates that involve the gapped number operator~\eqref{eq:cN_delta}. By definition
\begin{equation*}
    \cN_\delta = \sum_{p\in\ZZZ^2 : e(p)\geq N^{- \frac 12 -\delta}}\frac{1}{e(p)}e(p)a_p^*a_p
    \leq N^{\frac 12 + \delta}\HHH_0\;,
\end{equation*}
therefore, for $m=0,1,2$, we obtain
\begin{align*}
    \cN_\delta\cN^m=\cN^\frac{m}{2}\cN_\delta\cN^\frac{m}{2}
    \leq N^{\frac 12 + \delta} \cN^\frac{m}{2}\HHH_0\cN^\frac{m}{2}
    = N^{\frac 12 + \delta}\cN^m\HHH_0\;,
\end{align*}
which immediately leads to the claimed bounds.
\end{proof}

\section{Bounding Non-Bosonizable Terms}
\label{sec:non_bosonizable}

In this section, we bound the non-bosonizable terms $ \XXX $, $ \cE_1 $ and $ \cE_2 $ defined in~\eqref{eq:HNconjugation}. Additionally, for the upper bound in case of singular potentials as in~\cite[Theorem~A.1]{benedikter2023correlation}, we will estimate the two operators $ \widetilde{\XXX} $ and $ \widetilde{\cE_1} $, which are obtained from $ \XXX $ and $ \cE_1 $ by restricting the sum in $k$ to $|k| < C N^{\frac 12}$ for some fixed, large enough $C>0$:
\begin{equation} \label{eq:XXXtilde_cE1tilde}
\begin{aligned}
    \widetilde{\XXX}
    &\coloneq -\frac{1}{2 (2 \pi)^2 N} \sum_{k \in \ZZZ^2_* : |k| < C N^{\frac 12}} \hat{V}(k) \sum_{p \in L_k} ( a_p^* a_p + a_{p-k}^* a_{p-k} ) \;, \\
    \widetilde{\cE}_1
    &\coloneq \frac{1}{2 (2 \pi)^2 N} \sum_{k \in \ZZZ^2_* : |k| < C N^{\frac 12}} \hat{V}(k) d^*(k) d(k) \;.
\end{aligned}
\end{equation}

\begin{lemma} \label{lemma:Exchange-Term-Bound}
Recall the definitions~\eqref{eq:HNconjugation} and~\eqref{eq:XXXtilde_cE1tilde} of $ \XXX $ and $\widetilde{\XXX}$.
If $\sum_{k \in \ZZZ^2} |k|^{2-b} \hat{V}(k)^2 < \infty$ for some $b \in (0,1)$, then there exists a $C>0$ such that for all $ \xi \in \cF $,
\begin{equation} \label{eq:XXX_bound_1}
    |\langle\xi, \widetilde{\XXX} \xi \rangle|
    \leq C N^{-1 + \frac b4} \eva{\xi, \cN \xi} \;.
\end{equation}
Further, if $\sum_{k \in \ZZZ^2} \hat{V}(k) < \infty$ and $ \xi $ belongs to an approximate ground state in the sense of Definition~\ref{def:approxGS}, then for any $\varepsilon>0$, there exists a $C_\varepsilon>0$ such that
\begin{equation} \label{eq:XXX_bound_2}
    |\langle\xi,\XXX\xi \rangle|
    \leq C_\varepsilon \hbar N^{-\frac 14 + \varepsilon} \;.
\end{equation}
\end{lemma}

\begin{proof}
By definition of $ \widetilde{\XXX} $,
\begin{align*}
    |\langle\xi, \widetilde{\XXX} \xi\rangle|
    &\leq \frac{C}{N}\sum_{k \in \ZZZ^2_* : |k|< C N^{\frac 12}}|\hat{V}(k)|\left|\sum_{p\in\lune}\langle\xi,a_p^*a_p\xi\rangle+\sum_{h\in\lune-k}\langle\xi,a_h^*a_h\xi\rangle\right| \\
    &\leq \frac{C}{N} \eva{\xi, \cN \xi} \sum_{k \in \ZZZ^2_* : |k|< C N^{\frac 12}} \hat{V}(k) \;.
\end{align*}
From the Cauchy--Schwarz inequality, we get
\begin{equation} \label{eq:Vsplit}
    \sum_{k \in \ZZZ^2_* : |k|< C N^{\frac 12}} \hat{V}(k) 
    \leq \Bigg( \sum_{k \in \ZZZ^2_* : |k|< C N^{\frac 12}} |k|^{b-2} \Bigg)^{\frac 12} 
        \Bigg( \sum_{k \in \ZZZ^2_* : |k|< C N^{\frac 12}} |k|^{2-b} \hat{V}(k)^2 \Bigg)^{\frac 12}
    \leq C N^{\frac b4} \;,
\end{equation}
which implies~\eqref{eq:XXX_bound_1}.\\
To prove~\eqref{eq:XXX_bound_2}, we extend the sum to $k \in \ZZZ^2_*$ and note that $\sum_{k \in \ZZZ^2_*}  \hat{V}(k) < \infty$.
Then, for $\xi$ belonging to an approximate ground state, we bound by Lemma~\ref{lemma:A-Priori-Bounds}: $ \eva{\xi, \cN \xi} \le C_\varepsilon N^{\frac 14 + \varepsilon} $.
\end{proof}

Let us now turn to the terms $\cE_1$ and $\widetilde{\cE}_1$. For bounding $\widetilde{\cE}_1$,~\eqref{eq:naivebounds_bdXXX} will turn out sufficient. For $\cE_1$, in contrast to the 3d lower bound, we need a more sophisticated decomposition into 3 energy scales to improve over the naive bound~\eqref{eq:naivebounds_bdXXX}, which would be $\|d(k)\xi\|^2 \leq C_{\varepsilon} N^{\frac 12 + \varepsilon}$. This improvement is crucial to get an energy error $\ll E^{\RPA}$ in the lower bound.

\begin{lemma} \label{lem:cE_1estimate}
Recall definitions~\eqref{eq:HNconjugation} and~\eqref{eq:XXXtilde_cE1tilde} of $ \cE_1 $ and $\widetilde{\cE}_1$. If $\sum_{k \in \ZZZ^2} |k|^{2-b} \hat{V}(k)^2 < \infty$ for some $b \in (0,1)$, then there exists a $C>0$ such that for all $ \xi \in \cF $,
\begin{equation}
    | \langle \xi, \widetilde{\cE}_1 \xi \rangle |
    \leq C N^{-1 + \frac b4} \langle \xi, \cN^2 \xi \rangle \;.
\end{equation}
Further, if $\sum_{k \in \ZZZ^2}  \hat{V}(k) < \infty$ and $ \xi \in \cF $ belongs to an approximate ground state in the sense of Definition~\ref{def:approxGS}, such that $R \xi$ is an eigenvector of $H_N$, then for any $\varepsilon>0$, there is a constant $C_\varepsilon>0$ such that
\begin{equation} \label{eq:cE_1estimate}
    \|d(k)\xi\|^2
    \leq C_{\varepsilon} N^{\frac 12 - \frac{1}{68} + \varepsilon} \;, \qquad
    | \langle \xi, \cE_1 \xi \rangle |
	\le C_\varepsilon \hbar N^{-\frac{1}{68} + \varepsilon} \;.
\end{equation}
\end{lemma}

\begin{proof}
To bound $\widetilde{\cE}_1$, we use Lemma~\ref{lem:naivebounds_bdXXX} and~\eqref{eq:Vsplit}:
\begin{equation*}
    |\langle\xi, \widetilde{\cE}_1\xi\rangle|
    \leq \frac{C}{N}\sum_{k \in \ZZZ^2_* : |k| < C N^{\frac 12}}\hat{V}(k)\|d(k)\xi\|^2
    \leq \frac{C}{N} \langle \xi, \cN^2 \xi \rangle \sum_{k \in \ZZZ^2_* : |k| < C N^{\frac 12}}\hat{V}(k)
    \leq C N^{-1 + \frac b4} \langle \xi, \cN^2 \xi \rangle \;. 
\end{equation*}
Next, by definition of $\cE_1$~\eqref{eq:HNconjugation} and $d(k)$~\eqref{eq:Lkbkdk}, we have
\begin{align*}
    |\langle\xi,\cE_1\xi\rangle|
    &\leq \frac{C}{N}\sum_{k\in\ZZZ_*^2}\hat{V}(k)\|d(k)\xi\|^2 \;, \qquad
    \|d(k)\xi\|^2
    \leq 2 \|d_1(k)\xi\|^2 + 2 \|d_2(k)\xi\|^2 \;, \\
    d_1(k) &\coloneq \sum_{p\in B_\F^c\cap(B_\F^c+k)} a_{p-k}^*a_p \;, \qquad
    d_2(k) \coloneq \sum_{h\in B_\F\cap(B_\F-k)} a_{h+k}^*a_h \;.
\end{align*}
We consider only $d_1(k)$, as $d_2(k)$ is controlled analogously. Proceeding as in~\eqref{eq:d_estimate_naive}, we get
\begin{align*}
    \Vert d_1(k) \xi \Vert^2
    \leq \underbrace{\Bigg\vert \sum_{p,q\in B_\F^c\cap(B_\F^c+k)}\langle\xi,a_q^*a_{p-k}^*a_{q-k}a_p\xi\rangle \Bigg\vert}_{\eqqcolon A} + \underbrace{\sum_{p\in B_\F^c\cap(B_\F^c+k)}\langle\xi,a_p^*a_p\xi\rangle}_{\eqqcolon B} \;.
\end{align*}
By Lemma~\ref{lemma:A-Priori-Bounds} we readily bound $ B \leq \eva{\xi, \cN \xi} \leq C_\varepsilon N^{\frac 14 + \varepsilon} $.
The term $A$ is treated by introducing $\mu_p>0$ for $p \in \ZZZ^2$ to be fixed later, then applying the Cauchy--Schwarz inequality and the CAR:
\begin{align*}
    A&\leq \Bigg\vert \sum_{p,q\in B_\F^c\cap(B_\F^c+k)}\mu_p^{\frac 12}\mu_p^{-\frac 12}
        \langle\xi,a_q^*a_{p-k}^*a_{q-k}a_p\xi\rangle \Bigg\vert\\
    &\leq \sum_{p,q\in B_\F^c\cap(B_\F^c+k)}\mu_p\langle\xi,a_p^*a_{q-k}^*a_{q-k}a_p\xi\rangle + \sum_{p,q\in B_\F^c\cap(B_\F^c+k)}\mu_p^{-1}\langle\xi,a_{p-k}^*a_{q}^*a_qa_{p-k}\xi\rangle\\
    &\leq \sum_{p\in B_\F^c\cap(B_\F^c+k)}\mu_p\langle\xi,a_p^*\cN a_p\xi\rangle + \sum_{p\in B_\F^c\cap(B_\F^c+k)}\mu_p^{-1}\langle\xi,a_{p-k}^*\cN a_{p-k}\xi\rangle \;,
\end{align*}
using that $(\cN + 1)^\alpha a_p = a_p\cN^\alpha $ for all $ p \in \ZZZ^2$, we get
\begin{equation*}
    A \leq \sum_{p\in B_\F^c\cap(B_\F^c+k)}\mu_p\langle\cN^{\frac{1}{2}}\xi,a_p^*a_p\cN^{\frac{1}{2}}\xi\rangle+\sum_{p\in B_\F^c\cap(B_\F^c+k)}\mu_p^{-1}\langle\cN^{\frac{1}{2}}\xi,a^*_{p-k}a_{p-k}\cN^{\frac{1}{2}}\xi\rangle \;.
\end{equation*}
We now introduce two energy scale cutoffs indexed by $ 0 < \alpha < \delta < \frac 12 $ and split the sum over $ p $ into the two sets
\begin{equation}
    \cS_{k,\delta,\alpha}^\ge
    \coloneqq\left\{ p\in B_\F^c\cap(B_\F^c+k) ~|~
        \min\{e(p),e(p-k)\} \ge \hbar N^{-\delta}, 
        \max\{e(p),e(p-k)\} \ge \hbar N^{-\alpha} \right\} \;, \\
\end{equation}
and $ \cS_{k,\delta,\alpha}^< \coloneq B_\F^c\cap(B_\F^c+k) \setminus \cS_{k,\delta,\alpha}^\ge $. Abbreviating $\phi\coloneqq\cN^{\frac{1}{2}}\xi$, we get
\begin{equation*}
    A\leq\sum_{p\in\cS_{k,\delta,\alpha}^<}
        \big(\mu_p\langle\phi,a_p^*a_p\phi\rangle+\mu_p^{-1}\langle\phi,a_{p-k}^*a_{p-k}\phi\rangle\big)
        +\sum_{p\in\cS_{k,\delta,\alpha}^\ge}
        \big(\mu_p\langle\phi,a_p^*a_p\phi\rangle+\mu_p^{-1}\langle\phi,a_{p-k}^*a_{p-k}\phi\rangle \big) \;.
\end{equation*}
For $p\in\cS_{k,\delta,\alpha}^<$ we choose $\mu_p=1$ and we use $ \Vert a_p \Vert \le 1 $, while for $p\in\cS_{k,\delta,\alpha}^\ge$, we choose $\mu_p = \sqrt{\frac{e(p)}{e(p-k)}}$ in order to get a bound that involves $\HHH_0$. We then apply Lemma~\ref{lemma:A-Priori-Bounds}:
\begin{align}
    A&\leq 2\left|\cS_{k,\delta,\alpha}^<\right|\norm{\phi}^2
        +\sum_{p\in\cS_{k,\delta,\alpha}^\ge}\frac{1}{\sqrt{e(p)e(p-k)}}
        \big(e(p)\langle\phi,a_p^*a_p\phi\rangle+e(p-k)\langle\phi,a^*_{p-k}a_{p-k}\phi\rangle \big)\nonumber\\
    &\leq 2\left|\cS_{k,\delta,\alpha}^<\right| \langle\xi,\cN\xi\rangle
        + C N^{\frac{1}{4}+\frac{\alpha}{2}} N^{\frac{1}{4}+\frac{\delta}{2}}\langle\xi,\cN\HHH_0\xi\rangle\nonumber\\
    &\leq C_\varepsilon \left|\cS_{k,\delta,\alpha}^<\right| 
        N^{\frac{1}{4}+\varepsilon}
        +C_\varepsilon N^{\frac 14 +\frac{\alpha}{2} + \frac{\delta}{2}+\varepsilon} \;.
    \label{eq:bound_on_A}
\end{align}
To estimate $ \left|\cS_{k,\delta,\alpha}^<\right| $, note that there are two ways how $ p $ can be in this set: We can have $ p $ or $ p-k $ in $ \Gcutoff $, or both $ p $ and $ p-k $ in $ \mathcal{G}_\alpha $. Thus,
\begin{equation}
    \left|\cS_{k,\delta,\alpha}^<\right|
    \le 2 |\Gcutoff| + \left|\cS_{k,\alpha}^<\right| \;, \qquad
    \cS_{k,\alpha}^<
    \coloneqq\left\{ p\in B_\F^c\cap(B_\F^c+k) ~|~
        \max\{e(p),e(p-k)\} < \hbar N^{-\alpha} \right\} \;. \\
\end{equation}
The set $\cS_{k,\alpha}^<$ is an intersection of $\mathbb{Z}^2$ with two annuli of thickness $\sim N^{-\alpha}$, which we bound with Lemma~\ref{lem:annulus_intersection} as
\begin{align*}
    \left|\cS_{k,\alpha}^<\right|
    \leq C (N^{\frac 34 - \frac 52 \alpha} + N^{\frac 14 - \frac 12 \alpha}) \;.
\end{align*}
Recalling $|\Gcutoff| \le C_\varepsilon N^{\frac 12 - \delta + \varepsilon}$ from Lemma~\ref{lem:cN_delta_bound}, the bound (\ref{eq:bound_on_A}) becomes
\begin{equation*}
\begin{aligned}
    A
    &\leq C_\varepsilon
        \left(N^{\frac 12 -\delta+\varepsilon}
            +N^{\frac 34 -\frac52 \alpha}
            +N^{\frac 14 - \frac 12 \alpha} \right)
            N^{\frac{1}{4}+\varepsilon}
        + C_\varepsilon N^{\frac 14 + \frac{\alpha}{2} + \frac{\delta}{2} + \varepsilon} \\
    &\leq C_\varepsilon N^{\frac 12 + 2 \varepsilon}
        \left( N^{-(\delta - \frac 14)}
            + N^{-\frac 18 + \frac 52 (\frac 14 - \alpha)}
            + N^{-\frac 18 + \frac 12 (\frac 14 - \alpha)}
            + N^{-\frac 12 (\frac 14 - \alpha) + \frac 12 (\delta - \frac 14)} \right)\;.
\end{aligned}
\end{equation*}
Optimizing $ \delta - \frac 14 = \frac{1}{68} $ and $ \frac 14 - \alpha = \frac{3}{68} $, and re-defining $\varepsilon$, we get
\begin{equation*}
    A\leq C_\varepsilon N^{\frac{1}{2} -\frac{1}{68} +\varepsilon} \;.
\end{equation*}
Together with the above bound $B \leq C_\varepsilon N^{\frac 14 + \varepsilon} $, this concludes the proof.
\end{proof}

As in~\cite{benedikter2021correlation}, we bound $\cE_2$ by an interpolation between $ \cE_1 $ and $ b(k) $.

\begin{lemma} \label{lem:b_bounds}
For any $ k \in \ZZZ^2$, exists a constant $C > 0$, such that for all $\xi \in \cF$,
\begin{equation} \label{eq:b_bounds}
    \|b(k)\xi\|^2\leq C N \log(N) \langle \xi, \HHH_0 \xi \rangle \;,\qquad
    \|b^*(k)\xi\|^2
    \leq C N \log(N) \langle \xi, \HHH_0 \xi \rangle
        + C |k| N^{\frac 12} \;.
\end{equation}
\end{lemma}

\begin{proof}
The proof is analogous to the one of~\cite[Prop.~4.7]{hainzl2020correlation}:
For $ \lambda_{k,p} \coloneq \frac 12 (e(p) + e(p-k)) $, Proposition~\ref{prop:lambda-bound} provides us with the bound $ \sum_{p\in L_k} \lambda_{k,p}^{-1} \le C \hbar^{-2} \log(N) $. Using the Cauchy--Schwarz inequality and then $ \Vert a_{p-k} \Vert, \Vert a_p \Vert \le 1 $, we get
\begin{equation*}
    \|b(k)\xi\|^2
    \le \Big( \sum_{p\in L_k} \lambda_{k,p}^{-1} \Big)
        \Big( \sum_{p\in L_k} (e(p) + e(p-k)) \|a_{p-k}a_p\xi\|^2 \Big)
    \le C N \log(N) \eva{\xi, \HHH_0 \xi} \;.
\end{equation*}
For $b^*(k)$, note that $|\lune| \le C |k| N^{\frac 12}$ and $ \|b^*(k)\xi\|^2\leq |\lune|+\|b(k)\xi\|^2 $.
\end{proof}

\begin{lemma} \label{lem:cE_2estimate}
Let $ \sum_{k \in \ZZZ^2} \hat{V}(k) < \infty $ and $ \xi \in \cF $ belong to an approximate ground state in the sense of Definition~\ref{def:approxGS}, such that $R \xi$ is an eigenvector of $H_N$.
Recall the definition~\eqref{eq:HNconjugation} of $ \cE_2 $. Then for any $\varepsilon>0$, there is a constant $C_\varepsilon>0$ such that
\begin{equation} \label{eq:cE_2estimate}
    |\langle\xi,\cE_2\xi\rangle|
    \leq C_\varepsilon \hbar N^{-\frac{1}{136} + \varepsilon} \;.
\end{equation}
\end{lemma}

\begin{proof}
By the Cauchy--Schwarz inequality, we have
\begin{align*}
    |\langle\xi,\cE_2\xi\rangle|
    &\leq \frac{C}{N}\left|\sum_{k\in\ZZZ^2_*}\hat{V}(k)\left( \langle\xi,d(k)^*b(-k)\xi\rangle+\langle\xi,b(-k)^*d(k)\xi\rangle \right)\right|\\
    &\leq \frac{C}{N} \sum_{k\in\ZZZ^2_*}\hat{V}(k) \|d(k)\xi\|\|b(-k)\xi\| \;.
\end{align*}
Bounding $d(k)$ and $b(-k)$ by Lemmas~\ref{lem:cE_1estimate} and~\ref{lem:b_bounds}, and using~\eqref{eq:Onsager-Bound} $\eva{\xi, \HHH_0 \xi} \leq C N^{-\frac 12}$, we get
    \begin{equation*}
        |\langle\xi,\cE_2\xi\rangle|
        \leq C_{\varepsilon} N^{-1}\sum_{k\in\ZZZ^2_*}\hat{V}(k)
            N^{\frac 14 -\frac{1}{136} + \frac{\varepsilon}{2}}
            N^{\frac 14 + \frac{\varepsilon}{2}}
        \leq C_\varepsilon \hbar N^{-\frac{1}{136} + \varepsilon} \;.
    \end{equation*}
\end{proof}

\section{Patch Construction}
\label{sec:patch-decomposition}

We employ a two-dimensional version of the patch bosonization of~\cite{benedikter2020optimal,benedikter2022bosonization,benedikter2023correlation}, meaning we decompose the region close to the Fermi surface into disjoint regions (``patches'') with suitable properties. We then define collective pair excitation operators in each patch, which behave approximately like bosonic operators, in the sense that they satisfy the canonical commutation relations (CCR) up to a small error.

\subsection{Construction of the Patches}
As in the 3d case~\cite{benedikter2020optimal,benedikter2022bosonization,benedikter2023correlation}, we adopt the algorithmic procedure of \cite{leopardi2006partition}, which allows us to decompose the Fermi surface into equal-area boxes with uniformly bounded diameter. Our decomposition is characterized by two parameters depending on the particle number $ N $:
\begin{itemize}
    \item the patch number $ M=M(N) \in \NNN $, which we assume to be even,
    \item the patch thickness $ R=R(N)>0 $.
\end{itemize}
We will fix the precise dependences in $ N $, later, in order to optimize error bounds.
Since patches should be bigger than the lattice spacing 1, and since radius and circumference of the Fermi surface scale like $ \sim k_{\F} \sim N^{\frac{1}{2}} $, we have the trivial constraints $1 \ll R \ll N^{\frac{1}{2}}$ and $1 \ll M \ll N^{\frac{1}{2}}$.


\noindent

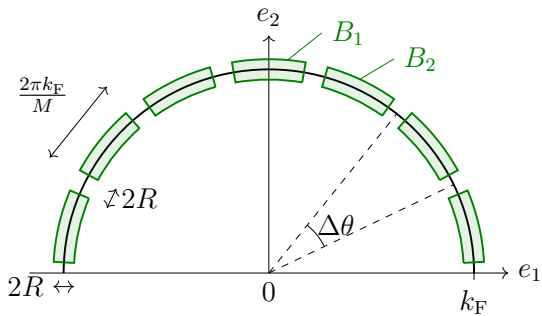
\begin{figure}
    \centering
    \scalebox{0.9}{
\def\R{3}
\def\h{0.05}
\def\shift{3} 
\def\N{7}      

\begin{tikzpicture}

\draw[->] (-3.5,0) -- (3.5,0) node[right] {$e_1$};
\draw[->] (0,0) -- (0,3.5) node[above] {$e_2$};
\node[anchor = north]  at (0,0) {$0$};

\draw ({\R},-0.1) node[anchor = north] {$k_{\F}$} -- ++(0,0.2);

\draw[thick] (-\R,0) arc (180:0:\R);

\foreach \i in {1,...,\N} {
    \pgfmathsetmacro{\angle}{180 - (\i/\N)*180}

    \pgfmathsetmacro{\xi}{\R*cos(\angle+\shift)}
    \pgfmathsetmacro{\yi}{\R*sin(\angle+\shift)}
    
    \pgfmathsetmacro{\xf}{\R*cos(\angle+180/\N-\shift)}
    \pgfmathsetmacro{\yf}{\R*sin(\angle+180/\N-\shift)}

    \pgfmathsetmacro{\ux}{cos(\angle)}
    \pgfmathsetmacro{\uy}{sin(\angle)}
    
    \filldraw[green!50!black, thick, fill opacity = .1] ({(1-\h)*\xi}, {(1-\h)*\yi})
          -- ({(1+\h)*\xi}, {(1+\h)*\yi})  
          arc ({\angle + \shift} : {\angle + 180/\N - \shift} : {(1+\h)*\R})
          -- ({(1-\h)*\xf}, {(1-\h)*\yf})
          arc ({\angle + 180/\N - \shift} : {\angle + \shift} : {(1-\h)*\R}) ;
}

\draw[green!50!black] ({(1+\h)*\R*cos(85)}, {(1+\h)*\R*sin(85)}) -- ++(0.5,0.4) node[anchor = west] {$B_1$};
\draw[green!50!black] ({(1+\h)*\R*cos(65)}, {(1+\h)*\R*sin(65)}) -- ++(0.5,0.3) node[anchor = west] {$B_2$};

\draw[dashed] (0,0) -- ({\R*cos(180/\N)}, {\R*sin(180/\N)}) ;
\draw[dashed] (0,0) -- ({\R*cos(2*180/\N)}, {\R*sin(2*180/\N)}) ;
\draw ({0.9*cos(180/\N)}, {0.9*sin(180/\N)})
    arc ({180/\N} : {2*180/\N} : 0.9) node[anchor = west] {$\Delta \theta$};

\draw[<->] ({-(1+\h)*\R},-0.2) node[anchor = east] {$2R$} -- ({-(1-\h)*\R},-0.2);

\draw[<->] ({\R*cos(5*180/\N)-0.5}, {\R*sin(5*180/\N)+0.4}) -- node[anchor = south east] {$\frac{2 \pi k_{\F}}{M}$} ({\R*cos(6*180/\N)-0.5}, {\R*sin(6*180/\N)+0.4}) ;

\draw[<->] ({\R*cos(6*180/\N - \shift)+0.4}, {\R*sin(6*180/\N - \shift)-0.2}) -- node[anchor = west] {$2R$} ({\R*cos(6*180/\N + \shift)+0.4}, {\R*sin(6*180/\N  + \shift)-0.2}) ;

\end{tikzpicture}}
    \caption{Example for a patch decomposition around the Northern Hemisphere of the Fermi surface. Here, half of all $M=14$ patches $B_{\alpha}$ are shown. The patches have thickness $2R$ and are separated by corridors of size $2R$, where $R$ grows slowly with increasing $N$.}
    \label{fig:patch}
\end{figure}

\noindent
\textbf{Flat patches on the Fermi circle.}~~~In two dimensions, the Fermi surface $ \partial B_{\F} $ is just a circle, which we divide into $M$ arcs, all having an opening angle $\Delta\theta \coloneqq 2\pi M^{-1}$. We choose to put the first arc to be centered at $e_2$, as showed in Figure~\ref{fig:patch}. Then, arc number $ \alpha \in \{1, \ldots, M\} $ is centered at $\theta_\alpha \coloneqq (\alpha-1)\Delta\theta $.
Next, we cut off pieces at the edges of each arc, creating corridors of size $2R$, which requires cutting away an angle $\Delta\theta_{\mathrm{corri}} \coloneqq \frac{2R}{k_{\F}}$. The remaining angle covered by a patch is $ \Delta\widetilde{\theta} \coloneqq \Delta\theta - \Delta\theta_{\mathrm{corri}} $. Denoting with $\hat{\omega}(\theta)$ the point in $\SSS^1$ that forms an angle $\theta$ with respect to $e_2$, we then define the flat patches
\begin{equation}
    P_\alpha
    \coloneqq \{k_{\F} \hat{\omega}(\theta) ~|~ \theta \in (\theta_\alpha - \tfrac{\Delta\widetilde\theta}{2}, \theta_\alpha + \tfrac{\Delta\widetilde\theta}{2})\} \;, \qquad
    \alpha \in \{1,\dots,M\}\;.
\end{equation}
Obviously, $ P_\alpha \subset \partial B_{\F} $ and the patches are disjoint. Further, we require corridors to be much smaller than patches, i.e., $ \Delta\theta_{\mathrm{corri}} \ll \Delta\theta $, which requires $ RM \ll N^{\frac 12} $. It is also clear that, by construction, the following properties hold:
\begin{enumerate}
    \item The diameter of every flat patch is $\textrm{diam}(P_\alpha)=\frac{2\pi k_\F}{M} + O_N(1) $ for all $\alpha \in \{1,\dots,M\}$.
    \item For every $\alpha \in \{1,\dots,M/2\}$, we have the reflection property $P_\alpha=-P_{\alpha+\frac{M}{2}}$.
\end{enumerate}

\noindent
\textbf{Final patches around the Fermi circle.}~~~Finally, we extend the flat patches $P_\alpha$ radially to obtain the patches
\begin{equation}
    B_\alpha \coloneqq \left\{k\in\mathbb{Z}^2~|~ k_\F-R < |k| < k_\F+R\right\}
        \cap \Bigg(\bigcup_{t\in(0,\infty)}tP_\alpha \Bigg)\;.
\end{equation}
As in~\cite{benedikter2020optimal}, $B_\alpha$ inherits the reflection property from $P_\alpha$ and also has a bounded diameter. Similarly, the $B_\alpha$ are also pairwise disjoint and separated by corridors of size $>R$.\\

\noindent
\textbf{Belt cut-off.}~~~As $k \cdot \hat{\omega}_\alpha \to 0$, the number of particle--hole pairs in a patch gets small or even zero, leading to problems with small or zero denominators. We avoid this problem as in~\cite{benedikter2020optimal}: For each $\alpha \in \{ 1,\dots,M \}$, let $\hat{\omega}_\alpha \coloneq \hat{\omega}(\theta_\alpha) \in \SSS^1$ be the vector pointing to the center of the patch $B_\alpha$. Note that $\hat{\omega}_\alpha$ inherits the reflection symmetry: $\hat{\omega}_\alpha = -\hat{\omega}_{\alpha+M/2}$ for $\alpha \in \{1,\dots, \tfrac M2\}$. For $k\in\ZZZ^2_*$, $|k|<R$, define the index set $\cI_k \coloneqq \cI_k^+\cup\cI_k^-$ via
\begin{align}
    \cI_k^+\coloneqq\left\{\alpha=1,\dots,M ~\middle|~ k \cdot \hat{\omega}_\alpha\geq N^{-\delta}\right\}\;, \qquad
    \cI_k^-\coloneqq\left\{\alpha=1,\dots,M ~\middle|~ k \cdot \hat{\omega}_\alpha\leq -N^{-\delta}\right\}\;,
    \label{eq:Ik+set}
\end{align}
where $ \delta > 0 $ is some exponent to be fixed later. In other words, we exclude patches in some thin belt orthogonal to $k$. As motivated below in the proof of Lemma~\ref{lem:normalization_constant}, we impose the constraint
\begin{equation} \label{eq:Mconstraint}
    R N^\delta \ll M \ll R^{-2} N^{\frac 12 - \delta} \;.
\end{equation}
This completes the patch construction, leaving $ (M,R,\delta) $ as the parameters to be optimized. Note that~\eqref{eq:Mconstraint} only makes sense if $\delta \in (0, \frac 14)$. In fact, we will later choose $\delta$ as an arbitrarily small number and $R \sim N^{\delta'}$ for some even smaller $0 < \delta' < \delta$.


\subsection{Patch Operators and Elementary Bounds}

As in~\cite{benedikter2020optimal}, we now split the pair operators $ b^*(k), b(k) $~\eqref{eq:Lkbkdk} among the patches. Given $k\in\ZZZ^2_*$ with $|k|<R$ and given $\alpha \in \cI_k^+$, we define the particle--hole pair creation operator
\begin{equation}
    b_\alpha^*(k)\coloneqq\frac{1}{n_\alpha(k)}\sum_{\substack{p:~p\in B_\F^c\cap B_\alpha\\p-k\in B_\F\cap B_\alpha}}a_p^*a_{p-k}^*\;, \qquad
    n_\alpha(k)^2\coloneqq \sum_{\substack{p:~p\in B_\F^c\cap B_\alpha\\p-k\in B_\F\cap B_\alpha}}1\;.
    \label{def:pair_creation}
\end{equation}
So the normalization constant $n_\alpha(k)^2$ counts the number of particle--hole pairs of relative momentum $k$ in patch $B_\alpha$. A larger $n_\alpha(k)$ corresponds to a better bosonic approximation of the $b^*$--operators.
Moreover, for $\alpha\in\cI_k$, we define
\begin{equation}
    c_\alpha^*(k)\coloneqq
    \begin{cases}
        b_\alpha^*(k) & \textrm{if }\alpha\in\cI_k^+\\
        b_\alpha^*(-k) & \textrm{if }\alpha\in\cI_k^-\\
    \end{cases}
    \;,
    \label{def:pair_operators}
\end{equation}
where $k \in \Gamma^\nor$ with
\begin{equation} \label{eq:Gammanor}
    \Gamma^\nor
    \coloneqq \{k=(k_1,k_2)\in\ZZZ^2_*~|~|k|<R~\textnormal{ and }~ k_2>0 ~\textnormal{ or }~ (k_2=0 \textnormal{ and } k_1>0)\} \;.
\end{equation}
This definition allows for conveniently combining modes associated with $k$ and $-k$, which is possible since $\hat{V}(k) = \hat{V}(-k)$. Next, we compile some bounds for $n_\alpha(k)$ and $c_\alpha(k)$, which are similar or identical to the 3d case~\cite{benedikter2021correlation,benedikter2023correlation}.


\begin{lemma}[Normalization Constant] \label{lem:normalization_constant}
    Assume that $R N^{\delta} \ll M\ll R^{-1}N^{\frac{1}{2}-\delta}$. Then for any $k\in\Gamma^{\textnormal{nor}}$, $\alpha\in\cI_k$, we have
    \begin{equation}
        n_\alpha(k)^2
        =\frac{2\pi k_\F}{M}|k\cdot\hat{\omega}_\alpha|
            \big( 1+\cO(R M^{-1} N^\delta + R M N^{-\frac 12 + \delta}) \big)\;.
        \label{eq:normalization_constant}
    \end{equation}
\end{lemma}

\begin{proof}
Follows by adapting the arguments in~\cite[Section~6]{benedikter2020optimal} and~\cite[Lemma~5.1]{benedikter2023correlation} to the two-dimensional case.\\
Here, the angle between $ k $ and the patch surface is approximated by $ |\hat{k}\cdot\hat{\omega}_\alpha| \ge N^{-\delta} |k|^{-1} $ with $\hat{k} \coloneqq k/|k|$, but actually varies by $ \sim \frac{2 \pi}{M} $ within a patch, leading to a relative error of order $ M^{-1} N^\delta |k| \le M^{-1} N^\delta R $.\\
Further, the error from approximating the projected patch (called $P_\alpha^k$) with its lattice discretization is now $\cO(1)$, and a line intersecting a patch may carry up to $R$ particle--hole pairs, leading to an absolute error of $\cO(R)$, and thus a relative error of $\cO(R M N^{-\frac 12 + \delta})$.
The assumptions on $ M $ are needed for both relative errors to be $ \ll 1 $.
\end{proof}

\begin{lemma}[Approximate CCR] \label{lem:quasi-bosonic-behavior}
    Let $k,\ell\in\Gamma^{\mathrm{nor}}$, $\alpha\in\cI_k$ and $\beta\in\cI_\ell$. Then, the operators $c_\alpha(k),c_\beta^*(k)$ defined above satisfy the following approximate bosonic commutation relations:
    \begin{equation}
    [c_\alpha(k),c_\beta(\ell)]
    =0
    =[c_\alpha^*(k),c_\beta^*(\ell)] \;,\quad[c_\alpha(k),c_\beta^*(\ell)]=\delta_{\alpha,\beta}(\delta_{k,\ell}+\cE_\alpha(k,\ell)) \;,
    \end{equation}
    where the error operator $\cE_\alpha(k,\ell)$ is given by
    \begin{equation} \label{eq:cE_alpha}
        \cE_\alpha(k,\ell)
        \coloneqq - \frac{1}{n_\alpha(k) n_\alpha(\ell)} \Bigg( \sum_{\substack{p:~p \in B_{\F}^c \cap B_\alpha \\ p-\ell, p-k \in B_{\F} \cap B_\alpha}} a^*_{p-\ell} a_{p-k}
		+ \sum_{\substack{h:~h \in B_{\F} \cap B_\alpha \\ h+\ell, h+k \in B_{\F}^c \cap B_\alpha}} a^*_{h+\ell} a_{h+k} \Bigg)\;.
    \end{equation}
    Moreover, $\cE_\alpha(k,\ell)=\cE_\alpha(\ell,k)^*$ commutes with $\cN$ and, for any $\gamma\in\cI_k\cap\cI_\ell$ and $\psi\in\cF$, we have the following bounds
    \begin{equation}
	\sum_{\alpha \in \cI_k \cap \cI_\ell} |\cE_\alpha(k,\ell)|^2
	\le C \big( M N^{-\frac 12 + \delta} \cN \big)^2 \;,
	\quad
	\sum_{\alpha \in \cI_k \cap \cI_\ell} \Vert \cE_\alpha(k,\ell) \psi \Vert
	\le C M^{\frac 32} N^{- \frac 12 + \delta} \Vert \cN \psi \Vert \;.
\end{equation}
\end{lemma}

\begin{proof}
The proof follows as in~\cite[Lemma~5.2]{benedikter2021correlation}, see also~\cite[Lemma~5.2]{benedikter2023correlation}. Note that our factor $N^{- \frac 12 + \delta}$ differs from the 3d case, since also~\eqref{eq:normalization_constant} is different.
\end{proof}

\begin{lemma}[Conversion into Gapped Number Operators] \label{lem:c_conversion} 
Recall the gapped number operator $\cN_\delta$~\eqref{eq:cN_delta}. Let $\delta>0$ be as in the belt cutoff~\eqref{eq:Ik+set} and let $k\in\Gamma^{\nor}$. Then, for $M \gg N^\delta$,
\begin{equation}
    \sum_{\alpha\in\cI_k}c_\alpha^*(k)c_\alpha(k)\leq\cN_\delta \;.
\end{equation}
Moreover, for any $\psi\in\cF$,
\begin{equation}
    \sum_{\alpha \in \cI_k} \Vert c_\alpha(k) \psi \Vert
	\le M^{\frac 12} \Vert \cN_\delta^{\frac 12} \psi \Vert \;,
	\qquad
	\sum_{\alpha \in \cI_k} \Vert c^*_\alpha(k) \psi \Vert
	\le M^{\frac 12} \Vert (\cN_\delta + M)^{\frac 12} \psi \Vert \;,
\end{equation}
and for any $f\in\ell^2(\cI_k)$,
\begin{equation}
\begin{aligned}
    \Big\Vert \sum_{\alpha \in \cI_k} f_\alpha c_\alpha(k) \psi \Big\Vert
	\le \Vert f \Vert_{\ell^2} \Vert \cN_\delta^{\frac 12} \psi \Vert \;, \qquad
    \Big\Vert \sum_{\alpha \in \cI_k} f_\alpha c^*_\alpha(k) \psi \Big\Vert
	\le \Vert f \Vert_{\ell^2} \Vert (\cN_\delta + 1)^{\frac 12} \psi \Vert \;.
\end{aligned}
\end{equation}
\end{lemma}

\begin{proof}
The proof is analogous to~\cite[Lemma~5.3]{benedikter2021correlation}, where $M \gg N^\delta$ is needed to ensure $\mathrm{diam}(B_\alpha) \leq C N^{\frac 12} M^{-1} \ll N^{\frac 12 - \delta}$, so finally $e(p)+e(p-k) \geq c N^{-\frac 12 - \delta}$.
\end{proof}

As in~\cite[(5.11)]{benedikter2021correlation}, for $ g: \ZZZ^2 \times \ZZZ^2 \to \RRR $ we define the weighted pair operators
\begin{equation} \label{eq:calphag}
    c_\alpha^g(k)
    \coloneqq \frac{1}{n_\alpha(k)} \sum_{\substack{p:~p\in B_\F^c\cap B_\alpha\\p \mp k\in B_\F\cap B_\alpha}} g(p,k) a_{p \mp k} a_p \qquad \textrm{for } \alpha \in \cI_k^\pm \;.
\end{equation}

\begin{lemma}[Weighted Pair Operators] \label{lem:weightedpairoperators}
Recall~\eqref{eq:cN_delta} $\cN_\delta$ and let $\delta>0$ as in~\eqref{eq:Ik+set}. Then, for all $k\in\Gamma^{\textnormal{nor}}$ and $\psi\in\cF$, we have
\begin{equation}
	\sum_{\alpha \in \cI_k} \Vert c_\alpha^g(k) \psi \Vert
	\le C M^{\frac 12} \Vert g \Vert_\infty
		\Vert \cN_\delta^{\frac 12} \psi \Vert \;,
	\qquad
	\sum_{\alpha \in \cI_k} \Vert c_\alpha^g(k)^* \psi \Vert
	\le C M^{\frac 12} \Vert g \Vert_\infty
		\Vert (\cN_\delta + M)^{\frac 12} \psi \Vert \;,
\end{equation}
and for all $f\in\ell^2(\cI_k)$ also
\begin{equation}
\begin{aligned}
    \Big\Vert \sum_{\alpha \in \cI_k} f_\alpha c_\alpha^g(k) \psi \Big\Vert
	&\le \Vert g \Vert_\infty
		\Vert f \Vert_{\ell^2}
		\Vert \cN_\delta^{\frac 12} \psi \Vert \;, \\
	\Big\Vert \sum_{\alpha \in \cI_k} f_\alpha c_\alpha^g(k)^* \psi \Big\Vert
	&\le \Vert g \Vert_\infty
		\Vert f \Vert_{\ell^2}
		\Vert (\cN_\delta + 1)^{\frac 12} \psi \Vert \;.
\end{aligned}
\end{equation}
\end{lemma}

\begin{proof}
The proof is analogous to~\cite[Lemma~5.4]{benedikter2021correlation}.
\end{proof}

\section{Pseudo-Bosonic Bogoliubov Transformations}
\label{sec:pseudobosonic}

Recall $Q_\B$ from the correlation Hamiltonian~\eqref{eq:HNconjugation}. For the upper bound on $E_{\GS}$, in analogy to $\widetilde{\XXX}$ and $\widetilde{\cE}_1$~\eqref{eq:XXXtilde_cE1tilde}, we define the low-momentum restriction
\begin{equation} \label{eq:QBtilde}
    \widetilde{Q}_{\B} 
    \coloneq \frac{1}{(2 \pi)^2 N} \sum_{k \in \ZZZ^2_* : |k| < C N^{\frac 12}} \hat{V}(k) \left( b^*(k) b(k) + \frac 12 \big( b^*(k) b^*(-k) + b(-k) b(k) \big) \right) \;.
\end{equation}
As in~\cite{benedikter2023correlation}, we approximate $Q_{\B}$ and $\widetilde{Q}_{\B}$ using the pairs operators introduced in \eqref{def:pair_operators} by
\begin{align}
    Q_\B^R
    &\coloneq\frac{1}{(2\pi)^2N}\sum_{k\in\Gamma^\nor}\hat{V}(k)
        \Bigg( \sum_{\alpha,\beta\in\cI_k^+}n_\alpha(k)n_\beta(k)c_\alpha^*(k)c_\beta(k) 
        + \sum_{\alpha,\beta\in\cI_k^-}n_\alpha(k)n_\beta(k)c_\alpha^*(k)c_\beta(k) \nonumber\\
    &\quad
        +\sum_{\alpha\in\cI^+_k,\beta\in\cI^-_k}n_\alpha(k)n_\beta(k)c_\alpha^*(k)c_\beta^*(k) + \sum_{\alpha\in\cI_k^-,\beta\in\cI^+_k}n_\alpha(k)n_\beta(k)c_\alpha(k)c_\beta(k)\Bigg)\;.
\end{align}
This approximation amounts to neglecting the contributions from corridors and patches close to the equator, whose smallness is ensured by the following lemma.

\begin{lemma}
Recall the definitions~\eqref{eq:HNconjugation} and~\eqref{eq:QBtilde} of $ Q_{\B} $ and $\widetilde{Q}_{\B}$. If $\sum_{k \in \ZZZ^2} |k|^{2-b} \hat{V}(k)^2 < \infty$ for some $b \in (0,1)$, then there exist $C, C_\varepsilon>0$ such that for all $ \xi \in \cF $,
\begin{equation} \label{eq:Q_QR_bound_mod}
\begin{aligned}
	&|\langle \xi, (\widetilde{Q}_{\B} - Q_{\B}^R) \xi \rangle| \\
	&\le C \hbar R^2 M^{\frac 32} N^{-\frac 14 + \frac b4 + \frac \delta2} \sup_{\lambda \in [0,1]} \langle T_\lambda \xi, (\cN+1)^3 T_\lambda \xi \rangle \\
    &\quad + C_\varepsilon \hbar R^{1 + \frac b2} N^\varepsilon 
        \big( N^{\frac 14} \langle \xi, \HHH_0 \xi \rangle^{\frac 12} + 1 \big)
        \big( N^{\frac 14} \langle \xi, \HHH_0 \xi \rangle^{\frac 12}
        \big( N^{-\frac 18}
        + N^{-\frac{\delta}{2}}
        + R M^{\frac 12} N^{-\frac 14 + \frac{\delta}{2}} \big) + R N^{-\frac 14} \big) \;.
\end{aligned}
\end{equation}
Further, if $\sum_{k\in\ZZZ^2}|k| \hat{V}(k) <\infty$ and $\xi$ belongs to an approximate ground state in the sense of Definition~\ref{def:approxGS}, then
\begin{equation} \label{eq:Q_QR_bound}
	|\langle \xi, (Q_{\B} - Q_{\B}^R) \xi \rangle|
	\le C_\varepsilon \hbar N^\varepsilon
        \big( R^{-\frac 12}
        + N^{-\frac 18}
        + N^{-\frac{\delta}{2}}
        + R M^{\frac 12} N^{-\frac 14 + \frac{\delta}{2}} \big) \;.
\end{equation}
\label{lem:Q_QR_bound}
\end{lemma}

\begin{proof}
First, note that for $|k| \geq R$, we have $k \notin \Gamma^{\nor}$ (compare~\eqref{eq:Gammanor}), so $k$ does not contribute to $Q_{\B}^R$. Thus,
\begin{equation*}
\begin{aligned}
    |\langle \xi, (Q_{\B} - Q_{\B}^R) \xi \rangle|
	&\le \frac{C}{N} \sum_{k \in \ZZZ^2_* : |k| \geq R} \hat{V}(k)
        \Big( \langle \xi ,b^*(k) b(k) \xi \rangle
        + |\langle \xi ,b(k) b(-k) \xi \rangle| \Big) \\
    &\quad + \frac{C}{N} \sum_{k \in \ZZZ^2_* : |k| < R} \hat{V}(k)
        \big( \Vert b(k) \xi \Vert + \Vert b^*(-k) \xi \Vert \big) 
        \Vert r^R(k) \xi \Vert \;,
\end{aligned}
\end{equation*}
where the bosonization error for $|k| < R$ is defined as
\begin{equation}
    r^R(k) \coloneq b(k) - \sum_{\alpha \in \cI_k^+} n_\alpha(k) c_\alpha(k) \;.
\end{equation}
The same formula is true for $|\langle \xi, (\widetilde{Q}_{\B} - Q_{\B}^R) \xi \rangle|$ with the additional constraint $|k| < C N^{\frac 12}$.\\

\textbf{Case $|k| \ge R$.} For $\widetilde{Q}_{\B}$, we follow the same steps as in~\cite[Lemma~A.3]{benedikter2023correlation}, using Lemma~\ref{lem:naivebounds_bdXXX} and $\sum_{|k| < C N^{\frac 12}} \hat{V}(k) |k|^{\frac 12} \leq N^{\frac{1+b}{4}}$ (compare~\eqref{eq:Vsplit}), which yields
\begin{equation}
\begin{aligned}
    &\frac{C}{N} \sum_{k \in \ZZZ^2_* : R \leq |k| < C N^{\frac 12}} \hat{V}(k)
        \Big( \langle \xi ,b^*(k) b(k) \xi \rangle
        + |\langle \xi ,b(k) b(-k) \xi \rangle| \Big) \\
    &\leq C R^2 M^{\frac 32} N^{-\frac 34 + \frac b4 + \frac \delta2} 
        \sup_{\lambda \in [0,1]} \langle T_\lambda \xi, (\cN+1)^3 T_\lambda \xi \rangle \;.
\end{aligned}
\end{equation}
For $Q_{\B}$, as in~\cite[Lemma~6.1]{benedikter2023correlation}, we use Lemmas~\ref{lem:b_bounds} and~\ref{lemma:Onsager-Bound} to get
\begin{equation*}
\begin{aligned}
    &\frac{1}{N} \sum_{|k| \ge R} \hat{V}(k)
        \left( \Vert b(k) \xi \Vert + \Vert b^*(-k) \xi \Vert \right) 
        \Vert r^R(k) \xi \Vert
    \leq \frac{C}{N} \sum_{|k| \ge R} \hat{V}(k)
        \left(N^{\frac 12} + |k| N^{\frac 12}\right)^{\frac 12}
        \left(N^{\frac 12}\right)^{\frac 12} \log(N)\\
    &\leq C N^{-\frac 12} \log(N) \sum_{|k| > R} \hat{V}(k) |k|^{\frac 12} R^{-\frac 12}
    \leq C \hbar R^{-\frac 12} \log(N) \;.
\end{aligned}
\end{equation*}

\textbf{Case $|k|<R$.} Here, the errors for $\widetilde{Q}_{\B}$ and $Q_{\B}$ are identical and proportional to
\begin{equation*}
    \frac{1}{N} \sum_{k \in \ZZZ^2_* : |k| < R} \hat{V}(k)
        \left( \Vert b(k) \xi \Vert + \Vert b^*(-k) \xi \Vert \right) 
        \Vert r^R(k) \xi \Vert
\end{equation*}
We write
\begin{equation*}
    \Vert r^R(k) \xi \Vert
    \le \sum_{p \in Y_k} \Vert a_{p-k} a_p \xi \Vert
    + \sum_{p \in U_k \setminus Y_k} \Vert a_{p-k} a_p \xi \Vert \;,
\end{equation*}
where $ U_k $ tracks all non-bosonized pairs and $ Y_k $ in particular such excluded by the belt cutoff\footnote{Note that if $p$ is excluded by the belt cutoff~\eqref{eq:Ik+set}, then $p \in B_{\alpha}$ or $p-k \in B_\alpha$ for some $\alpha \notin \cI_k$. We then write $\hbar^{-2} \lambda_{k,p} = (p \cdot k - \tfrac{|k|^2}{2}) \leq |p \cdot k - k_{\F}(k \cdot \hat{\omega}_\alpha)| + k_{\F} |k \cdot \hat{\omega}_\alpha|$, where by the belt cutoff $k_{\F} |k \cdot \hat{\omega}_\alpha| \leq \pi^{-\frac 12} N^{\frac 12 - \delta}$. From the patch geometry, $|p \cdot k - k_{\F}(k \cdot \hat{\omega}_\alpha)| \leq |k| |p-k_{\F} \hat{\omega}_\alpha| \leq |k|(R + C k_{\F} M^{-1})$, so with $M \gg R N^\delta \gg |k| N^\delta$, we conclude $\hbar^{-2} \lambda_{k,p} \leq \pi^{-\frac 12} (1+o_N(1)) |k|^{-1} N^{\frac 12-\delta}$. This implies~\eqref{eq:cutoffcondition} for $N$ large enough, so the set $Y_k$ indeed covers all $p$ excluded by the belt cutoff.} in~\eqref{eq:Ik+set}:
\begin{equation}
    U_k \coloneq L_k \setminus \bigcup_{\alpha = 1}^M (B_\alpha \cap (B_\alpha + k)) \;, \qquad
    Y_k \coloneq \{ p \in U_k ~|~ \lambda_{k,p} \le \hbar N^{-\delta} \} \;,
\end{equation}
with excitation energy $\lambda_{k,p} \coloneq \frac 12 \hbar^2 (|p|^2 - |p-k|^2)$. Note that, introducing $\hat{k} \coloneq k/|k|$ and the distance $s(p) \coloneq (p \cdot \hat{k} - \frac{|k|}{2})$ of $p$ in $k$-direction to the tip of the lune $B_{\F}^c \cap (B_{\F}+k)$, see Figure~\ref{fig:Fig_Lemma_6_1}, we have $\lambda_{k,p} = \hbar^2 |k| s(p)$. The cutoff in $Y_k$ then amounts to
\begin{equation} \label{eq:cutoffcondition}
    \lambda_{k,p} \le \hbar N^{-\delta}
    \quad \Leftrightarrow \quad
    s(p) \leq |k|^{-1} N^{\frac 12 - \delta} \;.
\end{equation}
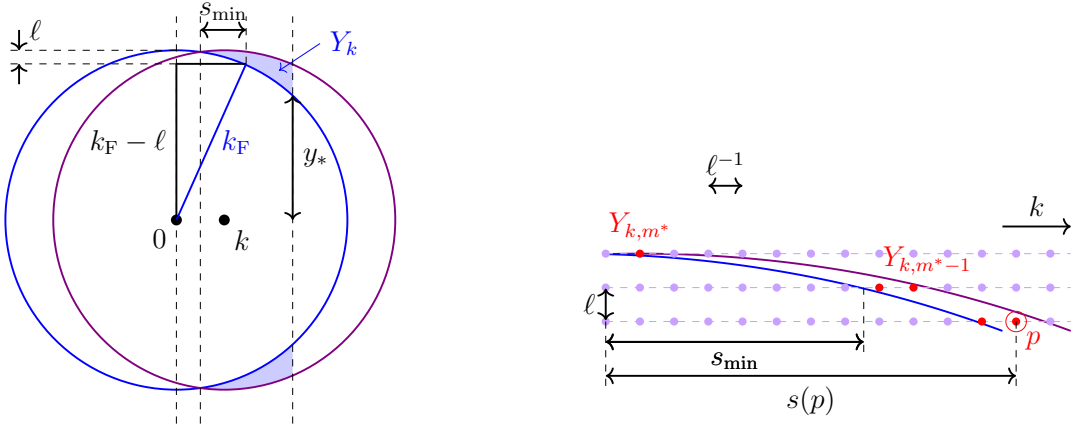
\begin{figure}
    \centering
    \scalebox{0.9}{\def\kF{2.5}      
\def\k{0.7}     
\def\cutoff{1.7}  
\def\l{0.2}     
\begin{tikzpicture}

\fill[blue, opacity=0.2]
    plot[domain={\cutoff}:{\k/2}, samples = 50] ({\x}, {sqrt(\kF^2 - (\x)^2)}) -- 
    plot[domain={\k/2}:{\cutoff}, samples = 50] ({\x}, {sqrt(\kF^2 - (\x-\k)^2)});
\fill[blue, opacity=0.2]
    plot[domain={\cutoff}:{\k/2}, samples = 50] ({\x}, {-sqrt(\kF^2 - (\x)^2)}) -- 
    plot[domain={\k/2}:{\cutoff}, samples = 50] ({\x}, {-sqrt(\kF^2 - (\x-\k)^2)});

\fill (0,0) circle (0.08) node[anchor = north east]{$0$};
\fill ({\k},0) circle (0.08) node[anchor = north west]{$k$};
\draw[thick, blue] (0,0) circle ({\kF});
\draw[thick, red!50!blue] ({\k},0) circle ({\kF});

\draw[dashed] (0,{-\kF-0.5}) -- ++(0,{2*\kF+1});
\draw[dashed] ({\k/2},{-\kF-0.5}) -- ++(0,{2*\kF+1});
\draw[dashed] ({\cutoff},{-\kF-0.5}) -- ++(0,{2*\kF+1});
\draw[thick] (0,0) --node[anchor = east]{$k_{\F} - \ell$} (0,{\kF-\l});
\draw[thick] (0,{\kF-\l}) -- ++({sqrt(2*\kF*\l + \l^2)},0);
\draw[thick, blue] (0,0) --node[anchor = west]{$k_{\F}$} ({sqrt(2*\kF*\l + \l^2)},{\kF-\l});

\draw[dashed] ({sqrt(2*\kF*\l + \l^2)},{\kF-\l}) -- ({sqrt(2*\kF*\l + \l^2)},{\kF+0.5});
\draw[<->, thick] ({\k/2},{\kF+0.3}) --node[anchor = south]{$s_{\min}$} ({sqrt(2*\kF*\l + \l^2)},{\kF+0.3});

\draw[dashed] (0,{\kF}) -- ++(-2.5,0);
\draw[dashed] (0,{\kF-\l}) -- ++(-2.5,0);
\draw[<-, thick] (-2.3,{\kF}) -- ++(0,0.3) node[anchor = west]{$\ell$};
\draw[<-, thick] (-2.3,{\kF-\l}) -- ++(0,-0.3);

\draw[<-, blue] (1.5,2.2) -- ++(0.6,0.4) node[anchor = west]{$Y_k$};

\draw[<->, thick] ({\cutoff},0) --node[anchor = west]{$y_*$} ({\cutoff},{sqrt(\kF^2 - \cutoff^2)});

\end{tikzpicture}}
    \hspace{2em}
    \scalebox{0.9}{\definecolor{lilla}{RGB}{200,160,255}

\def\kF{18}			
\def\k{1}
\def\d{\k+6.3}
\def\Yshift{-8.5}   
\def\Ymfirst{-\Yshift+1.6}  
\def\l{0.5}         
\def\spacing{0.5}
\def\nlines{3}

\begin{tikzpicture}[scale=1]

\draw[thick,blue!50!red] plot[domain={\k/2}:{\d}, samples=200] ({\x}, {sqrt(\kF^2 - (\x - \k)^2) - \kF});
\draw[thick,blue] plot[domain={\k/2}:{\d - \k}, samples=200] ({\x}, {sqrt(\kF^2 - (\x)^2) - \kF});

\pgfmathsetmacro{\kFSqr}{\kF*\kF}    
\foreach \i in {0,...,\numexpr\nlines-1} {
    \draw[dashed, lilla, thin] (\k/2, -\i*\l) -- (\d, -\i*\l);
    \foreach \j in {1,...,14} {
        \pgfmathsetmacro{\rSqr}{(\spacing*\j)*(\spacing*\j) + (-\i*\l+\kF)*(-\i*\l+\kF)}
        \pgfmathsetmacro{\sSqr}{(\spacing*\j-\k)*(\spacing*\j-\k) + (-\i*\l+\kF)*(-\i*\l+\kF)-0.1}
        \ifdim \rSqr pt < \kFSqr pt
            \fill[lilla]  (\spacing*\j, -\i*\l) circle (0.06);
        \else
            \ifdim \sSqr pt < \kFSqr pt
                \fill[red]  (\spacing*\j, -\i*\l) circle (0.06);
            \else
                \fill[lilla]  (\spacing*\j, -\i*\l) circle (0.06);
            \fi
        \fi
    }
}

\pgfmathsetmacro{\yTop}{-(\nlines-2)*\l}
\pgfmathsetmacro{\yBottom}{-(\nlines-1)*\l}
\pgfmathsetmacro{\yLast}{-(\nlines-1)*\l}

\draw[->, thick] ({\d-\k}, 0.4) -- ({\d},0.4) node[midway, above] {$k$};
\draw[<->, thick] (\k/2, +\yTop*1.03) -- (\k/2, +\yBottom*0.98) node[midway, left] {$\ell$};

\node[red] at (1,0.4) {$Y_{k,m^*}$};
\node[red] at (5.2,-0.1) {$Y_{k,m^*-1}$};

\draw[<->, thick] ({\k/2},-1.3) -- ({sqrt(2*\kF*\l + \l^2)},-1.3) node[midway, below]{$s_{\min}$};
\draw[<->, thick] ({\k/2},-1.3) -- ({sqrt(2*\kF*\l + \l^2)},-1.3) node[midway, below]{$s_{\min}$};

\draw[dashed] ({\k/2},-1) -- ++(0,-1);
\draw[dashed] ({sqrt(2*\kF*\l + \l^2)},-0.5) -- ++(0,-1);
\draw[dashed] ({\k/2+12*\l},-1) -- ++(0,-1);
\draw[<->, thick] ({\k/2},-1.8) -- ++({12*\l},0) node[midway, below]{$s(p)$};
\draw[red] ({\k/2+12*\l},-1) circle (0.15) node[anchor = north west]{$p$};

\draw[<->, thick] ({4*\l},1) --node[anchor = south]{$\ell^{-1}$} ++({\l},0);

\end{tikzpicture}}
    \caption{\textbf{Left}: Depiction of the set $Y_k$ and geometric considerations for determining $s_{\min}$. \textbf{Right}: We decompose the set $Y_k$ into planes $Y_{k,m}$ parallel to $k$. For a point $p \in Y_{k,m}$, the pair excitation energy is then $\lambda_{k,p} = \hbar^2 |k| s(p)$, which is conveniently lower-bounded for $|m| \neq m^*$ using $s(p) \geq s_{\min}$.}
    \label{fig:Fig_Lemma_6_1}
\end{figure}
We now decompose $Y_k$ into planes parallel to $\hat{k}$, i.e., perpendicular to $\hat{k}^\perp \coloneq \left( \begin{smallmatrix} 0 & -1 \\ 1 & 0 \end{smallmatrix} \right) \hat{k}$, where the distance of two planes is $\ell=|k|^{-1} \gcd(k_1,k_2) \leq1$:
\begin{equation}
    Y_{k,m} \coloneq \{ p \in Y_k ~|~ p \cdot \hat{k}^\perp = m \ell \} \;, \qquad m \in \mathbb{Z} \;,
\end{equation}
see Figure~\ref{fig:Fig_Lemma_6_1}. Here, $ Y_{k,m} $ can only be non-empty if
\begin{equation}
    m_* \leq |m| \leq m^* \;, \qquad
    m_* \coloneq \inf \{ m \in \mathbb{N} ~|~ m \ell \geq y_* \} \;, \qquad
    m^* \coloneq \sup \{ m \in \mathbb{N} ~|~ m \ell < k_{\F} \} \;,
\end{equation}
where $ y_* > 0 $ is defined such that (compare~\eqref{eq:cutoffcondition})
\begin{equation} \label{eq:ystar}
    k_{\F}^2 - y_*^2 = \big( \tfrac{|k|}{2} + |k|^{-1} N^{\frac 12 - \delta} \big)^2
    \quad \Rightarrow \quad
    y_* \ge c N^{\frac 12} \;, \qquad 
    k_{\F}^2 - y_*^2 
    \leq N^{1-2\delta} + R^2 \;.
\end{equation}
Here, $ R^2 \ll N^{1-2\delta} $ since $ \delta \in (0,\frac 14) $ and $R$ will be chosen as a sufficiently small power of $N$. We now consider the cases $m \in \{-m^*, m^*\}$ and $|m| \le m^* -1$, separately: Let $\tilde Y_k \coloneq Y_k \setminus (Y_{k,-m^*} \cup Y_{k,m^*})$. Then\footnote{Note that Proposition~\ref{prop:lambda-bound} already provides us with the bound $\sum_{p \in \tilde Y_k} \lambda_{k,p}^{-1} \leq C \hbar^{-2} \log(N)$. However, this is insufficient for this lemma: We need $\sum_{p \in \tilde Y_k} \lambda_{k,p}^{-1} = \hbar^{-2} o_N(1)$.}
\begin{equation} \label{eq:Yk_bound_split}
    \sum_{p \in Y_k} \Vert a_{p-k} a_p \xi \Vert
    \le |Y_{k,-m^*}|
        + |Y_{k,m^*}| 
        + \Bigg( \sum_{p \in \tilde Y_k} \lambda_{k,p}^{-1} \Bigg)^{\frac 12}\langle\xi,\HHH_0\xi\rangle^{\frac 12}\;.
\end{equation}
The spacing of points on each plane is $\ell^{-1}$, so the number of points per plane is bounded by $|Y_{k,m}| \leq \ell |k| + 1 \leq |k| + 1 \leq 2 R$, which is in particular true for $m \in \{-m^*, m^*\}$.\\
For $|m| \le m^* - 1$, note that since $|k| < R$, the lune is very thin, which results in a lower bound on $s(p)$, (i.e., an energy gap), see Figure~\ref{fig:Fig_Lemma_6_1}:
\begin{equation} \label{eq:smin}
    s(p)
    \geq s_{\min}
    = \sqrt{k_{\F}^2 - (k_{\F} - \ell)^2} - \frac{|k|}{2} \quad \Rightarrow \quad
    s_{\min}
    \geq C \sqrt{k_{\F} \ell} - R
    \geq C N^{\frac 14} |k|^{-\frac 12} \;.
\end{equation}
Likewise, $p \in Y_{k,m}$ satisfies $s(p) \geq \sqrt{k_{\F}^2 - (m \ell)^2} - \frac{|k|}{2}$, and since $|k| < R \ll s_{\min} \leq s(p)$, we have $s(p) \ge c \sqrt{k_{\F}^2 - (m \ell)^2} $. Since every plane accommodates $\le (|k| + 1) $ points, we have
\begin{equation*}
\begin{aligned}
    &\sum_{p\in \tilde{Y}_k} \lambda_{k,p}^{-1}
    = \sum_{p\in \tilde{Y}_k} \frac{1}{\hbar^2 |k| s(p)}
    \leq C \hbar^{-2} |k|^{-1} (|k| + 1) \sum_{m_* \leq |m| \le m^*-1} \left( k_{\F}^2 - (m \ell)^2 \right)^{-\frac 12} \\
    &\leq C \hbar^{-2} \Bigg( s_{\min}^{-1}
        + \int_{m_*}^{m^*-1} \left( k_{\F}^2 - (m \ell)^2 \right)^{-\frac 12} \di m \Bigg)
    \leq C \hbar^{-2} \Bigg( s_{\min}^{-1}
        + \ell^{-1} \int_{y_*}^{k_\F - \ell} \left( k_{\F}^2 - y^2 \right)^{-\frac 12} \di y \Bigg) \\
    &\leq C \hbar^{-2} \Bigg( s_{\min}^{-1}
        + |k| \Big[ \arctan\Big( \tfrac{y}{\sqrt{k_{\F}^2 - y^2}} \Big) \Big]_{y=y_*}^{k_{\F}-\ell} \Bigg) \;.
\end{aligned}
\end{equation*}
Using $\arctan(\tfrac 1x) = \tfrac{\pi}{2} - \arctan(x)$, where $\arctan(x) = x + \cO(x^3)$, we conclude
\begin{equation*}
    \sum_{p\in \tilde{Y}_k} \lambda_{k,p}^{-1}
    \leq C \hbar^{-2} \Big( s_{\min}^{-1} + |k| \tfrac{\sqrt{k_{\F}^2 - y_*^2}}{y_*} \Big)
    \leq C N |k| (N^{-\frac 14} +  N^{-\delta}) \;,
\end{equation*}
where we bounded $s_{\min}^{-1}$ via~\eqref{eq:smin}, and $\sqrt{k_{\F}^2 - y_*^2}$ and $y_*$ via~\eqref{eq:ystar}. Then, recalling that $|Y_{k,m^*}|, |Y_{k,-m^*}| \leq C R$,~\eqref{eq:Yk_bound_split} becomes
\begin{equation} \label{eq:Yk_bound}
    \sum_{p \in Y_k} \Vert a_{p-k} a_p \xi \Vert
    \le C |k|^{\frac 12} N^{\frac 12} \langle \xi, \HHH_0 \xi \rangle^{\frac 12} (N^{-\frac 18} + N^{-\frac{\delta}{2}}) + C R \;.
\end{equation}
For $ p \in U_k \setminus Y_k $, we exploit the even larger spectral gap $ e(p) + e(p-k) = 2 \lambda_{k,p} > 2 \hbar N^{-\delta} $:
\begin{equation*}
\begin{aligned}
    &\sum_{p \in U_k \setminus Y_k} \Vert a_{p-k} a_p \xi \Vert
    \le C \Bigg( \sum_{p \in U_k \setminus Y_k} \hbar^{-1} N^{\delta} (e(p) + e(p-k)) \Vert a_{p-k} a_p \xi \Vert^2 \Bigg)^{\frac 12} |U_k \setminus Y_k|^{\frac 12} \\
    &\le C R M^{\frac 12} N^{\frac 14 + \frac{\delta}{2}} \eva{\xi, \HHH_0 \xi}^{\frac 12} \;,
\end{aligned}
\end{equation*}
where we used $ |U_k \setminus Y_k| \le C M R^2 $, as this set consists of $ M $ corridors of area $ \le C R^2 $. Putting all bounds together, we obtain
\begin{equation}
    \Vert r^R(k) \xi \Vert
    \le C |k|^{\frac 12} N^{\frac 12} \langle \xi, \HHH_0 \xi \rangle^{\frac 12}
        \big(N^{-\frac 18}
        + N^{-\frac{\delta}{2}}
        + R M^{\frac 12} N^{-\frac 14 + \frac{\delta}{2}} \big) + C R\;.
\end{equation}
Combining this with the bounds~\eqref{eq:b_bounds} on $ \Vert b^\sharp(k) \xi \Vert^2 \le C N \log(N) \langle \xi, \HHH_0 \xi \rangle + C |k| N^{\frac{1}{2}} $ with $\sharp \in \{ *, \cdot \}$, and estimating $ \sum_{|k| < R} \hat{V}(k) |k| \leq C R^{\frac{2+b}{2}} $ as in~\eqref{eq:Vsplit} yields~\eqref{eq:Q_QR_bound_mod}.\\
For~\eqref{eq:Q_QR_bound}, we directly estimate $\sum_k \hat{V}(k) |k| < \infty$ and use that by Lemma~\ref{lemma:Onsager-Bound}, for approximate ground states, $\langle \xi, \HHH_0 \xi \rangle \leq C N^{-\frac 12}$.
\end{proof}

By contrast, the kinetic energy $\HHH_0$ cannot be directly expressed in terms of the quasi-bosonic pair operators $c$ and $c^*$. However, as in~\cite{benedikter2023correlation}, it behaves with respect to commutators as
\begin{equation}
    [\HHH_0,c^*_\alpha(k)]
    = \frac{1}{n_\alpha(k)}\sum_{p\in L_k \cap B_\alpha}(e(p)+e(p-k))a_p^*a_{p-k}^*
    \simeq 2\hbar\kappa|k\cdot\hat{\omega}_\alpha|c_\alpha^*(k)\;,
\end{equation}
where we linearized the dispersion relation as
$e(p)+e(p-k)\simeq2\hbar\kappa|k\cdot\hat{\omega}_\alpha|$ with $\kappa = \pi^{-\frac 12}$, so $k_{\F} = \kappa N^{\frac 12}$. Thus, heuristically,
\begin{equation} \label{eq:DDD_B}
    \HHH_0 
    \simeq 2\hbar\kappa\sum_{k\in\Gamma^\nor} \sum_{\alpha \in \cI_k} |k\cdot\hat{\omega}_\alpha|c_\alpha^*(k)c_\alpha(k)
    \eqqcolon \DDD_\B\;.
\end{equation}
We can then approximate $ (\HHH_0 + Q_\B) $ as follows: Define $ g(k) \in \RRR $, $ u(k), v(k) \in \RRR^{|\cI_k^+|} $, and $ d(k), b(k) \in \RRR^{|\cI_k^+| \times |\cI_k^+|} $ via
\begin{align}
    g(k)
    &\coloneqq \frac{1}{2(2 \pi)^2} \hat{V}(k) \;, \qquad
    u_\alpha(k)
    \coloneqq |\hat{k}\cdot\hat{\omega}_\alpha|^{\frac{1}{2}} \;, \qquad 
    v_\alpha(k)
    \coloneqq k_{\F}^{-\frac 12} |k|^{- \frac 12} n_\alpha(k)~\textrm{ for } \alpha \in \cI_k^+ \;, \nonumber \\
    d(k)
    &\coloneqq \textrm{diag}\{u_\alpha(k)^2~|~ \alpha \in \cI_k ^+\}\;, \qquad
    b(k)
    \coloneqq g(k) |v(k) \rangle \langle v(k)|\;, \label{eq:guvdb}
\end{align}
with $\hat{k} \coloneqq k/|k|$, as well as the $|\cI_k|\times|\cI_k|$ real symmetric matrices
\begin{equation} \label{eq:DWW}
    D(k)\coloneqq\begin{pmatrix}
        d(k) & 0\\
        0 & d(k)
    \end{pmatrix}\;,\qquad
    W(k)\coloneqq\begin{pmatrix}
        b(k) & 0\\
        0 & b(k)
    \end{pmatrix}\;,\qquad
    \widetilde{W}(k)\coloneqq\begin{pmatrix}
        0 & b(k)\\
        b(k) & 0
    \end{pmatrix}\;.
\end{equation}
Then, with the effective Hamiltonian
\begin{equation}
    h_\eff(k)\coloneqq \sum_{\alpha,\beta\in\cI_k}\left( (D(k)+W(k))_{\alpha,\beta}c^*_\alpha(k)c_\beta(k) + \frac{1}{2}\widetilde{W}(k)_{\alpha,\beta}(c_\alpha^*(k)c_\beta^*(k) + c_\beta(k)c_\alpha(k)) \right)\;,
    \label{eq:h_eff}
\end{equation}
we have
\begin{equation} \label{eq:h_eff_2}
    \HHH_0 + Q_\B
    \simeq \DDD_\B + Q_\B^R = \sum_{k\in\Gamma^\nor} 2\hbar\kappa|k|h_\eff(k)\;.
\end{equation}

To simplify the notation, we will often drop the explicit dependence on $k$. In analogy to~\cite[Section~7]{benedikter2023correlation}, we now introduce the two approximately bosonic Bogoliubov transformations in order to approximately diagonalize the quasi-bosonic Hamiltonian. Let us briefly recall the construction strategy: We write $ h_\eff $ in block matrix form
\begin{align}
    &h_\eff \simeq \HHH-\frac{1}{2}\tr(D+W)\;, \qquad
    \HHH \coloneq \frac{1}{2}((c^*)^T,c^T)
    \begin{pmatrix}
        D+W & \widetilde{W}\\
        \widetilde{W} & D+W
    \end{pmatrix}
    \begin{pmatrix}
        c\\
        c^*
    \end{pmatrix}\;.
    \label{eq:HHH_quadratic}
\end{align}
Introducing the $|\cI_k|\times|\cI_k|$ matrices 
\begin{equation} \label{eq:K}
\begin{aligned}
    E
    &\coloneqq \left((D+W-\widetilde{W})^{1/2}(D+W+\widetilde{W})(D+W-\widetilde{W})^{1/2}\right)^{1/2}\;,\\
    S_1
    &\coloneqq (D+W-\widetilde{W})^{1/2}E^{-1/2}\;,\\
    K
    &\coloneqq \log |S_1^T|\;,
\end{aligned}
\end{equation}
with polar decomposition $S_1=O|S_1|$, we can diagonalize
\begin{align}
    \begin{pmatrix}
        D+W & \widetilde{W}\\
        \widetilde{W} & D+W
    \end{pmatrix}
    &=
    \begin{pmatrix}
    \cosh K & \sinh K\\
    \sinh K & \cosh K
    \end{pmatrix}
    \begin{pmatrix}
        O & 0\\
        0 & O
    \end{pmatrix}
    \begin{pmatrix}
        E & 0\\
        0 & E
    \end{pmatrix}
    \nonumber\\
    &\quad \times
    \begin{pmatrix}
        O^T & 0\\
        0 & O^T
    \end{pmatrix}
    \begin{pmatrix}
    \cosh K & \sinh K\\
    \sinh K & \cosh K
    \end{pmatrix}
    \;.
\end{align}
As in~\cite[Sect.~9]{benedikter2023correlation} and~\cite{christiansen2023gell}, this first transformation will turn out insufficient for a lower bound: The approximation $ \HHH_0 \simeq \DDD_\B $ produces a contribution $ - \DDD_\B $ in the Hamiltonian, which could only be compensated if we had $ E \ge D $. But this is generally not true. We therefore adopt the second quasi-bosonic Bogoliubov transformation from~\cite[Sect.~7]{benedikter2023correlation} which renders a diagonal block matrix $ \widetilde{P} \ge D $: We introduce the $|\cI_k|\times|\cI_k|$ matrix
$U\coloneqq \tfrac{1}{\sqrt{2}}
\big( \begin{smallmatrix}
    1 & 1\\
    1 & -1
\end{smallmatrix} \big)$, 
and we notice that
\begin{equation}
    U^T(D+W+\widetilde{W})U =
    \begin{pmatrix}
        d+2b & 0\\
        0 & d
    \end{pmatrix}
    \;,\qquad
    U^T(D+W-\widetilde{W})U = 
    \begin{pmatrix}
        d & 0\\
        0 & d+2b
    \end{pmatrix}
    \;,
\end{equation}
\begin{align} \label{eq:UEU}
    U^TEU &=
    \begin{pmatrix}
        \left(d^{1/2}(d+2b)d^{1/2}\right)^{1/2} & 0\\
        0 & \left((d+2b)^{1/2}d(d+2b)^{1/2}\right)^{1/2}
    \end{pmatrix}\nonumber\\
    &=
    \begin{pmatrix}
        (X^*X)^{1/2} & 0\\
        0 & (XX^*)^{1/2}
    \end{pmatrix}
    =
    \begin{pmatrix}
        P & 0\\
        0 & APA^T
    \end{pmatrix}
    \;,
\end{align}
where $X\coloneqq (d+2b)^{1/2}d^{1/2} = AP$, with $A$ orthogonal and $P\coloneqq (X^*X)^{1/2}$ characterizing the polar decomposition of $X$. Finally, setting
\begin{equation}
    \widetilde{O}\coloneqq U
    \begin{pmatrix}
        1 & 0 \\
        0 & A
    \end{pmatrix}
    U^T\;, \qquad
    \widetilde{P}\coloneqq
    \begin{pmatrix}
        P & 0\\
        0 & P
    \end{pmatrix}\;,
\end{equation}
and noticing that $E=\widetilde{O}\widetilde{P}\widetilde{O}^T$, we conclude the final diagonalization
\begin{align}
    \begin{pmatrix}
        D+W & \widetilde{W}\\
        \widetilde{W} & D+W
    \end{pmatrix}
    &=
    \begin{pmatrix}
    \cosh K & \sinh K\\
    \sinh K & \cosh K
    \end{pmatrix}
    \begin{pmatrix}
        O & 0\\
        0 & O
    \end{pmatrix}
    \begin{pmatrix}
        \widetilde{O} & 0\\
        0 & \widetilde{O}
    \end{pmatrix}
    \begin{pmatrix}
        \widetilde{P} & 0\\
        0 & \widetilde{P}
    \end{pmatrix}
    \nonumber\\
    & \quad \times
    \begin{pmatrix}
        \widetilde{O}^T & 0\\
        0 & \widetilde{O}^T
    \end{pmatrix}
    \begin{pmatrix}
        O^T & 0\\
        0 & O^T
    \end{pmatrix}
    \begin{pmatrix}
    \cosh K & \sinh K\\
    \sinh K & \cosh K
    \end{pmatrix}
    \;.
\end{align}
Therefore, the following unitary transformations would diagonalize $ (\DDD_\B + Q_\B^R) $, if it was exactly bosonic:
\begin{equation} \label{eq:bogoliubov}
\begin{aligned}
    T &\coloneq T_1 \;, \qquad 
    &T_\lambda &\coloneq \exp\left( \frac{\lambda}{2}\sum_{k\in\Gamma^\nor}\sum_{\alpha,\beta\in\cI_k} K(k)_{\alpha,\beta}c_\alpha^*(k)c_\beta^*(k) - \textrm{h.c.} \right)\;, \qquad 
    &\lambda \in \RRR \;, \\
    Z &\coloneq Z_1 \;, \qquad 
    &Z_\lambda &\coloneq\exp\left( \lambda \sum_{k\in\Gamma^\nor}\sum_{\alpha,\beta\in\cI_k} L(k)_{\alpha,\beta}c_\alpha^*(k)c_\beta(k)\right)\;, \qquad 
    &\lambda \in \RRR \;,
\end{aligned}
\end{equation}
where $ K(k) $ was defined in~\eqref{eq:K} and $L(k)$ is given by
\begin{equation}
    L(k)\coloneqq \log\left(O(k)\widetilde{O}(k)\right)\;.
\end{equation}
The unitary diagonalization then follows as
\begin{align}
    Z^*T^*\HHH TZ
    \simeq \frac{1}{2}\sum_{\alpha,\beta\in\cI_k} \widetilde{P}_{\alpha,\beta}c_\alpha^*(k)c_\beta(k)+\frac{1}{2}\tr\widetilde{P}
    \geq\DDD_\B+\frac{1}{2}\tr E \;,
\end{align}
where the last line is obtained noticing that $\widetilde{P}\geq D$ and that $\tr\widetilde{P}\geq \tr E$. Together with~\eqref{eq:HHH_quadratic}, the diagonalization thus produces an energy of
\begin{equation} \label{eq:Ecorr_trace_formula}
    \hbar \kappa \sum_{k \in \Gamma^{\nor}} |k| \tr \left( E(k) - D(k) - W(k) \right) 
    \simeq E^{\RPA} \;.
\end{equation}
We will make this approximation rigorous. To do so, we start compiling some estimates on the transformations $T_\lambda$ and $Z_\lambda$.

\begin{lemma}[Bogoliubov Kernel for $T$] \label{lem:K-kernel}
For $k\in\Gamma^\nor$, $K(k)$ is a real symmetric matrix, and there is a $C>0$ such that for all $k \in \Gamma^{\nor}$ and $\alpha, \beta \in \cI_k$, we have
\begin{equation}
    |K(k)_{\alpha,\beta}|
    \leq C\frac{\hat{V}(k)}{M} \;, \qquad
    \Vert K(k) \Vert_{\HS}
	\le C \hat{V}(k) \;.
\end{equation}
\end{lemma}

\begin{proof}
The proof is a straightforward adaptation of~\cite[Lemma~2.5]{benedikter2022bosonization} to two dimensions. The only modification is that for us, $ v_\alpha = k_{\F}^{-\frac 12} |k|^{-\frac 12} n_\alpha(k) $, where $ \hbar^{\frac 12} $ replaces the factor of $ \hbar $ in~\cite[(2.15)]{benedikter2022bosonization}. With $u_\alpha \coloneqq |\hat{k}\cdot\hat{\omega}_\alpha(k)|^{\frac 12}$, we then still have $v_\alpha\simeq Cu_\alpha M^{-\frac{1}{2}}$, as in three dimensions.
The rest of the proof then follows as in~\cite[Lemma~2.5]{benedikter2022bosonization}.
\end{proof}

\begin{lemma} \label{lem:L-kernel}
Let $\delta>0$, $M>0$ and $R>0$ be defined as in Section \ref{sec:patch-decomposition}. Then there exists a $C>0$ such that for any $k\in\Gamma^\nor$ we have
\begin{equation}
\begin{aligned}
	&\Vert L(k) \Vert_{\HS}
	\le C \hat{V}(k) (1 + \delta \log(N))^2 \;,\\
	&\Vert L(k) \Vert_{\op}
	\le C \hat{V}(k) (1 + \delta \log(N)) \;.
\end{aligned}
\end{equation}
\end{lemma}
\begin{proof}
The claim follows by the same strategy as in \cite[Lemma~7.2]{benedikter2023correlation}.
\end{proof}

\begin{lemma}[Stability of number operators] \label{lem:gronwall}
Let $ \sum_{k \in \ZZZ^2} |k|^{2-b} \hat{V}(k)^2 < \infty $ for some $b \in (0,1)$ and recall $T_\lambda$, $Z_\lambda$ from~\eqref{eq:bogoliubov}. Then for any $m\in\NNN$ there exists a constant $C_m>0$ such that for all $\lambda\in[-1,1]$ we have
\begin{equation}
	T_\lambda^* \cN^m T_\lambda
	\le C_m \exp(C_m R^{\frac b2}) (\cN + 1)^m \;.
\end{equation}
Further, if $\sum_{k \in \ZZZ^2} \hat{V}(k) < \infty$, then we have the bounds
\begin{align}
	T_\lambda^* \cN^m T_\lambda
	&\le C_m (\cN + 1)^m \;,
	\qquad
	&T_\lambda^* \cN_\delta \cN^m T_\lambda
	&\le C_m (\cN_\delta + 1) (\cN + 1)^m \;, \\
	Z_\lambda^* \cN^m Z_\lambda
	&= \cN^m \;,
	\qquad
	&Z_\lambda^* \cN_\delta \cN^m Z_\lambda
	&\le C_m N^{C_m \delta} \cN_\delta \cN^m \;.
\end{align}
\end{lemma}
\begin{proof}
The proof for $\sum_{k \in \ZZZ^2} \hat{V}(k) < \infty$ is the same as in~\cite[Lemma~7.2]{benedikter2021correlation} and~\cite[Lemma~7.3]{benedikter2023correlation}.
For $ \sum_{k \in \ZZZ^2} |k|^{2-b} \hat{V}(k)^2 < \infty $, we adopt the modification of~\cite[Lemma~A.2]{benedikter2023correlation} to~\cite[Proposition~4.6]{benedikter2020optimal} with (compare~\eqref{eq:Vsplit})
\begin{equation*}
    \sum_{k \in \Gamma^{\nor}} \Vert K(k) \Vert_{\HS}
    \leq C \sum_{k \in \ZZZ^2_* : |k| < R } \hat{V}(k)
    \leq C R^{\frac b2} \;.
\end{equation*}
\end{proof}

The next lemma tells us that the operators $T$ and $Z$ behave like bosonic Bogoliubov transformations, up to errors $\fE$ and $\fF$. For this reason, we will call them pseudo-bosonic (or quasi-bosonic) Bogoliubov transformations.

\begin{lemma}[Approximate bosonic Bogoliubov transformations]
Let $\sum_{k \in \ZZZ^2} |k| \hat{V}(k)$. Then, for any $\lambda \in [-1,1]$, $k\in\Gamma^\nor$, $\gamma\in\cI_k$ it holds that
\begin{equation}
\begin{aligned}
	T_\lambda^* c_\gamma(k) T_\lambda
	&= \tilde{c}_\gamma(\lambda,k)
		+ \fE_\gamma(\lambda,k) \;, \\
	\tilde{c}_\gamma(\lambda,k)
	&\coloneq \sum_{\alpha \in \cI_k} \cosh(\lambda K(k))_{\alpha, \gamma} c_\alpha(k)
		+ \sum_{\alpha \in \cI_k} \sinh(\lambda K(k))_{\alpha, \gamma} c^*_\alpha(k) \;,
\end{aligned}
\label{eq:c_tilde_def}
\end{equation}
\begin{equation}
	Z_\lambda^* c_\gamma(k) Z_\lambda
	= \sum_{\beta \in \cI_k} \exp(\lambda L(k))_{\gamma, \beta} c_\beta(k)
		+ \fF_\gamma(\lambda,k) \;,
\end{equation}
with error estimates
\begin{equation} \label{eq:fEfF_bounds}
\begin{aligned}
	\sum_{\gamma \in \cI_k} \Vert \fE_\gamma(\lambda,k) \psi \Vert
	&\le C M N^{-\frac 12 + \delta}
        \eva{\psi, (\cN_\delta + M) (\cN + 1)^2 \psi}^{\frac 12} \;, \\
	\sum_{\gamma \in \cI_k} \Vert \fF_\gamma(\lambda,k) \psi \Vert
	&\le C M^{\frac 32} N^{-\frac 12 + C \delta}
        \eva{\psi, \cN_\delta \cN^2 \psi }^{\frac 12} \;,
\end{aligned}
\end{equation}
for $\psi \in \cF$. If only $\sum_{k \in \ZZZ^2} |k|^{2-b} \hat{V}(k)^2 < \infty$ is known for some $b \in (0,1)$, then we still have
\begin{equation} \label{eq:fEfF_bounds_mod}
	\sum_{\gamma \in \cI_k} \Vert \fE_\gamma(\lambda,k) \psi \Vert
	\le C M N^{-\frac 12 + \delta}
        e^{C R^\frac b2} \eva{\psi, (\cN_\delta + M) (\cN + 1)^2 \psi}^{\frac 12} \;.
\end{equation}
In both cases, for $ \tilde{c}_\gamma(k) \coloneq \tilde{c}_\gamma(1,k) $ the following bounds hold true
\begin{equation}
    \|\tilde{c}_\alpha(k)\psi\|,\|\tilde{c}^*_\alpha(k)\psi\| 
    \le C \|(\cN_\delta+1)^\frac 12\psi \|\;.
\end{equation}
\label{lem:approximate_bosonic_bogoliubov_transformation}
\end{lemma}
\begin{proof}
The proof of~\eqref{eq:c_tilde_def}--\eqref{eq:fEfF_bounds_mod} is analogous to \cite[Lemma~7.4]{benedikter2023correlation} and \cite[Lemma~7.5]{benedikter2023correlation}, using Lemmas~\ref{lem:quasi-bosonic-behavior},~\ref{lem:c_conversion}, and~\ref{lem:gronwall}.
Consider now the operator $\tilde{c}_\alpha(k)$, and notice that
\begin{equation*}
    \big|\cosh (K(k))_{\alpha,\beta}-\delta_{\alpha,\beta}\big|+\big|\sinh (K(k))_{\alpha,\beta}\big|\leq \frac CM\;.
\end{equation*}
By Lemma~\ref{lem:c_conversion} we have $\|c_\alpha(k)\psi\|\le \|\cN_\delta^\frac 12\psi\|$ and $\sum_{\alpha \in \cI_k} \|c_\alpha^*(k)\psi\|\le M^{\frac 12}\|(\cN_\delta+M)^\frac 12 \psi \|$, hence
\begin{align*}
    \|\tilde{c}_\alpha(k)\psi \| &\le \sum_{\beta\in\cI_k}\|\cosh(K(k))_{\alpha,\beta}c_\beta(k)\psi\|+\sum_{\beta\in\cI_k}\|\sinh(K(k))_{\alpha,\beta}c^*_\beta(k)\psi\|\\
    &\le \sum_{\beta\in\cI_k}\delta_{\alpha,\beta}\|c_\beta(k)\psi\| + \frac CM \sum_{\beta\in\cI_k}\|c_\beta(k)\psi\|+\frac CM \sum_{\beta\in\cI_k}\|c^*_\beta(k)\psi\|\\
    &\le C \|(\cN_\delta+1)^\frac 12\psi \|\;,
\end{align*}
where we used $|\cosh(K(k))_{\alpha,\beta}|\le |\cosh(K(k))_{\alpha,\beta} - \delta_{\alpha,\beta}| + \delta_{\alpha,\beta}$. An analogous argument applies to $\|\tilde{c}^*_\alpha(k)\psi\|$.
\end{proof}

\section{Linearizing the Kinetic Energy}
\label{sec:linearizing_kinetic}

Thanks to Lemma \ref{lem:approximate_bosonic_bogoliubov_transformation}, we can now make the heuristic argument of the last section rigorous.

\begin{lemma}[Kinetic commutators]
For all $k\in\Gamma^\nor$ and all $\alpha\in\cI_k$, we have
\begin{equation}
\begin{aligned}
    [\HHH_0, c^*_\alpha(k)]
    &= 2 \kappa \hbar | k \cdot \hat{\omega}_\alpha | c^*_\alpha(k)
    	+ \hbar \fE^{\lin}_\alpha(k)^* \;, \\
    [\DDD_{\B}, c^*_\alpha(k)]
    &= 2 \kappa \hbar | k \cdot \hat{\omega}_\alpha | c^*_\alpha(k)
    	+ \hbar \fE^{\B}_\alpha(k)^* \;,
\end{aligned}
\end{equation}
where there exists a $C>0$ such that\footnote{Note that in the analogous 3d bound on $\fE^{\B}_\alpha(k)$ in~\cite[(8.2)]{benedikter2023correlation}, a $|k|$ is missing on the r.~h.~s., which does, however, not influence the correctness of the proof.} for all $f\in\ell^2(\cI_k)$ and $\psi\in\cF$,
\begin{equation}
\begin{aligned}
    \sum_{\alpha \in \cI_k} \Vert \fE^{\lin}_\alpha(k) \psi \Vert
    &\le C |k| M^{-\frac 12} \langle \psi, \cN_\delta \psi \rangle^{\frac 12}\;, \\
    \Bigg\Vert \sum_{\alpha \in \cI_k} f_\alpha \fE^{\lin}_\alpha(k) \psi \Bigg\Vert
    &\le C |k| M^{-1} \Vert f \Vert_2 \langle \psi, \cN_\delta \psi \rangle^{\frac 12} \;, \\
    \sum_{\alpha \in \cI_k} \Vert \fE^{\B}_\alpha(k) \psi \Vert
    &\le C |k| R^2 M^{\frac 32} N^{-\frac 12 + \delta} \langle \psi, \cN_\delta \cN^2 \psi \rangle^{\frac 12} \;.
\end{aligned}
\end{equation}
\end{lemma}

\begin{proof}
As in the proof of~\cite[Lemma~8.2]{benedikter2021correlation}, we obtain $ \fE^{\lin}_\alpha(k) = c^g_\alpha(g) $ for some $ g $, which is bounded with $ \textnormal{diam}(B_\alpha) \le C N^{\frac 12} M^{-1} $ as $ \norm{g}_{\ell^\infty} \le C |k| M^{-1} $. Then, we apply Lemma~\ref{lem:weightedpairoperators} to obtain the bounds on $ \fE^{\lin}_\alpha(k) $.
The bound for $ \fE^{\B}_\alpha(k) $ follows as in~\cite[(8.6)]{benedikter2021correlation}, with Lemmas~\ref{lem:quasi-bosonic-behavior} and~\ref{lem:c_conversion}, as well as $\sum_{\ell \in \Gamma^{\nor}} 1 \leq C R^2$.
\end{proof}

\begin{lemma}[Approximate Bogoliubov invariance of $\HHH_0-\DDD_\B$]
Let $\sum_{k \in \ZZZ^2} |k| \hat{V}(k) < \infty$. Then, there exists a constant $C>0$ such that for all $\psi\in\cF$ we have
\begin{equation}
\begin{aligned}
    &| \langle T \psi, (\HHH_0 - \DDD_{\B}) T \psi \rangle
    	- \langle \psi, (\HHH_0 - \DDD_{\B}) \psi \rangle | \\
    &\le C \hbar \Big( M^{-1}
    	\langle \psi, (\cN_\delta + 1) \psi \rangle
    	+ R^2 M N^{-\frac 12 + \delta}
    	\langle \psi, (\cN_\delta + 1) \psi \rangle^{\frac 12}
    	\langle \psi, (\cN_\delta + 1) (\cN + 1)^2 \psi \rangle^{\frac 12} \Big) \;, \\
    &| \langle Z \psi, (\HHH_0 - \DDD_{\B}) Z \psi \rangle
    	- \langle \psi, (\HHH_0 - \DDD_{\B}) \psi \rangle | \\
    &\le C \hbar \Big( M^{-1} N^{C \delta}
    	\langle \psi, \cN_\delta \psi \rangle
    	+ R^2 M^{\frac 32} N^{-\frac 12 + C \delta}
    	\langle \psi, \cN_\delta \psi \rangle^{\frac 12}
    	\langle \psi, \cN_\delta \cN^2 \psi \rangle^{\frac 12} \Big) \;. \\
\end{aligned}
\end{equation}
If only $\sum_{k \in \ZZZ^2} |k|^{2-b} \hat{V}(k)^2 < \infty$ is known for some $b \in (0,1)$, then we still have
\begin{equation}
\begin{aligned}
    &| \langle T \psi, (\HHH_0 - \DDD_{\B}) T \psi \rangle
    	- \langle \psi, (\HHH_0 - \DDD_{\B}) \psi \rangle | \\
    &\le C \hbar e^{C R^{\frac b2}} \Big( M^{-1}
    	\langle \psi, (\cN_\delta + 1) \psi \rangle
    	+ M N^{-\frac 12 + \delta}
    	\langle \psi, (\cN_\delta + 1) \psi \rangle^{\frac 12}
    	\langle \psi, (\cN_\delta + 1) (\cN + 1)^2 \psi \rangle^{\frac 12} \Big) \;.
\end{aligned}
\end{equation}
\label{lem:Approximate_Bogoliubov_Invariance}
\end{lemma}

\begin{proof}
The proof follows by the same arguments as in~ \cite[Lemma~8.1]{benedikter2021correlation} and~\cite[Lemma~8.3]{benedikter2023correlation}.
\end{proof}

\section{Proof of Theorem~\ref{thm:main}}
\label{sec:proof_main}

We divide the proof into three steps. The first part is devoted to the computation of the correlation energy, while the other two parts concern respectively the lower and the upper bound on the ground state energy.

\subsection{Evaluation of the Trace}

Recall~\eqref{eq:Ecorr_trace_formula} that the diagonalization of the effective pseudo-bosonic operator resulted in an approximate correlation energy
\begin{equation*}
    E^{\RPA} \simeq \hbar \kappa \sum_{k \in \Gamma^{\nor}} |k| \tr \left( E(k) - D(k) - W(k) \right) \;.
\end{equation*}
The next lemma will make this approximation rigorous.

\begin{lemma}[Evaluation of the trace] \label{lem:trace_evaluation}
Recall the definitions~\eqref{eq:DWW} and~\eqref{eq:K} of $ E(k) $, $ D(k) $ and $ W(k) $, as well as~\eqref{eq:ERPA} of $ E^{\RPA} $.  If $\sum_{k \in \ZZZ^2} |k|^{2-b} \hat{V}(k)^2 < \infty$ for some $b \in (0,1)$, then there exists some $ C > 0 $ such that
\begin{align} \label{eq:trace_evaluation_mod}
    &\left| E^{\RPA}
    - \hbar \kappa \sum_{k \in \Gamma^{\nor}} |k|\tr \left( E(k) - D(k) - W(k) \right)
        \right| \nonumber \\
    &\le C \hbar R^{\frac{2+b}{2}}
        \big( N^{-\frac{\delta}{2}} 
        + R^{\frac 12} M^{-\frac 12} N^{\frac{\delta}{2}}
        + R^{\frac 12} M^{\frac 12} N^{-\frac 14 + \frac{\delta}{2}} \big) 
        + C \hbar R^{b-1}\;.
\end{align}
If even $\sum_{k \in \ZZZ^2} |k| \hat{V}(k) < \infty$, then
\begin{align} \label{eq:trace_evaluation}
    &\left| E^{\RPA}
    - \hbar \kappa \sum_{k \in \Gamma^{\nor}} |k|\tr \left( E(k) - D(k) - W(k) \right)
        \right| \nonumber \\
    &\le C \hbar \big( R^{-1}
        + N^{-\frac{\delta}{2}} 
        + R^{\frac 12} M^{-\frac 12} N^{\frac{\delta}{2}} 
        + R^{\frac 12} M^{\frac 12} N^{-\frac 14 + \frac{\delta}{2}} \big) \;.
\end{align}
In either case,
\begin{equation} \label{eq:ERPA_bound}
    |E^{\RPA}| \le C \hbar \;.
\end{equation}
\end{lemma}

\begin{proof}
By the same computation as in the 3d case~\cite[(5.14)]{benedikter2020optimal}, recalling the definition~\eqref{eq:guvdb} of $g(k)$, $v_\alpha(k)$ and $u_\alpha(k)$, we arrive at
\begin{equation}
\begin{aligned}
    \tr \left( E(k) - D(k) - W(k) \right)
    &= \frac{2}{\pi}\int_0^\infty \log\big( 1 + Q_k(\lambda) \big) \di \lambda
        - 2 g(k) \sum_{\alpha \in \cI_k^+} v_\alpha^2(k) \;, \\
    \textnormal{where} \quad
    Q_k(\lambda)
    &\coloneq 2g(k) \sum_{\alpha \in \cI_k^+} \frac{u_\alpha^2(k)v_\alpha^2(k)}{u_\alpha^4(k)+\lambda^2} \;.
\end{aligned}
\end{equation}
Now notice that $ \frac{2}{\pi} \int_0^\infty Q_k(\lambda) \di \lambda = 2 g(k) \sum_{\alpha \in \cI_k^+} v_\alpha^2(k) $, which allows writing
\begin{equation} \label{eq:trace_EDW_formula_1}
    \hbar \kappa \sum_{k \in \Gamma^{\nor}} |k| \tr \left( E(k) - D(k) - W(k) \right)
    = \hbar \kappa \sum_{k \in \Gamma^{\nor}} |k| \frac{2}{\pi}\int_0^\infty F\big(Q_k(\lambda) \big) \di \lambda\;,
\end{equation}
where $ F(x) = \log(1+x) - x $. On the other hand, using the symmetry $ k \mapsto -k $, we have
\begin{equation} \label{eq:trace_EDW_formula_2}
\begin{aligned}
    E^{\RPA}
    &= E^{\RPA}_< + E^{\RPA}_\ge \;, \qquad
    &E^{\RPA}_< 
    &\coloneqq \hbar \kappa \sum_{k \in \Gamma^{\nor}} |k| \frac{2}{\pi}\int_0^\infty F\big(\widetilde{Q}_k(\lambda) \big) \di \lambda\;, \\
    \widetilde{Q}_k(\lambda)
    &\coloneqq 2\pi g(k)\left( 1- \frac{\lambda}{\sqrt{\lambda^2+1}} \right) \;, \qquad
    &E^{\RPA}_\ge
    &\coloneqq \hbar \kappa \sum_{k \in \ZZZ^2 : |k| \ge R} \frac{|k|}{\pi}\int_0^\infty F\big(\widetilde{Q}_k(\lambda) \big) \di \lambda  \;.
\end{aligned}
\end{equation}
So it remains to estimate $E^{\RPA}_\ge$ and the error from replacing $ Q_k(\lambda) $ by $ \widetilde{Q}_k(\lambda) $. We start with the latter. By \eqref{eq:normalization_constant} we have
\begin{equation} \label{eq:Qk_replacement}
    Q_k(\lambda) = 2g(k) \sum_{\alpha \in \cI_k^+} \sigma(p_\alpha)\frac{u_\alpha(k)^4}{u_\alpha(k)^4+\lambda^2} \big( 1 + \mathcal{O} \big(RM^{-1}N^\delta + R M N^{-\frac 12 + \delta} \big) \big)\;,
\end{equation}
Where $\sigma(p_\alpha)=2\pi/M$ is the measure of the unit circle arc $p_\alpha$, centered at $ \hat\omega_\alpha $. In order to evaluate the sum, we define $ \theta_\alpha $ as the angle between $ \hat{k} $ and $ \hat\omega_\alpha $, so $\cos(\theta_\alpha) = u_\alpha(k)^2$. Then, since the partition is diameter-bounded as $\sup_{\hat\omega \in p_\alpha}|\theta(\hat\omega) -\theta_\alpha|\leq C/M$, where $\theta(\hat\omega) = \theta$ is the angle between $\hat{k}$ and $\hat\omega$. Then,
\begin{align*}
    &\left|\int_{p_\alpha}\frac{\cos^2(\theta(\hat\omega))}{\cos^2(\theta(\hat\omega)) + \lambda^2}\di \sigma(\hat\omega) - \sigma(p_\alpha)\frac{\cos^2(\theta_\alpha)}{\cos^2(\theta_\alpha)+\lambda^2} \right|
    \leq \int_{p_\alpha}\left| \frac{\cos^2(\theta(\hat\omega))}{\cos^2(\theta(\hat\omega)) + \lambda^2} - \frac{\cos^2(\theta_\alpha)}{\cos^2(\theta_\alpha)+\lambda^2} \right|\di \sigma \\
    &\leq \sup_{\hat\omega \in p_\alpha}
        \left|\frac{d}{d\theta}\frac{\cos^2(\theta)}{\cos^2(\theta) + \lambda^2}\right|\frac{C}{M}\sigma(p_\alpha)
    \le \sup_{\hat\omega \in p_\alpha}
        \left|\frac{2\lambda^2\cos(\theta)\sin(\theta)}{(\cos^2(\theta)+\lambda^2)^2}\right| \frac{C}{M^2}
    \leq \sup_{\hat\omega \in p_\alpha}
        \frac{C M^{-2}}{|\cos(\theta(\hat\omega))|} \;.
\end{align*}
Since $\alpha \in \cI_k^+$ (compare~\eqref{eq:Ik+set}), we have $\cos(\theta_\alpha) > N^{-\delta}|k|^{-1} \ge N^{-\delta} R^{-1} $, and as we assumed~\eqref{eq:Mconstraint} $ M \gg R N^\delta $, then also $\cos(\theta(\hat\omega))>N^{-\delta}R^{-1}$ for any $\hat\omega \in p_\alpha$, so
\begin{equation*}
    \left|\int_{p_\alpha}\frac{\cos^2(\theta)}{\cos^2(\theta) + \lambda^2}\di \sigma - \sigma(p_\alpha)\frac{\cos^2(\theta_\alpha)}{\cos^2(\theta_\alpha)+\lambda^2} \right|
    \leq C R M^{-2} N^\delta \;.
\end{equation*}
Therefore, we conclude that
\begin{equation}
    \left|\int_{\mathbb{S}^1_{\textrm{reduced}}}\frac{\cos^2(\theta)}{\cos^2(\theta) + \lambda^2}\di \sigma - \sum_{\alpha \in \cI_k^+} \sigma(p_\alpha)\frac{\cos^2(\theta_\alpha)}{\cos^2(\theta_\alpha)+\lambda^2} \right|
    \leq C R M^{-1} N^\delta \;,
    \label{eq:bound1}
\end{equation}
where $\mathbb{S}^1_{\textrm{reduced}} \coloneq \bigcup_{\alpha \in \cI_k^+} p_\alpha$ is the unit half-circle, excluding the belt of width $N^{-\delta} |k|^{-1}$. Moreover, since $\cos^2(\theta)(\cos^2(\theta) +\lambda^2)^{-1}\leq1$, we can compare with the integral over the whole unit half-circle, called $\SSS^1_{\textrm{half}}$
\begin{equation}
    \left|\int_{\mathbb{S}^1_{\textrm{half}}}\frac{\cos^2(\theta)}{\cos^2(\theta) + \lambda^2}\di \sigma - \sum_{\alpha \in \cI_k^+} \sigma(p_\alpha)\frac{\cos^2(\theta_\alpha)}{\cos^2(\theta_\alpha)+\lambda^2} \right|
    \leq C (N^{-\delta} + R M^{-1} N^\delta) \;.
    \label{eq:bound2}
\end{equation}
Now we compute the integral over the half-circle. First, using $\cos^2(\theta) = (1+\cos(2\theta))/2$ and the symmetry $\cos^2(\pi-\theta) = \cos^2(\theta)$, we can write
\begin{equation*}
    \int_{\mathbb{S}^1_{\textrm{half}}}\frac{\cos^2(\theta)}{\cos^2(\theta) + \lambda^2}\di \sigma = \int_0^\pi\frac{\cos^2(\theta)}{\cos^2(\theta) + \lambda^2}\di \sigma = \int_0^\pi\frac{1+\cos(2\theta)}{2\lambda^2+\cos(2\theta)+1} \di \theta\;.
\end{equation*}
Let $z\coloneqq e^{i2\theta}$ and let $\gamma$ be the complex unit circle. Then
\begin{equation*}
    \int_{\mathbb{S}^1_{\textrm{half}}}\frac{\cos^2(\theta)}{\cos^2(\theta) + \lambda^2}\di \sigma =-\frac{i}{2}\int_\gamma \frac{(z+1)^2}{z((z+1)^2+4\lambda^2z)} \di z\;.
\end{equation*}
We have three poles: $z=0,~z=z_+,~z=z_-$, where $z_\pm \coloneq -1-2\lambda^2 \pm 2\lambda\sqrt{1+\lambda^2}$. Notice that if $\lambda\in(0, \infty)$, then $z_+$ is inside the unit circle and $z_-$ is outside, while if $\lambda\in(-\infty, 0)$, the opposite is true. We are interested in the case $\lambda \in (0, \infty)$, and thus we have
\begin{equation*}
    -\frac{i}{2}\int_\gamma \frac{(z+1)^2}{z^3+z+4\lambda^2z^2+2z^2} \di z
    =\pi \left[1+\frac{(z_++1)^2}{z_+(z_+-z_-)}\right]\;,
\end{equation*}
and we finally conclude that
\begin{equation}
    \int_{\mathbb{S}^1_{\textrm{half}}}\frac{\cos^2(\theta)}{\cos^2(\theta) + \lambda^2}\di \sigma = \pi\left( 1 - \frac{\lambda}{\sqrt{\lambda^2+1}} \right)\;.
    \label{eq:integral_identity}
\end{equation}
Since $g(k) \le \hat{V}(k)$, using~\eqref{eq:Qk_replacement} and~\eqref{eq:bound2}, we conclude
\begin{equation}
    \left|Q_k(\lambda)-\widetilde{Q}_k(\lambda)\right|\leq C \hat{V}(k) \big( N^{-\delta} + R M^{-1} N^\delta + R M N^{-\frac 12 + \delta} \big) \;.
\end{equation}
Notice that for $x\geq0$, the function $F(x) = \log(1+x)-x$ has a unit Lipschitz constant, so
\begin{equation}
    \left| F\big(Q_k(\lambda)\big) - F\big( \widetilde{Q}_k(\lambda) \big)\right|
    \leq \left|Q_k(\lambda)-\widetilde{Q}_k(\lambda)\right|
    \leq C \hat{V}(k) \big( N^{-\delta} + R M^{-1} N^\delta + R M N^{-\frac 12 + \delta} \big) \;.
    \label{eq:f1}
\end{equation}
Now we have to compare the integrals with respect to $\lambda$. Using $|F(x)|\leq x$ for any $x \geq 0$ and that $0\leq u_\alpha(k)^4\leq 1$, we have
\begin{equation}
    \left|F \big( Q_k(\lambda) \big) \right|
    \leq C g(k) \sum_{\alpha \in \cI_k^+} \sigma(p_\alpha)\frac{u_\alpha(k)^4}{u_\alpha(k)^4+\lambda^2}\leq C g(k) \sum_{\alpha \in \cI_k^+} \frac{1}{M\lambda^2}
    \leq C \frac{\hat{V}(k)}{\lambda^2}\;.
    \label{eq:f2}
\end{equation}
Now, by \eqref{eq:trace_EDW_formula_2}, we see that 
\begin{equation}
    \left|F \big( \widetilde{Q}_k(\lambda) \big) \right|
    \leq 2\pi g(k)\left|1-\frac{\lambda}{\sqrt{\lambda^2+1}} \right|
    \leq C \frac{\hat{V}(k)}{\lambda^2}\;.
    \label{eq:f3}
\end{equation}
Take $\Lambda > 0$ to be optimized later. Then, putting \eqref{eq:f1}--\eqref{eq:f3} together, we have
\begin{align}
    &\left| \int_0^\infty F\big(Q_k(\lambda) \big) \di \lambda
        - \int_0^\infty F\big(\widetilde{Q}_k(\lambda) \big) \di \lambda \right|
    \leq \int_0^\Lambda \left| F\big(Q_k(\lambda) \big)
        - F\big(\widetilde{Q}_k(\lambda) \big) \right|\di \lambda 
        + C \int_\Lambda^\infty \frac{\hat{V}(k)}{\lambda^2}\di \lambda \nonumber\\
    &\leq C \hat{V}(k) \Lambda \big( N^{-\delta} + R M^{-1} N^\delta + R M N^{-\frac 12 + \delta} \big)
        + C \hat{V}(k) \Lambda^{-1} \nonumber\\
    &\leq C \hat{V}(k) \big( N^{-\frac{\delta}{2}} + R^{\frac 12} M^{-\frac 12} N^{\frac{\delta}{2}} + R^{\frac 12} M^{\frac 12} N^{-\frac 14 + \frac{\delta}{2}} \big) \;,
\end{align}
where, in the last step, we have optimized with respect to $\Lambda$. Comparing~\eqref{eq:trace_EDW_formula_1} and~\eqref{eq:trace_EDW_formula_2}, we obtain
\begin{equation} \label{eq:final_error_1}
\begin{aligned}
    &\left| E^{\RPA}_<
    - \hbar \kappa \sum_{k \in \Gamma^{\nor}} |k|\tr \left( E(k) - D(k) - W(k) \right)
        \right| \\
    &\le C \hbar \big( N^{-\frac{\delta}{2}} + R^{\frac 12} M^{-\frac 12} N^{\frac{\delta}{2}} + R^{\frac 12} M^{\frac 12} N^{-\frac 14 + \frac{\delta}{2}} \big)
    \sum_{k \in \ZZZ^2 : |k| < R} |k| \hat{V}(k) \;.
\end{aligned}
\end{equation}
As in~\eqref{eq:Vsplit}, $ \sum_{k \in \ZZZ^2 : |k| < R} |k| \hat{V}(k) \leq C R^{\frac{2+b}{2}} $. If $\sum_{k \in \ZZZ^2} |k| \hat{V}(k) < \infty$, then the sum on the r.~h.~s. is even $\leq C$.\\
To estimate $E^{\RPA}_\ge$, we use that $\widetilde{Q}_k(\lambda) \leq \frac{\hat{V}(k)}{4 \pi}$ is uniformly bounded in $k \in \ZZZ^2$, so by Taylor expansion,
\begin{equation*}
    F(\widetilde{Q}_k(\lambda))
    \leq C \widetilde{Q}_k(\lambda)^2 \;.
\end{equation*}
Thus,
\begin{align*}
    &|E^{\RPA}_\ge|
    \leq C \hbar \sum_{k \in \ZZZ^2 : |k| \ge R} |k| \int_0^\infty |\widetilde{Q}_k(\lambda)|^2 \; \di \lambda
    \leq C \hbar \sum_{k \in \ZZZ^2 : |k| \ge R} |k| \hat{V}(k)^2 \int_0^\infty \left( 1 - \frac{\lambda}{\sqrt{\lambda^2+1}} \right)^2 \di \lambda \\
    &\leq C \hbar \sum_{k \in \ZZZ^2 : |k| \ge R} |k| \hat{V}(k)^2 \Big( \int_0^1 1 \; \di \lambda + \int_1^\infty \frac{1}{\lambda^4} \di \lambda \Big)
    \leq C \hbar R^{b-1} \sum_{k \in \ZZZ^2 : |k| \ge R} |k|^{2-b} \hat{V}(k)^2 
    \leq C \hbar R^{b-1} \;.
\end{align*}
If $ \sum_{k \in \ZZZ^2} |k| \hat{V}(k) < \infty $, then $\hat{V}(k) \leq C |k|^{-1}$ (see Remark~\ref{rem:V}), so $\sum_{k \in \ZZZ^2 : |k| \ge R} |k|^2 \hat{V}(k)^2$ and the bound is true with $b=0$. Combining the bound on $|E^{\RPA}_\ge|$ with~\eqref{eq:final_error_1} renders~\eqref{eq:trace_evaluation_mod} and~\eqref{eq:trace_evaluation}. To establish~\eqref{eq:ERPA_bound}, we proceed similarly as for $|E^{\RPA}_\ge|$ and obtain
\begin{align*}
    |E^{\RPA}|
    &\leq C \hbar \sum_{k \in \ZZZ^2} |k| \int_0^\infty |\widetilde{Q}_k(\lambda)|^2 \; \di \lambda
    \leq C \hbar \sum_{k \in \ZZZ^2} |k|\hat{V}(k)^2 
    \leq C \hbar \;.
\end{align*}
\end{proof}

\subsection{Lower Bound}

\begin{prop} \label{prop:lowerbound}
Let $ \hat{V}(k) = \hat{V}(-k) \ge 0 $ and $ \sum_{k \in \ZZZ^2} |k| \hat{V}(k) < \infty $. Then, there exist $C > 0$ and $a > 0$ such that
\begin{equation} \label{eq:lowerbound}
    E_\GS\geq E_\FS + E^{\RPA} - C N^{-\frac 12 - a} \;.
\end{equation}
\end{prop}

\begin{proof}
We proceed as in~\cite[Sect.~10]{benedikter2021correlation} and~\cite[Sect.~9]{benedikter2023correlation}: Let $\psi_\GS$ be the ground state of $H_N$ and $\xi \coloneqq R\psi_\GS$, which obviously belongs to an approximate ground state in the sense of Definition~\ref{def:approxGS}. Also, $R \xi$ is obviously an eigenvector of $H_N$, so all a priori bounds of Lemma~\ref{lemma:A-Priori-Bounds} apply. Recall that by~\eqref{eq:HNconjugation}, the ground state energy is given by
\begin{equation*}
	E_{\GS}
	= \langle \psi_{\GS}, H_N \psi_{\GS} \rangle
	= \langle \xi, (\HHH_0 + Q_{\B} + \cE_1 + \cE_2 + \XXX) \xi \rangle + E_{\FS} \;.
\end{equation*}
We recall $\cE_1\geq 0$. Then, from Lemmas~\ref{lemma:Exchange-Term-Bound},~\ref{lem:cE_2estimate}, and~\ref{lem:Q_QR_bound}, we get
\begin{equation*}
    |\langle\xi,\XXX\xi \rangle|
    \leq C_\varepsilon \hbar
    N^{-\frac 14 + \varepsilon}\;, \qquad
    |\langle\xi,\cE_2\xi\rangle|
    \leq C_\varepsilon \hbar N^{-\frac{1}{136} + \varepsilon}\;,
\end{equation*}
\begin{equation*}
   |\langle \xi, (Q_{\B} - Q_{\B}^R) \xi \rangle|
	\le C_\varepsilon \hbar N^{\varepsilon}
        \big(R^{-\frac 12}
        + N^{-\frac 18}
        + N^{-\frac{\delta}{2}}
        + R M^{\frac 12} N^{-\frac 14 + \frac{\delta}{2}} \big) \;.
\end{equation*}
We conclude
\begin{align} \label{eq:E_GS}
    E_\GS
    \geq &E_\FS+\langle\xi,(\DDD_\B+Q^R_\B)\xi\rangle +\langle\xi,(\HHH_0-\DDD_\B)\xi\rangle \nonumber\\
    &-C_\varepsilon \hbar N^{\varepsilon}(
    N^{- \frac{1}{136}} +
    R^{-\frac 12} +
    N^{-\frac \delta2} +
    R M^\frac 12 N^{-\frac 14 + \frac \delta2} )\;.
\end{align}
By means of Lemmas~\ref{lemma:A-Priori-Bounds},~\ref{lem:gronwall} and~\ref{lem:Approximate_Bogoliubov_Invariance}, and writing $\xi=TZ\eta$, we can estimate
\begin{align}\label{eq:H0-DB-estimate}
    \langle\xi, (\HHH_0-\DDD_\B)\xi\rangle&\geq\langle\eta,(\HHH_0-\DDD_\B)\eta\rangle-C_\varepsilon \hbar(M^{-1}N^{C\delta}+R^2M^\frac 32 N^{-\frac 14 + C\delta+\frac \varepsilon2})\nonumber \\
    &\geq -\langle\eta,\DDD_\B\eta\rangle-C_\varepsilon \hbar(M^{-1}N^{C\delta}+R^2M^\frac 32 N^{-\frac 14 + C\delta+\frac \varepsilon2}) \;,
\end{align}
where we used that $\HHH_0\geq 0$. To treat $(\DDD_\B+Q^R_\B) $, we write (compare~\eqref{eq:h_eff} and~\eqref{eq:h_eff_2})
\begin{equation} \label{eq:DDD_QRB}
    \DDD_\B + Q^R_\B
    = \sum_{k\in\Gamma^\nor} 2\hbar\kappa|k|h_\eff(k) \;.
\end{equation}
By means of Lemma~\ref{lem:approximate_bosonic_bogoliubov_transformation}, as in~\cite[Sect.~10]{benedikter2021correlation}, conjugation with $T$ results in
\begin{equation}
    T^*h_\eff(k) T=h_\eff^\diag(k)+\fE^\diag(k)\;,
    \label{eq:h_eff_conjugation}
\end{equation}
where $\fE^\diag(k)$ is bounded in~\eqref{eq:fE_diag_normbound} and below, and where the leading-order term is given by
\begin{equation} \label{eq:heff_diag}
\begin{aligned}
	h_{\eff}^{\diag}(k)
	&\coloneq \sum_{\alpha, \beta \in \cI_k} \Big( \big( D(k) + W(k) \big)_{\alpha, \beta} \tilde{c}^*_\alpha(k) \tilde{c}_\beta(k)
		+ \frac 12 \widetilde{W}(k)_{\alpha, \beta} \big( \tilde{c}^*_\alpha(k) \tilde{c}^*_\beta(k) + \tilde{c}_\beta(k) \tilde{c}_\alpha(k) \big) \Big) \\
	&= \frac 12 \tr \big( E(k) - D(k) - W(k) \big)
		+ \sum_{\alpha, \beta \in \cI_k} \fK(k)_{\alpha, \beta} c^*_\alpha(k) c_\beta(k)
		+\fE^\textrm{no}(k)\;,\\
    \fE^\textrm{no}(k)
    &\coloneq \frac 12 \sum_{\alpha\in\cI_k}\Big(2\sinh (K(k))\big( D(k)+W(k) \big)\sinh (K(k)) +\\
    &\quad
    + \cosh (K(k)) \widetilde{W}(k) \sinh(K(k)) +\sinh (K(k)) \widetilde{W}(k)\cosh(K(k)) \Big)_{\alpha,\alpha}\cE_\alpha(k,k)\;,
\end{aligned}
\end{equation}
with $ \fK(k) \coloneq O(k) E(k) O(k)^T $, see~\eqref{eq:K} and below, and $\cE_\alpha(k,\ell)$ defined in~\eqref{eq:cE_alpha}. To bound the normal ordering error $\fE^\textrm{no}(k)$, notice that from Lemmas~\ref{lem:K-kernel} and~\ref{lem:normalization_constant}, as well as the definition~\eqref{eq:DWW} of $D, W, \widetilde{W}$,
\begin{equation} \label{eq:matrixbounds_KDWW}
\begin{aligned}
    |\sinh(K(k))_{\alpha, \beta}|,
    |W(k)_{\alpha,\beta}| ,
    |\widetilde{W}(k)_{\alpha,\beta}|
    &\leq C \hat{V}(k) M^{-1} \;, \qquad
    \norm{\cosh(K(k))} ,
    \norm{\sinh(K(k))}
    \leq C \;,\\
    \norm{W(k)},
    \Vert \widetilde{W}(k) \Vert
    &\leq C \hat{V}(k) \;, \qquad
    |D(k)_{\alpha,\beta}| 
    \leq \delta_{\alpha, \beta} \;, \qquad
    \norm{D(k)}
    \leq 1 \;,
\end{aligned}
\end{equation}
so, omitting the $k$-indices,
\begin{equation*}
    \Big|\Big(2\sinh (K)\big( D+W \big)\sinh (K)+ \cosh (K) \widetilde{W}\sinh(K) +\sinh (K) \widetilde{W}\cosh(K)\Big)_{\alpha,\alpha}\Big|\leq C\frac{\hat{V}(k)}{M}\;.
\end{equation*}
From~\eqref{eq:cE_alpha} and Lemma~\ref{lem:normalization_constant} with $|k \cdot \hat{\omega}_\alpha| \ge N^{-\delta}$, it becomes evident that
\begin{equation*}
\begin{aligned}
    \sum_{\alpha\in\cI_k}|\langle \xi ,\cE_\alpha(k,k) \xi \rangle|
    &\leq \sup_{\alpha\in\cI_k}\frac{1}{n_\alpha(k)^2}\langle \xi,\cN  \xi\rangle
    \leq C N^{-\frac12 +\delta}M\langle\xi,\cN\xi\rangle \\
    \Rightarrow \quad
    \pm \fE^\textrm{no}(k)
    &\leq C \hat{V}(k) N^{-\frac 12 + \delta} \cN \;.
\end{aligned}
\end{equation*}
We conclude
\begin{equation} \label{eq:heff_diag_bound}
    h_{\eff}^{\diag}(k)
    \geq \frac 12 \tr \big( E(k) - D(k) - W(k) \big)
		+ \sum_{\alpha, \beta \in \cI_k} \fK(k)_{\alpha, \beta} c^*_\alpha(k) c_\beta(k)
		- C \hat{V}(k) N^{-\frac 12 + \delta} \cN\;.
\end{equation}
Next, $\fE^\diag(k)$ is computed via Lemma~\ref{lem:approximate_bosonic_bogoliubov_transformation} with $\fE_\alpha(k) \coloneq \fE_\alpha(1,k)$, compare~\cite[(10.6)]{benedikter2021correlation}
\begin{align} \label{eq:fE_diag_normbound}
    | \langle Z \eta , \fE^{\diag}(k) Z \eta \rangle|
    &\leq \sum_{\alpha, \beta \in \cI_k} \big( D(k) + W(k) \big)_{\alpha, \beta} \big(
        2 \Vert \tilde{c}_\alpha(k) Z \eta \Vert
        \Vert \fE_\beta(k) Z \eta \Vert
        + \Vert \fE_\alpha(k) Z \eta \Vert
        \Vert \fE_\beta(k) Z \eta \Vert\big) \nonumber \\
    &\quad + 2 \sum_{\alpha, \beta \in \cI_k} \widetilde{W}(k)_{\alpha, \beta}\big(
        \Vert \tilde{c}_\alpha(k) Z \eta \Vert
        \Vert \fE_\beta(k)^* Z \eta \Vert
        + \Vert \fE_\alpha(k) Z \eta \Vert
        \Vert \fE_\beta(k)^* Z \eta \Vert\big) \nonumber \\
    &\quad + 2 \sum_{\alpha \in \cI_k} \Vert \fE_\alpha(k) Z \eta \Vert
        \Big\Vert \sum_{\beta \in \cI_k} \widetilde{W}(k)_{\alpha, \beta} \tilde{c}^*_\beta(k) Z \eta \Big\Vert \;.
\end{align}
Then, we apply Lemma~\ref{lem:approximate_bosonic_bogoliubov_transformation}, as well as the matrix element bounds~\eqref{eq:matrixbounds_KDWW}, and conclude
\begin{equation} \label{eq:fE_diag_bound}
\begin{aligned}
	| \langle Z \eta, \fE^{\diag}(k) Z \eta \rangle|
	&\le C M N^{-\frac 12 + \delta}
		\langle Z \eta, (\cN_\delta + M)(\cN + 1)^2 Z \eta \rangle^{\frac 12}
		\langle Z \eta, (\cN_\delta + 1) Z \eta \rangle^{\frac 12}\\
	&\quad + C M^2 N^{-1 + 2 \delta} \langle Z \eta, (\cN_\delta + M)(\cN + 1)^2 Z \eta \rangle \;.
\end{aligned}
\end{equation}
Plugging~\eqref{eq:h_eff_conjugation} into~\eqref{eq:DDD_QRB}, and applying~\eqref{eq:heff_diag_bound} and~\eqref{eq:fE_diag_bound} results in
\begin{align*}
    &\langle\xi,(\DDD_\B+Q_\B^R)\xi\rangle \nonumber \\
    &\geq \hbar\kappa \sum_{k\in\Gamma^\nor}|k|\tr \big( E(k)-D(k)-W(k) \big)
    +2\hbar\kappa\sum_{k\in\Gamma^\nor}\sum_{\alpha,\beta\in\cI_k}|k|\fK(k)_{\alpha,\beta}\langle Z \eta ,c_\alpha^*(k)c_\beta(k) Z \eta \rangle \nonumber \\
    &\quad - C\hbar\sum_{k\in\Gamma^\nor} |k|\Big( \hat{V}(k) N^{-\frac 12+\delta}\langle \xi,(\cN + 1) \xi\rangle 
    +M N^{-\frac 12 + \delta}
		\langle \xi, (\cN_\delta + M)(\cN + 1)^2 \xi \rangle^{\frac 12}
		\langle \xi, (\cN_\delta + 1) \xi \rangle^{\frac 12} \nonumber \\
    &\quad + M^2 N^{-1 + 2 \delta} \langle \xi, (\cN_\delta + M)(\cN + 1)^2 \xi \rangle \Big)\;,
\end{align*}
where we propagated expectations in $Z \eta = T^* \xi = T_{-1} \xi$ to expectations in $\xi$ using Lemma~\ref{lem:gronwall}. By means of the a priori bounds given in Lemma~\ref{lemma:A-Priori-Bounds}, and with $\sum_{k \in \Gamma^{\nor}} |k| \leq C R^3$, we conclude
\begin{align}\label{eq:D_B+Q_B-lower-bound}
    &\langle\xi,(\DDD_\B+Q_\B^R)\xi\rangle \nonumber \\
    &\geq \hbar\kappa \sum_{k\in\Gamma^\nor}|k|\tr \big( E(k)-D(k)-W(k) \big)
    +2\hbar\kappa\sum_{k\in\Gamma^\nor}\sum_{\alpha,\beta\in\cI_k}|k|\fK(k)_{\alpha,\beta}\langle Z\eta,c_\alpha^*(k)c_\beta(k)Z\eta\rangle \nonumber \\
    &\quad - C_\varepsilon \hbar N^\varepsilon 
        R^3
        \Big( MN^{-\frac 14 + 2\delta}
        + M^{\frac 32}N^{-\frac 14 +\frac 32 \delta} 
        + M^2 N^{-\frac 12 + 3\delta} 
        + M^3 N^{-\frac 12 + 2\delta} \Big)\;.
\end{align}
Similarly, using Lemma~\ref{lem:approximate_bosonic_bogoliubov_transformation}, we approximately diagonalize $\fK(k)$ using the transformation $ Z $. In analogy to~\cite[Eq.~(9.15)]{benedikter2023correlation}, with Lemma~\ref{lemma:A-Priori-Bounds}, we obtain
\begin{align}\label{eq:h_eff_estimate}
	&2 \kappa \hbar \sum_{k \in \Gamma^{\nor}} \sum_{\alpha, \beta \in \cI_k} |k| \fK(k)_{\alpha, \beta} \langle Z \eta, c^*_\alpha(k) c_\beta(k) Z \eta \rangle \nonumber \\
    &\ge \langle \eta, \DDD_{\B} \eta \rangle
	- C \hbar \sum_{k \in \Gamma^{\nor}} |k| \big( M^2 N^{-\frac 12 + C \delta} \eva{\eta, \cN_\delta \eta}^{\frac 12} \eva{\eta, \cN_\delta \cN^2 \eta}^{\frac 12}
        + M^{\frac 72} N^{-1 + C \delta} \eva{\eta, \cN_\delta \cN^2 \eta} \big) \nonumber \\
	&\ge \langle \eta, \DDD_{\B} \eta \rangle
	- C_\varepsilon \hbar R^3 \big( M^2 N^{-\frac 14 + C \delta + \varepsilon}
        + M^{\frac 72} N^{-\frac 12 + C \delta + \varepsilon} \big) \;,
\end{align}
Plugging~\eqref{eq:H0-DB-estimate} and~\eqref{eq:D_B+Q_B-lower-bound} into~\eqref{eq:E_GS}, and then inserting~\eqref{eq:h_eff_estimate} and Lemma~\ref{lem:trace_evaluation}, we obtain
\begin{align}
    E_\GS \geq &E_\FS + E^{\RPA} - C_\varepsilon \hbar N^\varepsilon \big( 
        R^3 M^2 N^{-\frac 14 + C\delta} 
        + R^3 M^{\frac 72} N^{-\frac 12 + C\delta}
        + M^{-1} N^{C\delta} \nonumber\\
    &+ N^{-\frac{1}{136}}
        + R^{-\frac 12} 
        + N^{-\frac \delta2}
        + R M^{\frac 12}N^{-\frac 14+\frac \delta2}
        + R^{\frac 12} M^{-\frac 12} N^{\frac{\delta}{2}} \big)\;.
\end{align}
Choosing $M = N^{2C\delta}$, then $\delta < \frac{1}{20 C}$ small enough, then $R = N^{4 \alpha}$ with $\alpha < \frac{\delta}{4}$ small enough and then $\varepsilon = \alpha$, we conclude the lower bound~\eqref{eq:lowerbound}.
\end{proof}

\subsection{Upper Bound}

\begin{prop} \label{prop:upperbound}
Let $ \hat{V}(k) = \hat{V}(-k) \ge 0 $ and $ \sum_{k \in \ZZZ^2} |k|^{2-b} \hat{V}(k)^2 < \infty $ for some $b \in (0,1)$. Then,
\begin{equation} \label{eq:upperbound}
\begin{aligned}
    E_{\GS} 
    &\leq E_{\FS} + E^{\RPA} + o(N^{-\frac 12}) \;.
\end{aligned}
\end{equation}
\end{prop}

\begin{proof}
As in~\cite[Sect.~9]{benedikter2023correlation}, for the upper bound, we use the trial state $ \tilde{\psi} = R T \Omega $, with $R$ and $T$ defined in~\eqref{eq:R} and~\eqref{eq:bogoliubov}, and set $ \tilde{\xi} \coloneq T \Omega $. Then, we may use $\cN_\delta \leq \cN$ and Lemma~\ref{lem:gronwall} provides us with the following a priori bounds:
\begin{equation} \label{eq:a_priori_upper_bound}
    \langle \tilde\xi, \cN^m \tilde\xi \rangle
    \leq C_m e^{C_m R^{\frac b2}} \langle \Omega, (\cN+1)^m \Omega \rangle
    \leq C_m e^{C_m R^{\frac b2}} \qquad \forall m \in \NNN \;.
\end{equation}
To obtain an a priori bound for $\HHH_0$, we write
\begin{equation*}
    \langle \tilde\xi, \HHH_0 \tilde\xi \rangle
    = \langle \tilde\xi, \DDD_{\B} \tilde\xi \rangle
        + \langle \tilde\xi, (\HHH_0 - \DDD_{\B}) \tilde\xi \rangle \;,
\end{equation*}
where $\DDD_{\B}$ was defined in~\eqref{eq:DDD_B}.
With Lemma~\ref{lem:c_conversion} and~\eqref{eq:a_priori_upper_bound}, using $R^3 \ll e^{C R^{\frac b2}}$, we bound $\langle \tilde\xi, \DDD_{\B} \tilde\xi \rangle \leq C \hbar e^{C R^{\frac b2}} $, and by Lemma~\ref{lem:Approximate_Bogoliubov_Invariance} and~\eqref{eq:a_priori_upper_bound}, we obtain
\begin{equation}\label{eq:UpperBoundKinetic}
    \langle \tilde{\xi}, (\HHH_0 - \DDD_{\B}) \tilde{\xi} \rangle = \langle \Omega, T^*(\HHH_0 - \DDD_{\B})T \Omega \rangle 
    \leq C \hbar e^{C R^{\frac b2}} \big( M^{-1} + M N^{-\frac 12+\delta} \big)\;.
\end{equation}
In total, we conclude the a priori bound
\begin{equation*}
    \langle \tilde\xi, \HHH_0 \tilde\xi \rangle
    \leq C \hbar e^{C R^{\frac b2}} \;.
\end{equation*}
Now, recall~\eqref{eq:HNconjugation}:
\begin{equation*}
	E_{\GS}
	\le \langle \tilde{\psi}, H_N \tilde{\psi} \rangle
	= \langle \tilde{\xi}, (\HHH_0 + Q_{\B} + \cE_1 + \cE_2 + \XXX) \tilde{\xi} \rangle + E_{\FS} \;.
\end{equation*}
First, note that the number of excitations in $\tilde\xi$ is an integer multiple of 4, while $\cE_2$ changes the excitation number by 2, so $\langle \tilde\xi, \cE_2 \tilde\xi \rangle = 0$. 
Next, note that the trial state $\tilde\xi$ only contains excitations of momenta $ |p| \leq k_{\F} + R < C N^{\frac 12}$ for $C>0$ large enough. We may therefore write
\begin{equation}
	E_{\GS}
	\le \langle \tilde{\xi}, (\HHH_0 + \widetilde{Q}_{\B} + \widetilde{\cE}_1 + \widetilde{\XXX}) \tilde{\xi} \rangle + E_{\FS} \;,
\end{equation}
where in $\widetilde{Q}_{\B}$, $\widetilde{\cE}_1$, and $\widetilde{\XXX}$ we restrict to $|k| < C N^{\frac 12}$, see~\eqref{eq:QBtilde} and~\eqref{eq:XXXtilde_cE1tilde}.
With Lemmas~\ref{lemma:Exchange-Term-Bound} and~\ref{lem:cE_1estimate}, as well as~\eqref{eq:a_priori_upper_bound}, we then bound
\begin{align*}
    \langle \tilde\xi, \widetilde{\XXX} \tilde\xi \rangle
    \leq C N^{-1 + \frac b4} \langle \tilde\xi, \cN \tilde\xi \rangle
    \leq C \hbar e^{C R^{\frac b2}} N^{-\frac 12 + \frac b4} \;, \quad
    \langle \tilde\xi, \widetilde{\cE}_1 \tilde\xi \rangle
    \leq C N^{-1 + \frac b4} \langle \tilde\xi, \cN^2 \tilde\xi \rangle
    \leq C \hbar e^{C R^{\frac b2}} N^{-\frac 12 + \frac b4} \;.
\end{align*}
Using these bounds with Lemma~\ref{lem:Q_QR_bound} and~\eqref{eq:UpperBoundKinetic}, we obtain
\begin{equation}\label{eq:UpperBoundGS}
    \begin{aligned}
	E_\GS
	&\le E_\FS+\langle \tilde{\xi}, (\DDD_{\B} + Q_{\B}^R) \tilde{\xi} \rangle + C_\varepsilon \hbar e^{C R^{\frac b2}} N^{\varepsilon} \\
        &\quad \times \big( N^{-\frac 12 + \frac b4}
            + M^{\frac 32} N^{-\frac 14 + \frac b4 + \frac \delta2}
            + N^{-\frac 18}
            + N^{-\frac \delta2}
            + M^{-1}
            + M N^{-\frac 12+\delta}
        \big) \;.
    \end{aligned}
\end{equation}
For the remaining term $\langle \tilde{\xi}, (\DDD_\B+Q_\B^R)\tilde{\xi}\rangle$, we proceed as in~\eqref{eq:DDD_QRB}--\eqref{eq:h_eff_estimate}, where now $c_\beta(k) \Omega = 0$ and $\langle \Omega, \fE^{\mathrm{no}} \Omega \rangle = 0$.
\begin{equation} \label{eq:UpperBoundRPA}
    \begin{aligned}
        &\langle \tilde{\xi}, (\DDD_\B+Q_\B^R)\tilde{\xi}\rangle
        =\sum_{k\in\Gamma^\nor}2 \hbar \kappa |k|\langle\Omega, T^*h_\eff(k)T\Omega\rangle \\
        &= \hbar \kappa \sum_{k\in\Gamma^\nor}|k|
        \Big( \tr \big(E(k)-D(k)-W(k) \big)
            + 2 \langle\Omega, \fE^\diag(k)\Omega\rangle \Big) \\
        &\leq \hbar \kappa \sum_{k\in\Gamma^\nor}|k|\tr \big(E(k)-D(k)-W(k) \big)
            + C \hbar e^{C R^{\frac b2}} \big(
            M^{\frac 32} N^{-\frac 12+\delta}
            + M^3 N^{-1+2\delta} \big) \\
        &\leq E^\RPA
            + C \hbar e^{C R^{\frac b2}} \big( N^{-\frac \delta2}
            + M^{-\frac12} N^\frac \delta2
            + M^\frac 12 N^{-\frac14+\frac \delta2}
            + M^{\frac 32} N^{-\frac 12+\delta}
            + M^3 N^{-1+2\delta}\big) 
            + C \hbar R^{b-1} \;,
    \end{aligned}
\end{equation}
where in the second-last line, we bounded $\fE^\diag(k)$ as in~\eqref{eq:fE_diag_bound}, and in the last line we evaluated the trace by Lemma~\ref{lem:trace_evaluation}. Inserting~\eqref{eq:UpperBoundRPA} into~\eqref{eq:UpperBoundGS}, we have
\begin{equation}
    \begin{aligned}
    E_\GS
	&\le E_\FS+E^\RPA
        + C_\varepsilon \hbar e^{C R^{\frac b2}} N^{\varepsilon} \big( N^{-\frac 12 + \frac b4}
            + M^{\frac 32} N^{-\frac 14 + \frac b4 + \frac \delta2}
            + N^{-\frac 18}
            + N^{-\frac \delta2}\\
    &\quad + M^{-\frac12} N^\frac \delta2
            + M^{\frac 32} N^{-\frac 12+\delta}
            + M^3 N^{-1+2\delta}
        \big)
        + C \hbar R^{b-1} \;.
    \end{aligned}
\end{equation}
Finally, we choose the patch size $R$ to grow slowly with $N$, such that $e^{C R^{\frac b2}} < N^\varepsilon$, then fix $M = N^{4 \delta}$, and then choose $\varepsilon, \delta > 0$ so small that $ 2 \varepsilon + 14 \delta - \frac{1-b}{4} < 0 $.
\end{proof}

\section*{Acknowledgments}
The authors would like to thank Niels Benedikter for helpful discussions. SL was supported by the European Union (ERC \textsc{FermiMath} nr.~101040991 and ERC \textsc{MathBEC} nr.~101095820). Views and opinions expressed are those of the authors and do not necessarily reflect those of the European Union or the European Research Council Executive Agency. Neither the European Union nor the granting authority can be held responsible for them. SL was partially supported by Gruppo Nazionale per la Fisica Matematica in Italy.

\section*{Statements and Declarations}
The authors have no competing interests to declare.

\section*{Data Availability}
As purely mathematical research, there are no datasets related to the article.

\appendix

\section{Number Theoretical Estimates}
\label{app:numbertheory}

\subsection{Inverse Energy Sum over a Lune}
The following result is the 2d analog of~\cite[Proposition~A.2]{christiansen2023random}.

\begin{prop} \label{prop:lambda-bound}
Let $ k\in\mathbb{Z}^2$ and recall for $ p \in L_k = B_{\F}^c \cap (B_{\F} + k) $ the pair excitation energy $ \lambda_{k,p} = \frac 12 \hbar^2 (|p|^2 - |p-k|^2) $ with $\hbar = N^{-\frac 12}$. Then,
\begin{equation} \label{eq:lambda-bound}
    \hbar^2 \sum_{p\in\lune}\lambda_{k,p}^{-1}
    \leq C \log(N) \;.
\end{equation}
\end{prop}

\noindent Note that the analogous bound in 3d is of order $k_{\F}$, so it is larger by almost a factor of $k_{\F}$.

\begin{proof}
\noindent
We proceed similarly to~\cite{christiansen2023random}, treating the cases $|k| < 2 k_{\F}$ and $|k| \ge 2 k_{\F}$ separately.\\

\textbf{Case $|k|<2k_\F$}. At the tips of the lune, i.e., if $m \approx m_*$, we expect $\la$ to blow up. Therefore, we subdivide the lune in a "bulk" region, where we can estimate the sum by an integral, and a "tip" region:
\begin{equation} \label{eq:Lk_Bulk}
    L^{\textrm{Bulk}}_k \coloneq \left\{ p\in\lune ~\middle|~ \hat{k}\cdot p -\tfrac{|k|}{2}>\tfrac{3}{2}\sqrt{2} \right\}\;, \qquad
    L^{\textrm{Tip}}_k \coloneq L_k \setminus L^{\textrm{Bulk}}_k \;.
\end{equation}
If we call $\mathcal{C}_p = [-\frac{1}{2},\frac{1}{2}]^2+p$ the box around the lattice point p, we have
\begin{equation}
    \hbar^2 \sum_{p\in\lbulk}\la
    = \int \hbar^2 \sum_{q\in\lbulk}\chi_{\mathcal{C}_q}(p)\lambda_{k,q}^{-1} \di p\;.
\end{equation}
Notice that $ \hat{k} \cdot q \ge \hat{k} \cdot p - \frac{\sqrt{2}}{2} $ for all $ p \in \mathcal{C}_q $, so we can dominate the integrand as
\begin{equation}
\begin{aligned}
    \hbar^2 \sum_{q\in\lbulk}\chi_{\mathcal{C}_q}(p)\lambda_{k,q}^{-1}
    &\leq \left(|k|\left(\hat{k}\cdot p -\tfrac{|k|}{2} -\tfrac{\sqrt{2}}{2}\right)\right)^{-1} \\
    \Rightarrow \quad
    \hbar^2 \sum_{p\in\lbulk}\la
    &\leq \int_{S^{\textrm{Bulk}}} \left(|k|\left(\hat{k}\cdot p -\tfrac{|k|}{2} -\tfrac{\sqrt{2}}{2}\right)\right)^{-1} \di p\;,
\end{aligned}
\end{equation}
with $S^{\textrm{Bulk}} \coloneq \bigcup_{p\in\lbulk}\mathcal{C}_p$. We enlarge this integration domain to facilitate calculations:
\begin{equation}
    \overline{S}^{\textrm{Bulk}}
    \coloneq \left\{ p \in \RRR^2 ~\middle|~ |p| \ge k_{\F} - \tfrac{\sqrt{2}}{2} , \quad |p-k| < k_{\F} + \tfrac{\sqrt{2}}{2} , \quad \hat{k} \cdot p > \tfrac{|k|}{2} + \sqrt{2} \right\} \;,
\end{equation}
and split it as $\overline{S}^{\textrm{Bulk}} = S^1 \cup S^2$, where
\begin{align}
    S^1 \coloneq \left\{ p\in \overline{S}^{\textrm{Bulk}} ~\middle|~ \hat{k}\cdot p< k_\F -\tfrac{\sqrt{2}}{2} \right\}\;, \qquad
    S^2 \coloneq \left\{ p\in \overline{S}^{\textrm{Bulk}} ~\middle|~ \hat{k}\cdot p\geq k_\F -\tfrac{\sqrt{2}}{2} \right\}\;.
\end{align}
Then
\begin{equation}
    \hbar^2 \sum_{p\in\lbulk}\la
    \leq \underbrace{\int_{S^1}\left(|k|\left(\hat{k}\cdot p - \tfrac{|k|}{2} - \tfrac{\sqrt{2}}{2}\right)\right)^{-1}\di p}_{\eqqcolon~I_1} + \underbrace{\int_{S^2}\left(|k|\left(\hat{k}\cdot p - \tfrac{|k|}{2} - \tfrac{\sqrt{2}}{2}\right)\right)^{-1}\di p}_{\eqqcolon~I_2}\;.
\end{equation}

We begin by estimating $I_1$. Calling $\hat{k}\cdot p= z$, we have
\begin{align*}
    I_1 &= \int_{\frac{|k|}{2}+ \sqrt{2}}^{k_\F-\frac{\sqrt{2}}{2}} \frac{2}{|k|\left(z-\tfrac{|k|}{2} - \tfrac{\sqrt{2}}{2}\right)} \left(\sqrt{\left(k_\F +\tfrac{\sqrt{2}}{2}\right)^2 - \left(z-|k|\right)^2} - \sqrt{\left(k_\F -\tfrac{\sqrt{2}}{2}\right)^2 - z^2}\right) \di z\\
    &= \int_{\frac{\sqrt{2}}{2}}^{k_\F-\frac{|k|}{2}-\sqrt{2}} 2|k|^{-1}t^{-1}\left(\sqrt{a_+^2-b_+^2} - \sqrt{a_-^2-b_-^2}\right)\di t\;,
\end{align*}
where we performed the change of variable $t\coloneq z-\frac{|k|}{2}-\frac{\sqrt{2}}{2}$, and introduced
\begin{align*}
    a_+\coloneq k_\F +\tfrac{\sqrt{2}}{2}\;, \qquad
    b_+\coloneq t-\tfrac{|k|-\sqrt{2}}{2}\;, \qquad
    a_-\coloneq k_\F -\tfrac{\sqrt{2}}{2}\;, \qquad
    b_-\coloneq t+\tfrac{|k|+\sqrt{2}}{2}\;.
\end{align*}
We estimate the square roots as follows:
\begin{equation} \label{eq:sqrt_diff_estimate}
    \sqrt{a_+^2-b_+^2} - \sqrt{a_-^2-b_-^2} 
    = \frac{a_+^2-b_+^2 - a_-^2+b_-^2}{\sqrt{a_+^2-b_+^2} + \sqrt{a_-^2-b_-^2}} \leq \frac{2\sqrt{2}k_\F + 2|k|(t+\frac{\sqrt{2}}{2})}{\sqrt{a_+^2-b_+^2}}\;.
\end{equation}
In case $ |k| \leq k_{\F} $, we have $ a_+ + b_+ \ge \tfrac 12 k_{\F} $, so
\begin{equation*}
    a_+^2-b_+^2
    = (a_++b_+)(a_+-b_+)
    \geq \tfrac 12 k_\F\left( k_\F - t +\tfrac{|k|}{2} \right)\;,
\end{equation*}
and therefore
\begin{equation*}
    \sqrt{a_+^2-b_+^2} - \sqrt{a_-^2-b_-^2}
    \leq Ck_\F^{-1/2} \frac{k_\F + |k|t}{\left( k_\F - t +\tfrac{|k|}{2} \right)^{1/2}}\;.
\end{equation*}
The denominator might grow large as $ t $ approaches its maximum value, $ t = k_{\F} - \frac{|k|}{2} - \sqrt{2} $. Still, if $t \leq \frac{k_\F}{2}$, we can safely bound
\begin{equation*}
    \sqrt{a_+^2-b_+^2} - \sqrt{a_-^2-b_-^2}\leq Ck_\F^{-1/2} \frac{k_\F + |k|t}{\left( k_\F  +|k| \right)^{1/2}} \leq C\left(1 + \frac{|k|}{k_\F}t\right)\;.
\end{equation*}
Conversely, for $ t > \frac{k_{\F}}{2} $, we may conveniently estimate $ t^{-1} \le C k_{\F}^{-1} $ in the numerator. Thus
\begin{align*}
    I_1&\leq C\int_{\frac{\sqrt{2}}{2}}^{\frac{k_{\F}}{2}} |k|^{-1} t^{-1} \left(1 + \frac{|k|}{k_\F}t\right) \di t
    + C\int_{\frac{k_{\F}}{2}}^{k_{\F} - \frac{|k|}{2} - \sqrt{2}} |k|^{-1} t^{-1} k_\F^{-1/2}\frac{k_\F + |k|t}{\left( k_\F - t +\tfrac{|k|}{2} \right)^{1/2}} \di t\\
    &\leq C\int_{\frac{\sqrt{2}}{2}}^{\frac{k_{\F}}{2}} (|k|^{-1}t^{-1}+k_\F^{-1})\di t
    + C\int_{\frac{k_{\F}}{2}}^{k_\F-\frac{|k|}{2}-\sqrt{2}} \frac{|k|^{-1}+1}{k_\F^{1/2}\left( k_\F - t +\tfrac{|k|}{2} \right)^{1/2}}\di t\\
    &\leq C + C|k|^{-1}\log(k_\F)\;.
\end{align*}
In case $k_{\F} < |k| < 2 k_{\F}$, we have $a_+ - b_+ \geq k_{\F}$ and estimate
\begin{equation}
    \sqrt{a_+^2-b_+^2} - \sqrt{a_-^2-b_-^2}
    \leq \frac{C(k_{\F} + |k| t )}{k_{\F}^{\frac 12} (t + \sqrt{2})^{\frac 12}}
    \leq C k_{\F}^{\frac 12} t^{\frac 12}\;.
\end{equation}
Therefore,
\begin{equation}
\begin{aligned}
    I_1
    \leq C \int_{\frac{\sqrt{2}}{2}}^{k_{\F}} k_{\F}^{-\frac 12} t^{-\frac 12} \di t
    \leq C \;.
\end{aligned}
\end{equation}
We next estimate $I_2$. Recall the constraint~\eqref{eq:Lk_Bulk} $ z = \hat{k} \cdot p > \frac{|k|}{2} + \frac 32 \sqrt{2} \Rightarrow t \ge \sqrt{2} $, so integrals start from $ z_* \coloneq \max\{ k_{\F} - \frac{\sqrt{2}}{2}, \frac{|k|}{2} + \frac 32 \sqrt{2} \} $ and $ t_* \coloneq \max \{ k_{\F} - \frac{|k|}{2} - \sqrt{2}, \sqrt{2} \} $:
\begin{align*}
    I_2 &= \int_{z_*}^{k_\F + |k|+\frac{\sqrt{2}}{2}}2\left(\left(k_\F+\tfrac{\sqrt{2}}{2}\right)^2 - (z-|k|)^2\right)^{1/2} |k|^{-1}\left(z-\tfrac{|k|}{2}-\tfrac{\sqrt{2}}{2}\right)^{-1}\di z\\
    &\leq Ck_\F \int_{z_*}^{k_\F + |k|+\frac{\sqrt{2}}{2}} |k|^{-1}\left(z-\tfrac{|k|}{2}-\tfrac{\sqrt{2}}{2}\right)^{-1}\di z
    = Ck_\F \int_{t_*}^{k_\F + \frac{|k|}{2}} |k|^{-1}t^{-1}\di t \\
    &= Ck_\F|k|^{-1} \log \Bigg( \frac{k_{\F} + \tfrac{|k|}{2}}{t_*} \Bigg)\;.
\end{align*}
If $|k|\leq k_\F$, then, since $ t_* = k_{\F} - \frac{|k|}{2} - \sqrt{2} $ and $\log(1+x)\leq x$,
\begin{equation*}
    I_2 \leq Ck_\F |k|^{-1}\frac{\sqrt{2}+|k|}{k_\F -\tfrac{|k|}{2}-\sqrt{2}}\leq Ck_\F|k|^{-1}\frac{|k|}{k_\F}
    \le C \;.
\end{equation*}
In case $k_\F<|k|<2k_\F$, we have $ t_* \ge \sqrt{2} $ so
\begin{equation*}
    I_2\leq Ck_\F|k|^{-1}\log(1+Ck_\F)\leq C\log(k_\F)\;.
\end{equation*}
This concludes the analysis of $L_k^{\mathrm{Bulk}}$.\\
To treat $L_k^{\mathrm{Tip}}$, we decompose it into planes of distance $\ell= |k|^{-1}\gcd({k_1,k_2})$. The planes are given by
\begin{equation} \label{eq:luneslice}
    L_k^m
    \coloneq\big\{ p\in\lune ~\big\vert~ \hat{k}\cdot p = \ell m\big\} \quad \Rightarrow \quad 
    \lune^{\mathrm{Tip}} = \bigcup_{m=m_*}^{m^*}L_k^m\;,
\end{equation}
where $m_*$ and $m^*$ are given by
\begin{align}
    m_*\coloneq\inf \big\{ m\in\mathbb{N} ~\big\vert~ \tfrac{|k|}{2}<\ell m \big\}\;, \qquad
    m^*\coloneq \sup\left\{ m\in\mathbb{N} ~\middle|~ \ell m \le \tfrac{|k|}{2} + \tfrac 32 \sqrt{2} \right\}\;.
\end{align}
With $\lambda_{k,p} = \hbar^2 |k| (\ell m - \frac{|k|}{2})$, introducing $s_m\coloneq\ell m - \frac{|k|}{2}$, we get
\begin{align} \label{eq:Tipbound_preliminary}
    \hbar^2 \sum_{p\in L^{\textrm{Tip}}_k}\la
    = \sum_{m=m_*}^{m^*}|k|^{-1}\left(\ell m-\tfrac{|k|}{2}\right)^{-1}|L_k^m|
    = \sum_{m=m_*}^{m^*}|k|^{-1}s_m^{-1}|L_k^m| \;.
\end{align}
We now claim that
\begin{equation} \label{eq:Lkm_bound}
    |L_k^m|\leq C+C\ell\sqrt{k_\F s_m} \;.
\end{equation}

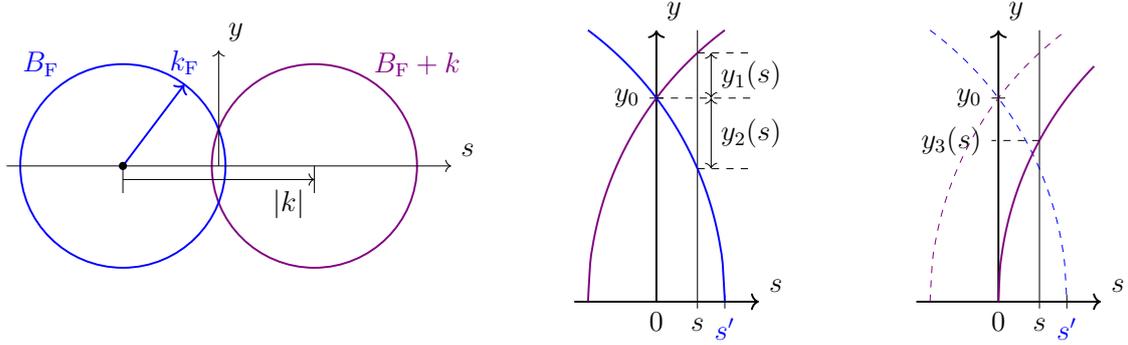
\begin{figure}
    \centering
    \scalebox{0.9}{\def\kF{5}			
\def\sprime{1}		
\def\s{0.6}			
\def\shift{5}		
\begin{tikzpicture}

\draw[red!50!blue,thick] (-5,2) circle (1.5);
\draw[blue,thick] (-7.8,2) circle (1.5);
\draw[->] (-9.5,2) -- (-3,2) node[anchor = south west] {$s$};
\draw[->] (-6.4,2) -- (-6.4,3.7) node[anchor = south west] {$y$};
\node[blue] at (-9,3.5) {$ B_{\F} $};
\node[red!50!blue] at (-3.5,3.5) {$ B_{\F} + k $};
\draw[->,thick, blue] (-7.8,2) -- ++(0.9,1.2) node[anchor = south] {$k_{\F}$};

\draw (-7.8,2) -- ++(0,-0.4);
\draw (-5,2) -- ++(0,-0.4);
\draw[->] (-7.8,1.8) -- (-5,1.8) node[anchor = north east] {$|k|$};
\fill (-7.8,2) circle (0.06);

\draw[thick, ->] (-\sprime-0.2,0) -- (\sprime+0.5,0) node[anchor = south west] {$s$};
\draw[thick, ->] (0,0) node[anchor = north] {$0$} -- ++ (0,4) node[anchor = south west] {$y$};
\draw ({\s},-0.1) node[anchor = north] {$s$} -- ++ (0,4.1);
\draw (\sprime,0.1) -- ++(0,-0.2) node[blue,anchor = north] {$s'$};
\draw (0.1,{sqrt(\kF^2-(\kF-\sprime)^2)}) -- ++(-0.2,0) node[anchor=east] {$y_0$};

\draw[thick,red!50!blue,smooth,samples=65,domain={-\sprime}:1] plot ({\x}, {sqrt(\kF^2-(\kF-\x-\sprime)^2)});
\draw[thick,blue,smooth,samples=65,domain={-\sprime}:\sprime] plot ({\x}, {sqrt(\kF^2-(\kF+\x-\sprime)^2)});

\draw[dashed] (0,{sqrt(\kF^2-(\kF-\sprime)^2)}) -- (\sprime+0.4,{sqrt(\kF^2-(\kF-\sprime)^2)});
\draw[dashed] (\s,{sqrt(\kF^2-(\kF-\sprime+\s)^2)}) -- (\sprime+0.4,{sqrt(\kF^2-(\kF-\sprime+\s)^2)});
\draw[dashed] (\s,{sqrt(\kF^2-(\kF-\sprime-\s)^2)}) -- (\sprime+0.4,{sqrt(\kF^2-(\kF-\sprime-\s)^2)});

\draw[<->] (\s+0.2,{sqrt(\kF^2-(\kF-\sprime)^2)}) -- node[anchor = west] {$y_2(s)$} (\s+0.2,{sqrt(\kF^2-(\kF-\sprime+\s)^2)});
\draw[<->] (\s+0.2,{sqrt(\kF^2-(\kF-\sprime)^2)}) -- node[anchor = west] {$y_1(s)$} (\s+0.2,{sqrt(\kF^2-(\kF-\sprime-\s)^2)});

\draw[thick, ->] (-\sprime-0.2+\shift,0) -- (\sprime+0.5+\shift,0) node[anchor = south west] {$s$};
\draw[thick, ->] (+\shift,0) node[anchor = north] {$0$} -- ++ (0,4) node[anchor = south west] {$y$};
\draw ({\s+\shift},-0.1) node[anchor = north] {$s$} -- ++ (0,4.1);
\draw (\sprime+\shift,0.1) -- ++(0,-0.2) node[blue,anchor = north] {$s'$};
\draw (0.1+\shift,{sqrt(\kF^2-(\kF-\sprime)^2)}) -- ++(-0.2,0) node[anchor=east] {$y_0$};

\draw[dashed,red!50!blue,smooth,samples=65,domain={-\sprime+\shift}:{\sprime+\shift}] plot ({\x}, {sqrt(\kF^2-(\kF-\x+\shift-\sprime)^2)});
\draw[dashed,blue,smooth,samples=65,domain={-\sprime+\shift}:{\sprime+\shift}] plot ({\x}, {sqrt(\kF^2-(\kF+\x-\shift-\sprime)^2)});
\draw[thick,red!50!blue,smooth,samples=65,domain={\shift}:{\sprime+\shift+0.4}] plot ({\x}, {sqrt(\kF^2-(\kF-\x+\shift)^2)});

\draw[dashed] (\s+\shift,{sqrt(\kF^2-(\kF-\s)^2)}) -- (-0.1+\shift,{sqrt(\kF^2-(\kF-\s)^2)}) node[anchor = east] {$ y_3(s) $};

\end{tikzpicture}}
    \caption{\textbf{Left}: For $ |k| \approx 2 k_{\F} $, the lune $L_k$ is almost identical to the shifted Fermi ball $B_{\F}+k$; only a small cap around $s=0$ is cut away.
    \textbf{Right}: The intersection length of $B_{\F}^c \cap (B_{\F} + k)$ with a vertical line at fixed $s< s'$ is given by $y_1(s)+y_2(s)$, which starts off at 0 for $s=0$ and then grows rapidly. It therefore has to be estimated carefully, using the properties of $y_1(s)$, $y_2(s)$ and $y_3(s)$.}
    \label{fig:y_y1_y2_y3}
\end{figure}

For $|k| \le k_{\F}$, this is obvious as the tip has a fixed opening angle $ \le \tfrac{\pi}{3} $, so $|L_k^m| \le C + C s_m$ with $s_m \le \frac 32 \sqrt{2}$. For $|k| \ge k_{\F}$, this argument does not go through since the opening angle approaches $ \pi $ as $|k| \approx 2 k_{\F}$. To prove~\eqref{eq:Lkm_bound}, we then proceed as follows, see Figure~\ref{fig:y_y1_y2_y3}: We introduce the continuous planes
\begin{equation} \label{eq:luneslice}
    L(s)
    \coloneq\big\{ p\in\lune ~\big\vert~ \hat{k}\cdot p = s + \tfrac{|k|}{2} \big\} \cap B_{k_{\F}}(0)^c \cap (B_{k_{\F}}(0) + k) \;, \qquad 
    s \in \mathbb{R} \;,
\end{equation}
and denote by $2y(s)$ the volume of $L(s)$. Let $s' \coloneqq k_{\F} - \frac{|k|}{2}$ and $y_0 \coloneq \sqrt{k_{\F}^2-(k_{\F}-s')^2}$. Then, we can write $y(s) = y_1(s)+y_2(s)$ with
\begin{equation*}
    y_1(s) \coloneq \sqrt{k_{\F}^2 -(k_{\F}-s'-s)^2}-y_0 \;, \qquad
    y_2(s) \coloneq \begin{cases}
        y_0-\sqrt{k_{\F}^2 -(k_{\F}-s'+s)^2} \quad &\textnormal{for } s < s' \\
        y_0 \quad
        &\textnormal{for } s \ge s'
    \end{cases} \;.
\end{equation*}
We compare this with $y_3(s) \coloneq \sqrt{k_{\F}^2-(k_{\F}-s)^2} \le \sqrt{2 k_{\F} s}$: Initially, $y(0)=y_3(0)$. For $0 \le s < s'/2$, we have $\frac{\di}{\di s} y_1(s) \le \frac{\di}{\di s} y_3(s)$ and $\frac{\di}{\di s} y_2(s) \le \frac{\di}{\di s} y_3(s)$, so $y(s) \le 2 y_3(s)$. For $s'/2 \le s \le \frac 32 \sqrt{2} $ (which only happens if $s' \le 3 \sqrt{2}$), we have both $y_1(s) \sim \sqrt{k_{\F} s}$ and $y_3(s) \sim \sqrt{k_{\F} s}$, as well as $y_2(s) \le C y_3(s)$, whence $y(s) \le C y_3(s) \le C \sqrt{k_{\F} s}$. Recalling that the lattice spacing on the plane is $\ell^{-1}$ and that lattice discretization leads to an error $C$, we obtain $|L_k^m|\leq C+\ell y(s) \leq C+C\ell\sqrt{k_\F s_m}$, which proves~\eqref{eq:Lkm_bound}.\\
Also, by lattice discretization, $\hbar^{-2} \lambda_{k,p} = |k| s_m \ge \frac 12 \Leftrightarrow s_m^{-1} |k|^{-1} \le 2$.\\
Further, since $ \ell \le 1 $, we have $|L_k^m| \le C + C \sqrt{|k| s_m} \le C |k| s_m$, so the contribution from each plane is $ \le C $, and we can remove the first two planes in~\eqref{eq:Tipbound_preliminary}:
\begin{align*}
    \hbar^2 \sum_{p\in L^{\textrm{Tip}}_k}\la
    &\leq \frac{|L_{m_*}|}{|k| s_{m_*}}
        + \frac{|L_{m_*+1}|}{|k| s_{m_*+1}}
        + C\sum_{m=m_*+2}^{m^*}(|k|^{-1}s_m^{-1}+k_\F^{1/2}s_m^{-1/2}|k|^{-1}\ell)\\
    &\leq C + C|k|^{-1}\ell^{-1}\sum_{m=0}^{m^*-m_*-2}(m+2)^{-1}+Ck_\F^{1/2}|k|^{-1}\ell^{1/2}\sum_{m=0}^{m^*-m_*-2}(m+2)^{-1/2}\\
    &\leq C + C|k|^{-1}\ell^{-1} \int_0^{C \ell^{-1}}\frac{1}{m+1} \di m
        + Ck_\F^{1/2}|k|^{-1}\ell^{1/2} \int_{0}^{C \ell^{-1}}\frac{1}{\sqrt{m+1}} \di m \\
    &\leq C + C|k|^{-1}\ell^{-1}\log\left(1+\ell^{-1}\right) + Ck_\F^{1/2}|k|^{-1}\\
    &\leq C + C\log(1+|k|) + Ck_\F^{1/2}|k|^{-1}
    \leq C\log(k_\F) \;,
\end{align*}
where we used $|k|^{-1}\leq \ell \leq 1$ and $k_\F\leq |k| < 2k_\F$ in the last two lines.\\

For $|k|< k_\F$, with $\lambda_{k,p} \geq \tfrac 12 \hbar^2$, we estimate
\begin{equation*}
    \hbar^2 \sum_{p\in L^{\textrm{Tip}}_k}\la
    \leq 2 |S^{\mathrm{Tip}}_k| \;, \qquad
    S^{\mathrm{Tip}}_k \coloneq \bigcup_{p\in L^{\textrm{Tip}}_k} (p+[-\tfrac 12, \tfrac 12]^2) \;.
\end{equation*}
\begin{figure}
    \centering
    \scalebox{0.9}{\definecolor{lilla}{RGB}{200,160,255}

\def\r{2}
\def\R{5}
\def\k{0.8}
\def\K{3}
\def\d{3.5}

\def\shift{4}

\begin{tikzpicture}

\draw (-\d-\shift,0) -- (\d-\shift,0);
\draw[red!50!blue, thick, ->] ({2.4-\shift},2.2) -- ++({\k},0) node[anchor = west] {$k$};

\draw[thick, blue] ({-\k/2-\shift},0) circle (\r);
\node[blue] at ({-2-\shift},1.8) {$B_{\F}$};
\fill[black] ({-\k/2-\shift},0) circle (0.08) node[anchor = north east]{$0$};
\fill[black] ({\k/2-\shift},0) circle (0.08) node[anchor = north west]{$k$};
\draw[red!50!blue, thick] ({\k/2-\shift},0) circle (\r);

\draw[dashed, thin] (-\shift,-2.5) -- (-\shift,2.5);
\fill (-\shift,{-sqrt(\r^2-(\k/2)^2)}) circle (1.2pt);

\draw[thick,<-,blue] (-\shift,{-sqrt(\r^2-(\k/2)^2)}) --node[anchor = east] {$k_{\F}$} ({-\k/2-\shift},0);
\draw[thin] (-\shift, {-sqrt(\r^2-(\k/2)^2)/2}) node[anchor = west] {$\alpha$}
    arc[start angle=90, end angle={90 + atan( \k / (2*sqrt(\r^2-(\k/2)^2)) )}, radius={sqrt(\r^2-(\k/2)^2)/2}] ;

\draw[|-|, thin] ({-\k/2-\shift},0.08) -- (-\shift,0.08) node[midway, above] {$|k|/2$};

\pgfmathsetmacro{\yshift}{-sqrt(\r^2-(\k/2)^2)*0.6 + 0.5}

\draw[thin] (\shift,\yshift) -- (\shift+\R,\yshift);

\draw[thin] plot[domain = {\shift}:{\shift+\R}, samples = 200] ({\x}, {\yshift + (\x-\shift-0.5)*(\shift + 0.5 + (\K/2 - \shift - 0.5))/sqrt(\R^2-(\shift + 0.5 + (\K/2 - \shift - 0.5))^2)});
\draw[thin] plot[domain = {\shift}:{\shift+\R}, samples = 200] ({\x}, {\yshift - (\x-\shift-0.5)*(\shift + 0.5 + (\K/2 - \shift - 0.5))/sqrt(\R^2-(\shift + 0.5 + (\K/2 - \shift - 0.5))^2)});

\draw[thin] ({\shift+1.3},\yshift) arc[start angle = 0, end angle = 11, radius = 1.3] node[midway, right] {$\alpha$};

\draw[thick,blue] plot[domain={\shift}:{\shift + \R -\K/2}, samples=200] ({\x}, {\yshift + sqrt(\R^2-(\K/2)^2) - sqrt(\R^2 - (\x + \K/2 - \shift - 0.5)^2)});
\draw[thick, red!50!blue] plot[domain={\shift}:{\shift + \R}, samples=200] ({\x}, {\yshift + sqrt(\R^2-(\K/2)^2) - sqrt(\R^2 - (\x - \K/2 - \shift - 0.5)^2)});

\draw[dashed] ({\shift + 0.5},{\yshift - 1.5}) -- ++(0,4);
\draw[dashed] ({\shift + 0.7},{\yshift - 1.5}) -- ++(0,4);
\draw[<-] ({\shift + 0.5},{\yshift + 2}) -- ++(-0.3,0);
\draw[<-] ({\shift + 0.7},{\yshift + 2}) -- ++(0.3,0) node[anchor = west] {$\tfrac{3}{2} \sqrt{2}$};
\draw[<-] ({\shift + 0.6},{\yshift -0.1}) -- ++(-0.8,-0.8) node[anchor = east] {$L_k^{\mathrm{Tip}}$};

\end{tikzpicture}}
    \caption{\textbf{Left}: The angle $\alpha$ is given by $\sin(\alpha) = \frac{|k|}{2 k_{\F}}$. For $|k| < k_{\F}$, we have $\alpha \le \frac{\pi}{6}$.\\
    \textbf{Right}: The opening angle of the lune $L_k$ is $2 \alpha$ and for $|k| < k_{\F}$, the tip $L^{\textrm{Tip}}_k$ is linearly approximated by a triangle.}
    \label{fig:luneTip}
\end{figure}
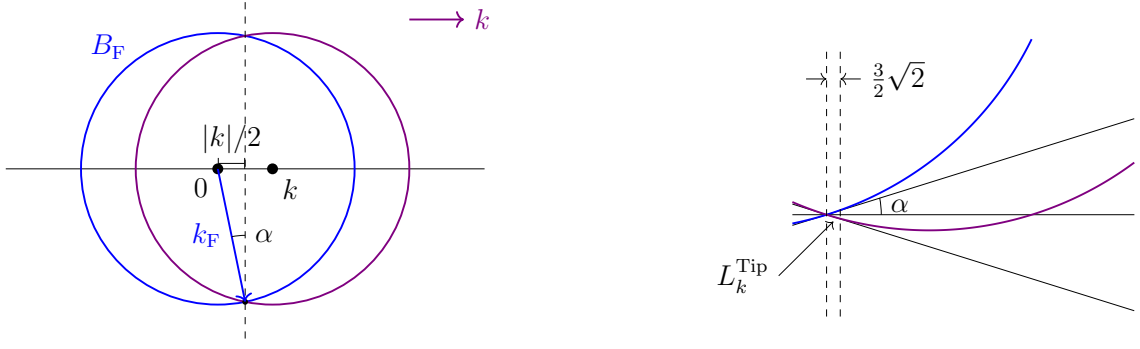
The extension of the tip in $k$-direction is $\tfrac 32 \sqrt{2} \ll k_{\F}$, so the extension of $S^{\mathrm{Tip}}_k$ in this direction is $\leq \tfrac 52 \sqrt{2}$. By a Taylor expansion of $y_1(s)$ and $y_2(s)$, we also estimate the extension of $S^{\mathrm{Tip}}_k$ perpendicular to $k$ and conclude
\begin{equation*}
    |S^{\mathrm{Tip}}_k|
    = \tan(\alpha) (\tfrac 52 \sqrt{2})^2 + \cO(k_{\F}^{-1}) \;,
\end{equation*}
where $\sin(\alpha) \coloneq \frac{|k|}{2} \frac{1}{k_{\F}} \le \frac 12$, see Figure~\ref{fig:luneTip}, so the lune has an opening angle of $2 \alpha$. Thus, $\alpha \le \frac{\pi}{6}$ and we conclude $|S^{\mathrm{Tip}}_k| \le C$. Summing up all bounds, the final result is
\begin{equation}
    \hbar^2 \sum_{p \in L_k}\la \leq C \log(k_\F) \leq C \log(N) \quad \text{ for } |k|<2k_\F\;.
\end{equation}

\textbf{Case $|k|\geq 2k_\F$}.
Here, $ \lambda_{k,p} $ may get small if $ \hat{k} \cdot p \approx k_{\F} - |k| $. We thus split
\begin{equation*}
    L_k = L^{\textrm{Cap}}_k\cup L_k^{\textrm{Rest}} \;, \qquad
    L^{\textrm{Cap}}_k \coloneq \left\{p\in L_k~\middle|~\hat{k}\cdot p -\tfrac{|k|}{2}\leq C' \right\} \;,
\end{equation*}
for some sufficiently large constant $C' > 0$. We estimate the sum over $L_k^{\textrm{Rest}}$ by a similar domination argument as for $ L_k^{\textrm{Bulk}} $ above: With $ z = \hat{k} \cdot p $, we have $ \hbar^{-2} \lambda_{k,p} = |k| \left( z - \frac{|k|}{2} \right) $, so
\begin{equation*}
    \hbar^2 \sum_{p\in L_k^{\textrm{Rest}}}\la \leq C\int_{|k|-k_\F +C'}^{|k|+k_\F} |k|^{-1}\left(z-\tfrac{|k|}{2}\right)^{-1}\sqrt{k_\F^2 -(z-|k|)^2} \; \di z\;.
\end{equation*}
Since the shifted Fermi ball $B_{k_{\F}}(k)$ is enclosed by the parabola
\begin{equation*}
    \left\{ p \in \RRR^2 ~\middle|~ (\hat{k}^\perp \cdot p)^2 = |k| \big( \hat{k} \cdot p - \tfrac{|k|}{2} \big) \right\} \;, \qquad 
    \hat{k}^\perp \coloneq \left( \begin{smallmatrix} 0 & -1 \\ 1 & 0 \end{smallmatrix} \right) \hat{k} \;,
\end{equation*}
we have $\sqrt{k_\F^2-(z-|k|)^2}\leq \sqrt{|k|}\sqrt{z-\tfrac{|k|}{2}}$, so
\begin{align*}
    \hbar^2 \sum_{p\in L_k^{\textrm{Rest}}}\la
    \leq C\int_{|k|-k_\F +C'}^{|k|+k_\F} |k|^{-\frac 12}\left(z-\tfrac{|k|}{2}\right)^{-\frac 12}\di z
    \leq Ck_\F^{-\frac 12}\left(\sqrt{\tfrac{|k|}{2}+k_\F} - \sqrt{\tfrac{|k|}{2}-k_\F + C'}\right)
    \le C \;,
\end{align*}
where we used $ \sqrt{a+b} - \sqrt{a} \le \sqrt{b} $.\\
For what concerns $L^{\textrm{Cap}}_k$, we proceed by decomposing the set into planes, as we did for the lune's tips. The lowest and highest $ m $ contributing to the cap are
\begin{align}
    M_*\coloneq\inf\left\{ m\in\mathbb{N} ~\middle|~ |k| - k_{\F} <\ell m \right\}\;, \qquad
    M^*\coloneq\sup\left\{ m\in\mathbb{N} ~\middle|~ \ell m\leq  \tfrac{|k|}{2} + C' \right\}\;.
\end{align}
Note that $M^* \leq M_* + C'$. Moreover, for $|k| > 2 k_{\F} + 2 C'$, we have $M^* < M_*$ and the cap is empty, so we can restrict to $|k| \leq 2 k_{\F} + 2 C'$. The bounds are now analogous to the tip: We still have $|L_k^m| \leq C + C \ell \sqrt{k_{\F} s_m}$ with $\lambda_{k,p} = \hbar^2 |k| s_m$ and $s_m = \ell m - \tfrac{|k|}{2}$ on the $m$-th plane. In particular, the contribution of each plane is $ \le C $, so we can remove the first two planes:
\begin{align*}
    \hbar^2 \sum_{p\in L_k^{\textrm{Cap}}}\la 
    &= \frac{|L_{M_*}|}{|k| s_{M_*}}
        + \frac{|L_{M_*+1}|}{|k| s_{M_*+1}}
        + C \sum_{m=M_*+2}^{M^*} \frac{1 + \ell \sqrt{k_{\F} s_m}}{|k| s_m} \\
    &\leq C + C|k|^{-1}\ell^{-1}\sum_{m=0}^{M^*-M_*-2}(m+2)^{-1} + Ck_\F^{\frac 12}|k|^{-1}\ell^{\frac 12}\sum_{m=0}^{M^*-M_*-2}(m+2)^{-\frac 12}\\
    &\leq C + C|k|^{-1}\ell^{-1}\int_0^{C' \ell^{-1}}\frac{1}{m+1} \di m+Ck_\F^{\frac 12}|k|^{-1}\ell^{\frac 12}\int_0^{C' \ell^{-1}}\frac{1}{\sqrt{m+1}} \di m\\
    &\leq C + C|k|^{-1}\ell^{-1}\log\left(1+|k|\right)+Ck_\F^{\frac 12}|k|^{-1}
    \leq C\log(k_\F)
    \leq C \log(N) \;,
\end{align*}
where we used $\ell^{-1} \leq |k| \leq C k_{\F}$. Hence, we proved~\eqref{eq:lambda-bound} for any $k\in\mathbb{Z}^2$.
\end{proof}

\subsection{Inverse Energy Sum over an Annulus}

The following number theoretical result is well-known~\cite[Thm.~338]{NumberTheory}:

\begin{lemma}[Bound on points on a sphere] \label{lem:sphere_point_estimate}
For any $\varepsilon>0$, there exists a $ C_\varepsilon > 0 $ such that for all $n\in\mathbb{N}$,
\begin{equation} \label{eq:sphere_point_estimate}
    r_2(n) \coloneq |\{ p \in \ZZZ^2 ~|~ |p|^2 = n^2 \}|
    \leq C_\varepsilon n^{\varepsilon} \;.
\end{equation}
\end{lemma}

We use this result to prove the following analog of~\cite[Lemma~3.2]{christiansen2024correlation}. Recall that $ e(p) = \hbar^2 ||p|^2 - k_{\F}^2| $, where we chose w.l.o.g. $ k_{\F} = \frac{1}{2}\left(\inf_{p\in B_F^c}|p|^2 - \sup_{q\in B_F}|q|^2\right)$.

\begin{lemma} \label{lem:spheres_estimate}
Given $\varepsilon>0$, there exists some $C_\varepsilon > 0$ such that for any $A\subset\mathbb{Z}^2$ with $|A|\leq |\overline{B_{2k_{\F}}(0)}\cap\mathbb{Z}^2|$, we have
\begin{equation}
    \hbar^2 \sum_{p\in A} e(p)^{-1}
    \leq C_\varepsilon N^\varepsilon \;.
\end{equation}
\end{lemma}

Note that in the convention of~\cite{christiansen2024correlation}, this bound is of order $k_{\F}^{\varepsilon}$, which is smaller by a factor of $ k_{\F} $, as compared to the 3d case.

\begin{proof}
First, note that
\begin{equation}
    \hbar^2 \sup_{|p| > 2 k_{\F}} e(p)^{-1}
    = \sup_{|p| > 2 k_{\F}} ||p|^2 - k_{\F}^2|^{-1}
    < \frac 13 k_{\F}^{-2}
    = \hbar^2 \inf_{|p| \le 2 k_{\F}} e(p)^{-1} \;.
\end{equation}
Hence we can restrict our attention to $ A = \overline{B_{2k_{\F}}(0)}\cap\mathbb{Z}^2 $ by reordering. If we call $m\coloneqq \sup_{q\in B_F}|q|^2$ and $m'\coloneqq \inf_{p\in B_F^c}|p|^2$, then we have the decomposition into spheres
\begin{equation} \label{eq:spheredecomposition}
    \sum_{p\in \overline{B_{2k_{\F}}(0)}\cap\mathbb{Z}^2} ||p|^2 - k_{\F}^2|^{-1}
    =\sum_{n=1}^{\lfloor4k_{\F}^2\rfloor}\frac{r_2(n)}{|n-k_{\F}^2|} 
    = \sum_{n=1}^m \frac{r_2(n)}{k_{\F}^2 - n} + \sum_{n=m'}^{\lfloor4k_{\F}^2\rfloor}\frac{r_2(n)}{n - k_{\F}^2} 
    =: S^1 + S^2 \;.
\end{equation}
Here, by definition of $ k_{\F} $, we have $ |n-k_{\F}^2| \ge \frac 12 $. Thus, with Lemma~\ref{lem:sphere_point_estimate} and $ N \sim k_{\F}^2 $,
\begin{align*}
    S^1 &\leq C_\varepsilon\sum_{n=1}^m \frac{n^\varepsilon}{k_{\F}^2-n}
    \leq C_\varepsilon k_{\F}^\varepsilon
    \Big( 2 + \sum_{n=1}^{m-1}\frac{1}{k_{\F}^2-n} \Big)
    \leq C_\varepsilon k_{\F}^\varepsilon
        \Big( 2 + \int_2^m\frac{\di t}{k_{\F}^2-t} \Big)
    \leq C_\varepsilon k_{\F}^\varepsilon\log \Big( \frac{k_{\F}^2}{k_{\F}^2-m} \Big)\\
    &\leq C_\varepsilon N^\varepsilon \log(N)
    \leq C_\varepsilon N^\varepsilon \;.
\end{align*}
Likewise, for $S^2$, we have
\begin{equation*}
    S^2 \leq C_\varepsilon k_{\F}^{\varepsilon} \Big( 2 + \sum_{n=m'+1}^{\lfloor 4k_{\F}^2\rfloor}\frac{1}{n-k_{\F}^2} \Big)
    \leq C_\varepsilon k_{\F}^\varepsilon \Big( 2+ \int_{m'}^{4k_{\F}^2-1}\frac{\di t}{t-k_{\F}^2} \Big)
    \leq C_\varepsilon N^\varepsilon \log(N) 
    \leq C_\varepsilon N^\varepsilon \;.
\end{equation*}
Plugging this into~\eqref{eq:spheredecomposition} renders the desired result.
\end{proof}

\subsection{Points in Annulus Intersections}

The following estimate is a key ingredient for the bound of $\cE_1$ in Lemma~\ref{lem:cE_1estimate}.

\begin{lemma}[Points in Annulus Intersections] \label{lem:annulus_intersection}
For any $N = |B_{k_{\F}}(0)|$, consider the annulus $ \mathcal{A} \coloneq \{ p \in \mathbb{R}^2 ~|~ k_{\F} \le |p| < k_{\F} + \Delta \} $, where the thickness $ \Delta > 0 $ satisfies $ c N^{-\alpha} < \Delta < C N^{-\alpha} $ for some $\alpha \in (0,\frac 12)$. Then, for $k \in \mathbb{Z}^2 \setminus \{0\}$, we have
\begin{equation}
    |\mathcal{A} \cap (\mathcal{A}+k) \cap \mathbb{Z}^2|
    \leq C (N^{\frac 34 - \frac 52 \alpha} + N^{\frac 14 - \frac 12 \alpha}) \;.
\end{equation}
\end{lemma}

\begin{proof}
For $N$ large enough, we have $ |k| \ge 1 > \Delta $, so if $|k| < 2 k_{\F} $, then $\mathcal{A} \cap (\mathcal{A}+k)$ consists of two areas, each bounded by four arcs spanned between four points, see Figure~\ref{fig:annuli_intersection}. We consider one of those areas and call the points $P_1$, $P_2$, $P_3$ and $P_4$, characterized by
\begin{equation}
\begin{aligned}
    |P_1| &= |P_1-k| = k_{\F} + \Delta \;, \qquad
    &|P_2| &= k_{\F} \;, \quad |P_2-k| = k_{\F} + \Delta \;, \\
    |P_3| &= k_{\F} + \Delta \;, \quad |P_3-k| = k_{\F} \;, \qquad
    &|P_4| &= |P_4-k| = k_{\F} \;.
\end{aligned}
\end{equation}
The set $\mathcal{A} \cap (\mathcal{A}+k) \cap \mathbb{Z}^2$ is then decomposed into several planes that run either parallel or orthogonal to $k$, where we bound the number of planes and points per plane. For this, we distinguish five cases: Let $k_{\mathrm{crit}}$ be the value of $|k|$ for which $|P_2 - P_3| = |k| = k_{\mathrm{crit}}$, that is,
\begin{equation}
    (k_{\F} + \Delta)^2
    = k_{\F}^2 + k_{\mathrm{crit}}^2 
    \qquad \Rightarrow \qquad
    k_{\mathrm{crit}}
    = \sqrt{2 k_{\F} \Delta + \Delta^2}
    \sim N^{\frac 14 - \frac{\alpha}{2}} \;.
\end{equation}
So $P_2$ is right above 0, and $P_3$ is right above $k$.\\

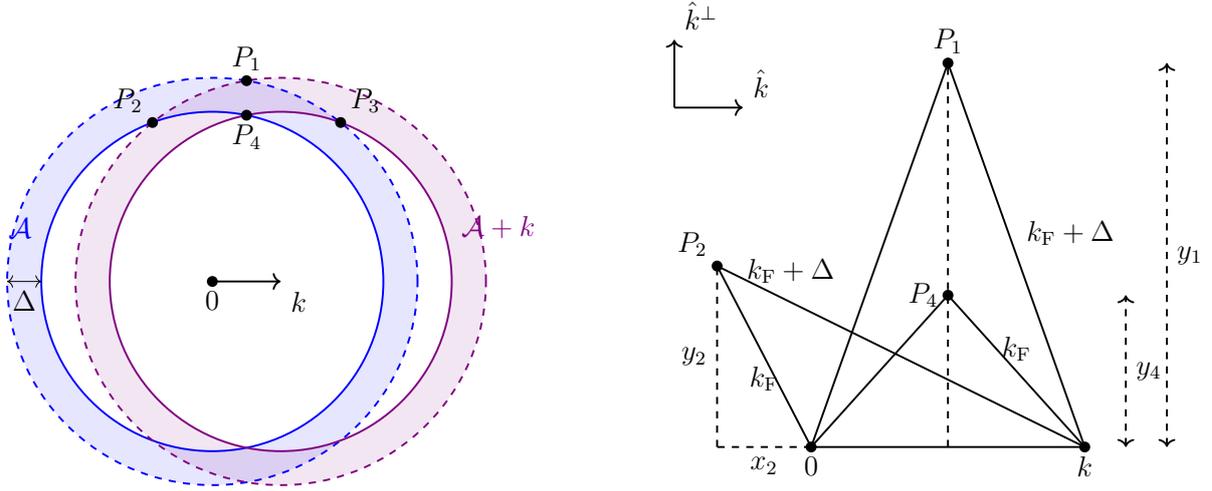
\begin{figure}
    \centering
    \scalebox{0.9}{\def\kF{2.5}    
\def\delta{0.5} 
\def\k{1}       
\begin{tikzpicture}

\fill[blue, opacity = .1] (0,{\kF + \delta}) arc[start angle=90, end angle=450, radius={\kF + \delta}] -- (0,{\kF}) arc[start angle=90, end angle=-270, radius={\kF}];
\fill[blue!50!red, opacity = .1] ({\k},{\kF + \delta}) arc[start angle=90, end angle=450, radius={\kF + \delta}] -- ({\k},{\kF}) arc[start angle=90, end angle=-270, radius={\kF}];
\draw[blue,thick] (0,0) circle ({\kF}) ;
\draw[blue,thick, dashed] (0,0) circle ({\kF + \delta}) ;
\draw[blue!50!red,thick] (\k,0) circle ({\kF}) ;
\draw[blue!50!red,thick, dashed] (\k,0) circle ({\kF + \delta}) ;

\fill (0,0) circle (0.08) node[anchor = north]{$0$};
\draw[thick,->] (0,0) -- ++(\k,0) node[anchor = north west]{$k$};

\coordinate (P1) at (
    {0.5*\k},
    { sqrt((\kF + \delta)^2 - (0.5*\k)^2)} );
\coordinate (P2) at (
    {(\k^2 - \delta^2 - 2*\kF*\delta)/(2*\k)}
    ,{sqrt(\kF^2 - (\k^2 - \delta^2 - 2*\kF*\delta)^2/(2*\k)^2) } );
\coordinate (P3) at (
    {\k - (\k^2 - \delta^2 - 2*\kF*\delta)/(2*\k)}
    ,{sqrt(\kF^2 - (\k^2 - \delta^2 - 2*\kF*\delta)^2/(2*\k)^2) } );
\coordinate (P4) at (
    {0.5*\k},
    { sqrt((\kF)^2 - (0.5*\k)^2)} );
\fill (P1) circle (0.08) node[anchor = south]{$P_1$};
\fill (P2) circle (0.08) node[anchor = south east]{$P_2$};
\fill (P3) circle (0.08) node[anchor = south west]{$P_3$};
\fill (P4) circle (0.08) node[anchor = north]{$P_4$};

\node[blue, anchor = south east] at ({-\kF},0.5) {$\mathcal{A}$};
\node[blue!50!red, anchor = south west] at ({\k+\kF},0.5) {$\mathcal{A} + k$};
\draw[<->] ({-\kF},0) --node[anchor = north] {$\Delta$} ++({-\delta},0); 

\end{tikzpicture}}
    \hspace{2em}
    \scalebox{0.9}{\def\kF{3}    
\def\delta{3} 
\def\k{4}       
\begin{tikzpicture}

\draw[thick,->] (-2,5) -- ++(1,0) node[anchor = south west]{$\hat{k}$};
\draw[thick,->] (-2,5) -- ++(0,1) node[anchor = south west]{$\hat{k}^\perp$};

\fill (0,0) circle (0.08) node[anchor = north]{$0$};
\fill ({\k},0) circle (0.08) node[anchor = north]{$k$};
\draw[thick] (0,0) -- ++(\k,0);

\coordinate (P1) at (
    {0.5*\k},
    { sqrt((\kF + \delta)^2 - (0.5*\k)^2)} );
\coordinate (P2) at (
    {(\k^2 - \delta^2 - 2*\kF*\delta)/(2*\k)}
    ,{sqrt(\kF^2 - (\k^2 - \delta^2 - 2*\kF*\delta)^2/(2*\k)^2) } );
\coordinate (P3) at (
    {\k - (\k^2 - \delta^2 - 2*\kF*\delta)/(2*\k)}
    ,{sqrt(\kF^2 - (\k^2 - \delta^2 - 2*\kF*\delta)^2/(2*\k)^2) } );
\coordinate (P4) at (
    {0.5*\k},
    { sqrt((\kF)^2 - (0.5*\k)^2)} );
\coordinate (A2) at (
    {(\k^2 - \delta^2 - 2*\kF*\delta)/(2*\k)}
    ,0 );

\fill (P1) circle (0.08) node[anchor = south]{$P_1$};
\fill (P2) circle (0.08) node[anchor = south east]{$P_2$};
\fill (P4) circle (0.08) node[anchor = east]{$P_4$};

\draw[thick] (0,0) -- (P1);
\draw[thick] ({\k},0) --node[anchor = south west]{$k_{\F} + \Delta$} (P1);
\draw[thick, dashed] ({0.5*\k},0) -- (P1);
\draw[thick] (0,0) -- (P4);
\draw[thick] ({\k},0) --node[anchor = south]{$k_{\F}$} (P4);

\draw[thick, dashed] (0,0) --node[anchor = north]{$x_2$} (A2);
\draw[thick, dashed] (A2) --node[anchor = east]{$y_2$} (P2);
\draw[thick] (0,0) --node[anchor = north]{$k_{\F}$} (P2);
\draw[thick] ({\k},0) -- (P2);

\node at (-0.3,2.6) {$k_{\F} + \Delta$};
\draw[thick, dashed, <->] ({\k+0.6},0) --node[anchor = west]{$y_4$} ++(0,{ sqrt((\kF)^2 - (0.5*\k)^2)});
\draw[thick, dashed, <->] ({\k+1.2},0) --node[anchor = west]{$y_1$} ++(0,{ sqrt((\kF + \delta)^2 - (0.5*\k)^2)});

\end{tikzpicture}}
    \caption{\textbf{Left}: For $1 \le |k| < 2 k_{\F} $, the intersection of the two annuli $ \mathcal{A} \cap (\mathcal{A}+k) $ amounts to two areas, each bordered by four arcs between four points $P_1$, $P_2$, $P_3$ and $P_4$.\\
    \textbf{Right}: The coordinates $x_j$ and $y_j $ are defined by putting the intersection points $P_j$ into the coordinate system spanned by $\hat{k} $ and $\hat{k}^\perp$}
    \label{fig:annuli_intersection}
\end{figure}

\textbf{Case 1: $1 \le |k| \le k_{\mathrm{crit}}$.}
We divide $\mathcal{A} \cap (\mathcal{A}+k) \cap \mathbb{Z}^2$ into planes $\mathcal{A}_m$, parallel to $\hat{k} = k/|k|$, and hence orthogonal to $\hat{k}^\perp \coloneq \left( \begin{smallmatrix} 0 & -1 \\ 1 & 0 \end{smallmatrix} \right) \hat{k}$, with the separation between two planes being $\ell=|k|^{-1} \gcd(k_1,k_2) \leq1$:
\begin{equation} \label{eq:cAm}
    \mathcal{A}_m \coloneq \{ p \in \mathcal{A} \cap (\mathcal{A}+k) \cap \mathbb{Z}^2 ~|~ p \cdot \hat{k}^\perp = m \ell \} \;, \qquad
    m \in \mathbb{Z} \;.
\end{equation}
Let the coordinates of $P_j$ in the system spanned by $(\hat{k}, \hat{k}^\perp)$ be $x_j \coloneq P_j \cdot \hat{k}$ and $y_j \coloneq P_j \cdot \hat{k}^\perp $, see Figure~\ref{fig:annuli_intersection}. Without loss of generality, let $y_j > 0$. Then, $\mathcal{A}_m$ can only be non-empty if
\begin{equation}
    m_* \le |m| \le m^* \;, \qquad
    m_* \coloneq \inf\{ m \in \mathbb{N} ~|~ m \ell \ge \min(y_2,y_4) \} \;, \quad
    m^* \coloneq \sup\{ m \in \mathbb{N} ~|~ m \ell \le y_1\} \;.
\end{equation}
The number of non-empty planes is thus bounded by
\begin{equation} \label{eq:numberofplanes_bound}
    2(m^* - m_* + 1)
    \le \tfrac{2}{\ell} (1 + |y_1 - y_2| + |y_1 - y_4|) \;.
\end{equation}
By the Pythagorean theorem, we conclude
\begin{equation} \label{eq:y1y4_bound}
    y_4^2 = k_{\F}^2 - \frac{|k|^2}{4} \;, \quad
    y_1^2 = (k_{\F} + \Delta)^2 - \frac{|k|^2}{4} \quad \Rightarrow \quad
    |y_1 - y_4| = \frac{2 k_{\F} \Delta + \Delta^2}{y_1 + y_4}
    \leq C \Delta
    \leq C N^{-\alpha} \;,
\end{equation}
where we used in the last two steps that $|k| \ll k_{\F}$ implies $(y_1 + y_4) \ge C k_{\F}$. Moreover, 
\begin{equation*}
    y_2^2 = k_{\F}^2 - x_2^2
    = (k_{\F} + \Delta)^2 - (|x_2| + |k|)^2 \quad \Rightarrow \quad
    |x_2| = \frac{2 k_{\F} \Delta + \Delta ^2 - |k|^2}{2 |k|}
    \leq C N^{\frac 12 - \alpha} |k|^{-1} \;,
\end{equation*}
so in particular $|x_2| \ll k_{\F}$, hence $y_2 \ge c k_{\F}$, and we have
\begin{equation}
\begin{aligned}
    (k_{\F} + \Delta)^2 - \frac{|k|^2}{4}
    &= y_1^2
    = (y_2 + (y_1-y_2))^2
    \geq 2 y_2 (y_1 - y_2) + y_2^2 \\
    \Rightarrow \quad
    |y_1 - y_2| 
    &\leq \frac{(k_{\F} + \Delta)^2 - \frac{|k|^2}{4} - y_2^2}{2 y_2}
    = \frac{2 k_{\F} \Delta + \Delta^2 - \frac{|k|^2}{4} + x_2^2}{2 y_2}
    \leq C (N^{-\alpha} + N^{\frac 12 - 2 \alpha} |k|^{-2}) \;.
\end{aligned}
\end{equation}
With~\eqref{eq:numberofplanes_bound}, this bounds the number of planes by 
\begin{equation} \label{eq:numberofplanes_bound2}
    2(m^* - m_* + 1)
    \leq C \ell^{-1}  (N^{\frac 12 - 2 \alpha} |k|^{-1} + 1 + N^{-\alpha}) \;.
\end{equation}
The maximum number of points that can fit on a plane is bounded by (see Figure~\ref{fig:plane_maxpoint})
\begin{equation} \label{eq:pointsinplanes_bound}
    |\mathcal{A}_m|
    \leq 2 \ell \sqrt{(k_{\F} + \Delta)^2 - k_{\F}^2} + 1
    \leq C \ell \sqrt{k_{\F} \Delta} + C
    \leq C \ell N^{\frac 14 - \frac{\alpha}{2}} \;,
\end{equation}
where we used $\ell \sqrt{k_{\F} \Delta} \ge c |k|^{-1} N^{\frac 14 - \frac{\alpha}{2}} \ge c |k|^{-1} k_{\mathrm{crit}} \ge c$. With $|k| \ge 1$, the final bound is then
\begin{equation} \label{eq:final_bound_annulus_intersection}
    |\mathcal{A} \cap (\mathcal{A}+k) \cap \mathbb{Z}^2|
    = \sum_{m_* \le |m| \le m^*} |\mathcal{A}_m|
    \leq C (N^{\frac 12 - 2 \alpha} |k|^{-1} + 1)
        N^{\frac 14 - \frac{\alpha}{2}}
    \leq C (N^{\frac 34 - \frac 52 \alpha} + N^{\frac 14 - \frac 12 \alpha}) \;.
\end{equation}

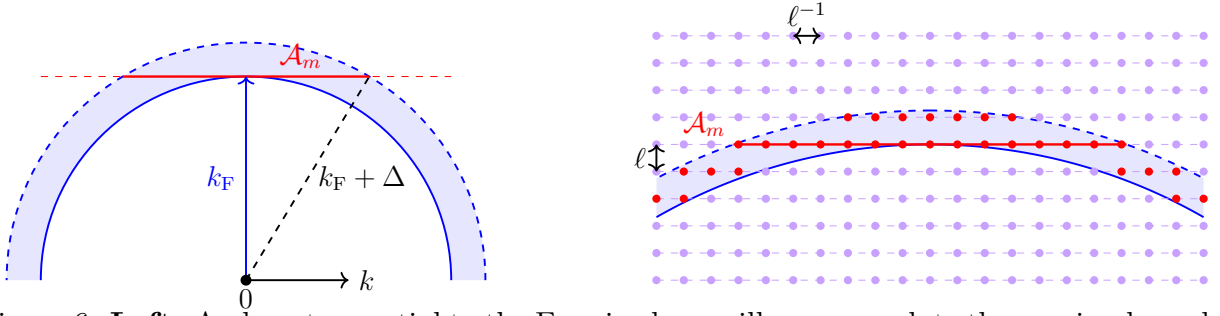
\begin{figure}
    \centering
    \scalebox{0.9}{\definecolor{lilla}{RGB}{200,160,255}
\usetikzlibrary{decorations.markings}

\begin{tikzpicture}
\useasboundingbox (-12,-2) rectangle (5,2);

\begin{scope}[xshift = -10cm, yshift = -2cm]
\def\kF{3}			
\def\delta{0.5}     
\def\k{1.5}         

\draw[thick,blue] ({\kF},0) arc[start angle=0, end angle=180, radius={\kF}];
\draw[thick,blue, dashed] ({\kF + \delta},0) arc[start angle=0, end angle=180, radius={\kF + \delta}];
\fill[blue, opacity = 0.1] ({\kF + \delta},0) 
    arc[start angle=0, end angle=180, radius={\kF + \delta}]
    -- ({-\kF},0)
    arc[start angle=180, end angle=0, radius={\kF}];

\coordinate (X1) at ({sqrt((\kF + \delta)^2 - (\kF)^2)} ,{\kF});
\coordinate (X2) at ({-sqrt((\kF + \delta)^2 - (\kF)^2)} ,{\kF});

\draw[dashed, red] ({-\kF} ,{\kF}) -- ({\kF} ,{\kF}) ;
\draw[thick, dashed] (0,0) --node[anchor = west]{$k_{\F} + \Delta$} (X1);
\draw[line width = 1, red] (X1) -- (X2);
\node[red] at (0.8,3.3) {$\mathcal{A}_m$};
\draw[->,thick, blue] (0,0) --node[anchor = east]{$k_{\F}$} ++(0,{\kF});

\fill (0,0) circle (0.08) node[anchor = north]{$0$};
\draw[thick,->] (0,0) -- ({\k},0) node[anchor = west] {$k$};
\end{scope}

\begin{scope}[xshift = 0cm]
\def\kF{8}			
\def\delta{0.5}     
\def\xlim{4}        
\def\l{0.4}         

\fill[blue, opacity=0.1] (0,0) --
    plot[domain=0:\xlim, samples=50] ({\x}, {sqrt(\kF^2 - \x^2) - \kF})
    -- ({\xlim},{sqrt((\kF + \delta)^2 - \xlim^2) - \kF}) 
    -- plot[domain=\xlim:{-\xlim}, samples=100] ({\x}, {sqrt((\kF + \delta)^2 - \x*\x) - \kF})
    -- ({-\xlim},{sqrt(\kF^2 - \xlim^2) - \kF}) 
    -- plot[domain={-\xlim}:0, samples=50] ({\x}, {sqrt(\kF^2 - \x*\x) - \kF});
    
\draw[blue, thick] (0,0) -- plot[domain=0:\xlim, samples=50] ({\x}, {sqrt(\kF^2 - \x^2) - \kF});
\draw[blue, thick] (0,0) -- plot[domain=0:{-\xlim}, samples=50] ({\x}, {sqrt(\kF^2 - \x*\x) - \kF});
\draw[blue, thick, dashed] (0,{\delta}) -- plot[domain=0:\xlim, samples=50] ({\x}, {sqrt((\kF + \delta)^2 - \x^2) - \kF});
\draw[blue, thick, dashed] (0,{\delta}) -- plot[domain=0:{-\xlim}, samples=50] ({\x}, {sqrt((\kF + \delta)^2 - \x*\x) - \kF});


\pgfmathsetmacro{\kFSqr}{\kF*\kF} 
\pgfmathsetmacro{\kFdeltaSqr}{(\kF+\delta)*(\kF+\delta)} 

\foreach \i [evaluate=\i as \iy using \i*\l] in {-5,...,4} {
    \draw[dashed,lilla,thin] ({-\xlim}, {\iy}) -- ({\xlim}, {\iy});
    \foreach \j [evaluate=\j as \jy using \j*\l] in {-10,...,10} {
    \pgfmathsetmacro{\rSqr}{(\iy+\kF)*(\iy+\kF) + \jy*\jy}
    \ifdim \rSqr pt < \kFSqr pt
        \fill[lilla] (\jy,\iy) circle (0.06);
    \else
        \ifdim \rSqr pt < \kFdeltaSqr pt
            \fill[red] (\jy,\iy) circle (0.06);
        \else
            \fill[lilla] (\jy,\iy) circle (0.06);
        \fi
    \fi
    }
}

\coordinate (Y1) at ({sqrt((\kF + \delta)^2 - (\kF)^2)} ,0);
\coordinate (Y2) at ({-sqrt((\kF + \delta)^2 - (\kF)^2)} ,0);
\draw[line width = 1, red] (Y1) -- (Y2);

\node[red] at (-3.3,0.3) {$\mathcal{A}_m$};

\draw[thick,<->] (-2,1.6) -- node[anchor = south]{$\ell^{-1}$} ++({\l},0);
\draw[thick,<->] (-4,-0.4) -- node[anchor = east]{$\ell$} ++(0,{\l});

\end{scope}

\end{tikzpicture}}
    \caption{\textbf{Left}: A plane tangential to the Fermi sphere will accommodate the maximal number of points inside the annulus on a single plane.
    \textbf{Right}: Decomposition of $\mathbb{Z}^2$ into planes parallel to $k$. A situation is shown, where the number of points $|\mathcal{A}_m|$ becomes maximal. Note that the distance between two planes is $\ell$, while the spacing between two lattice points on a plane is $\ell^{-1}$.}
    \label{fig:plane_maxpoint}
\end{figure}

\textbf{Case 2: $k_{\mathrm{crit}} < |k| \le k_{\F}$.}
Here, the number of non-empty planes is bounded as $ 2(m^*-m_*+1) \leq \tfrac{2}{\ell}(1 + |y_1-y_4|) $, where the bound~\eqref{eq:y1y4_bound} on $|y_1 - y_4|$ remains valid. Thus, the number of planes is still bounded by~\eqref{eq:numberofplanes_bound2}. The estimate on the number of points per plane~\eqref{eq:pointsinplanes_bound} holds irrespective of $k$, so also~\eqref{eq:final_bound_annulus_intersection} remains valid.\\

\textbf{Case 3: $k_{\F} < |k| \le 2 k_{\F}$.}
Here, $y_1 \ll k_{\F}$ may occur, so~\eqref{eq:y1y4_bound} loses its validity. Instead, we decompose $\mathcal{A} \cap (\mathcal{A}+k) \cap \mathbb{Z}^2$ into planes orthogonal to $k$,
\begin{equation} \label{eq:cAm_tilde}
    \tilde{\mathcal{A}}_m \coloneq \{ p \in \mathcal{A} \cap (\mathcal{A}+k) \cap \mathbb{Z}^2 ~|~ p \cdot \hat{k} = m \ell \} \;, \qquad
    m \in \mathbb{Z} \;,
\end{equation}
which are only non-empty if
\begin{equation}
    \tilde{m}_* \le |m| \le \tilde{m}^* \;, \qquad
    \tilde{m}_* \coloneq \inf\{ m \in \mathbb{N} ~|~ m \ell \ge x_2 \} \;, \quad
    \tilde{m}^* \coloneq \sup\{ m \in \mathbb{N} ~|~ m \ell \le x_3\} \;.
\end{equation}
The number of planes is bounded as $2(\tilde{m}^* - \tilde{m}_* + 1) \leq \frac{2}{\ell}(1 + x_3 - x_2)$. From the Pythagorean theorem and $x_3+x_2 = |k| \ge k_{\F}$, we get
\begin{equation}
\begin{aligned}
    y_2^2 &= k_{\F}^2 - x_2^2 = (k_{\F} + \Delta)^2 - x_3^2
    \quad \Rightarrow \quad 
    (x_3-x_2)(x_3+x_2) = 2 k_{\F} \Delta + \Delta^2 \\
    \Rightarrow \quad
    (x_3-x_2) &\le 2 \Delta + \Delta^2 k_{\F}^{-1}\le C N^{-\alpha} \;.
\end{aligned}
\end{equation}
By the same argument as in~\eqref{eq:pointsinplanes_bound}, we conclude $|\tilde{\mathcal{A}}_m| \le C \ell N^{\frac 14 - \frac{\alpha}{2}}$, so
\begin{equation} \label{eq:final_bound_annulus_intersection_2}
    |\mathcal{A} \cap (\mathcal{A}+k) \cap \mathbb{Z}^2|
    = \sum_{\tilde{m}_* \le |m| \le \tilde{m}^*} |\tilde{\mathcal{A}}_m|
    \leq C (1 + N^{-\alpha})
        N^{\frac 14 - \frac{\alpha}{2}}
    \leq C N^{\frac 14 - \frac 12 \alpha} \;.
\end{equation}

\textbf{Case 4: $2 k_{\F} < |k| \le 2 k_{\F} + 2 \Delta$.}
Here, the intersection points $P_1$, $P_2$, and $P_3$ may cease to exist. Nevertheless, we can still decompose $\mathcal{A} \cap (\mathcal{A}+k) \cap \mathbb{Z}^2$ into planes $\tilde{\mathcal{A}}_m$. The extension of $\mathcal{A} \cap (\mathcal{A}+k)$ in $k$-direction is now bounded by $2 \Delta \le C N^{-\alpha}$, so there are $\le C \ell^{-1}(1+ C N^{-\alpha})$ many planes, which still satisfy $|\tilde{\mathcal{A}}_m| \le C \ell N^{\frac 14 - \frac{\alpha}{2}}$. Hence,~\eqref{eq:final_bound_annulus_intersection_2} remains valid.\\

\textbf{Case 5: $|k| > 2 k_{\F} + 2 \Delta$.}
This case is trivial, since $\mathcal{A} \cap (\mathcal{A}+k) = \emptyset$. 

\end{proof}

\medskip

\printbibliography

\end{document}